\documentclass[letterpaper]{article}
\usepackage[preprint]{aaai2027}
\nocopyright
\usepackage[hyphens]{url}  \usepackage{graphicx}    \usepackage{natbib}  \usepackage{caption} 
\usepackage{amsmath,amssymb,amsthm}
\usepackage{booktabs}
\usepackage{algorithm}
\usepackage{algpseudocode}

\newcommand{\codeavailabilityblock}{
\paragraph{Code.}
Code and independent verification tests are available at
\url{https://github.com/jiaxing-guo/safe-observation-capacity}; the repository also
provides configurations, local experiment entry points, and reporting tools.
}

\graphicspath{{figures/}}

\newtheorem{theorem}{Theorem}
\newtheorem{proposition}{Proposition}
\newtheorem{lemma}{Lemma}
\newtheorem{corollary}{Corollary}
\theoremstyle{definition}
\newtheorem{definition}{Definition}
\newtheorem{example}{Example}

\newcommand{\Xpoly}{\mathcal{X}}
\newcommand{\Ypoly}{\mathcal{Y}}
\newcommand{\Amat}{A}
\newcommand{\Obs}{\mathcal{O}}
\newcommand{\ystar}{y^\star}
\newcommand{\Wval}{W}
\newcommand{\vref}{v_{\mathrm{ref}}}

\newcommand{\rok}{\rho}                       \newcommand{\Sset}{\mathcal{S}(\rok)}          \newcommand{\Cpub}{C^{\mathrm{pub}}}           \newcommand{\Cid}{C^{\mathrm{id}}}              \newcommand{\Nres}{\mathcal{N}_{\mathrm{res}}} \newcommand{\Dpub}{D_{\mathrm{pub}}}           \newcommand{\De}{D_{\mathrm{e}}}               \newcommand{\kcap}{\kappa_{\rok}}              \newcommand{\lrate}{\lambda_{\rok}}            \newcommand{\Gactivepop}{G^{\infty}_{\mathrm{active}}}         \newcommand{\Goracle}{G_{\mathrm{oracle}}}

\ifdefined\AAAIReviewMain
  \newcommand{\suppappendix}[2]{Supplementary App.~#2}
  \newcommand{\suppapp}[2]{Supplementary App.~#2}
  \newcommand{\suppappendices}[4]{Supplementary Apps.~#2 and~#4}
\else
  \newcommand{\suppappendix}[2]{Appendix~\ref{#1}}
  \newcommand{\suppapp}[2]{App.~\ref{#1}}
  \newcommand{\suppappendices}[4]{Appendices~\ref{#1} and~\ref{#3}}
\fi

\title{Safe Observation Capacity for Opponent Exploitation under Showdown Censoring}
\author{
    Jiaxing Guo
}
\affiliations{
    Imperial College London\\
    \texttt{henry.guo23@imperial.ac.uk}
}

\begin{document}
\maketitle

\begin{abstract}
In poker-like games, folds hide private cards, so showdown data are missing not
at random: per-card estimates converge to behavior conditional on reveal, and
shrinking confidence sets can lose coverage.  A floor-safe probe carries a
line to showdown; sequence-form flow then recovers censored fold mass on
reveal-certified histories.  We price acquisition through \emph{safe observation
capacity}, the largest target reach attainable by a floor-safe plan at a given value
slack.  Its frontier is concave and piecewise linear, with initial slope given by the
floor's shadow price.  When that reach converts fully to reveal and parent flow is
non-bottleneck, matching local bounds make the hands required for
conditional-probability half-width $\varepsilon$ inversely proportional, up to
logarithms, to capacity, opponent continuation mass, and $\varepsilon^2$.
Safe Active De-censoring (SAD) combines public screening, an independent reveal batch,
and robust deployment; a max--min safe audit gives positive joint reveal rate to
every coordinate in a finite library-covered target set, including public-null
deviations.
With $10^6$ hands, SAD raises the river over-fold certified gain from $0.485$ to
$0.692$ and improves both gains over public-only collection on all three deviations
(Holm-adjusted paired $p\le0.012$), while selecting no control target.
On a fixed-board public twin, a disjoint audit--refit--deploy loop detects all $30$
simulation seeds and no control seed, certifying absolute value $0.655$ (95\%
confidence-interval half-width $0.008$).
Every floor-constrained probe and response passes a floor audit.
\end{abstract}

\section{Introduction}
\label{sec:intro}

Folds create a structural identification problem in poker-like games.  A fold ends the
hand before private cards are revealed, so passive showdown data cannot generally
separate type-specific fold behavior.  Weak hands may call more often and strong hands
fold more often while producing the same public action frequencies.  Two such
\emph{public twins} are indistinguishable under the collection policy yet can require
different responses; additional passive hands only estimate their shared aggregate more
precisely.

This ambiguity limits safe opponent exploitation.  An equilibrium blueprint supplies a
worst-case value guarantee, and safe refinement preserves that guarantee while adapting
to predictable play
\citep{zinkevich2007regret,bowling2015heads,moravcik2017deepstack,
brown2017safe}.  However, response optimization cannot recover an
unobserved private coordinate.  The usual per-card showdown estimator converges
to behavior conditional on reveal, not to the uncensored behavior; its confidence sets
can therefore lose coverage as the sample grows.

Sequence-form constraints can correct this bias when data come from a uniformly
random non-fold probe \citep{davis2019solving}.  We study the additional requirement
that collection itself preserve worst-case value.  A floor-safe probe carries a
selected line to showdown and exposes every non-fold continuation; sequence-form flow
then recovers censored fold mass without weakening the floor.

Identification leaves a quantitative question: how often can a safe policy produce
the needed observation?  At permitted value slack $\rok$, safe observation capacity
$\kcap(I)$ is the maximum controlled reach of target information set $I$ among
floor-safe plans.  If maximal reach can be made revealing, the informative-event rate
is capacity times residual opponent reach and continuation; otherwise it follows from
the full observation law.  The capacity frontier is concave and piecewise linear,
with initial slope given by the floor's shadow price.

Hands needed for conditional-action half-width $\varepsilon$ scale, up to logarithms,
inversely with this rate and $\varepsilon^2$; a local lower bound matches when parent
flow is learned no slower.

\begin{figure*}[t]
      \centering
      \includegraphics[width=\textwidth]{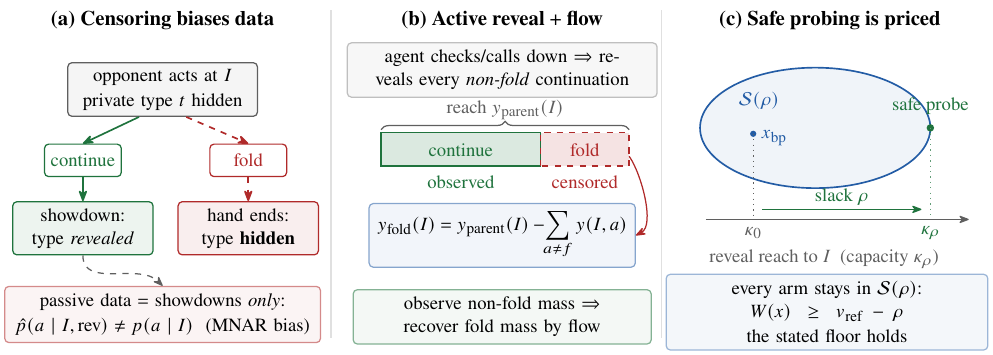}
      \caption{Safe Active De-censoring.
      (a)~A fold makes passive showdown samples missing not at random (MNAR).
      (b)~A reveal probe exposes non-fold continuation masses; the displayed flow
      identity recovers fold mass, with $f$ denoting fold.
      (c)~Every probe lies in
      $\Sset=\{x:\Wval(x)\ge\vref-\rok\}$, where $\Wval(x)$ is worst-case value,
      $\vref$ the blueprint guarantee, and $\rok$ permitted slack.  More slack can
      increase capacity above its zero-slack value $\kappa_0$.}
      \label{fig:mechanism}
\end{figure*}

Safe Active De-censoring (SAD) screens public anomalies, collects an independent
floor-safe reveal batch, and deploys a robust response.  A max--min universal audit
assigns positive rate to every coordinate when its certified library covers the finite
target set, including targets with no public anomaly.

Matched-budget experiments test robust deployment; a fixed-board public twin tests
active discovery; controlled instances test capacity frontiers, and simulated
showdown counts calibrate their predicted reveal rates.  Every floor-constrained probe
and response is audited against the floor.

\paragraph{Contributions.}
\begin{itemize}
  \item We characterize the usual passive showdown estimator and give target-level
        conditions under which floor-safe, reveal-certified flow identifies
        realization mass and positive-support behavior, complementing prior
        payoff-equivalence constraints
        (Thms.~\ref{thm:passive_mnar}--\ref{thm:active_id}).
  \item We introduce safe observation capacity, characterize its floor-priced frontier,
        and derive matching local rates under explicit factorization and parent-flow
        conditions (Thms.~\ref{thm:capacity_law} and~\ref{thm:lower_bound}).
  \item We combine public screening, floor-safe probes, and robust deployment in SAD;
        a universal audit gives positive coverage to predeclared public-null targets.
        Matched-budget and
        fixed-board studies evaluate the resulting certificates
        (Thm.~\ref{thm:portfolio}; \S\ref{sec:exp}).
\end{itemize}

\section{Related Work}
\label{sec:related}

\paragraph{Safe response and model acquisition.}
Equilibrium and safe subgame solving supply modern poker blueprints and responses
\citep{zinkevich2007regret,tammelin2014solving,bowling2015heads,
moravcik2017deepstack,brown2017safe,brown2019superhuman}.
Safe exploitation, including SES and OX-Search, refines responses to estimated models
\citep{mccracken2004safe,johanson2007robust,johanson2009data,
ganzfried2015safe,jeary2023safe,liu2022safe,ge2024safe}.
\citet{hoehn2005effective} select among equal-value poker strategies for faster
learning; \citet{davis2019solving} use a random non-fold probe for unbiased sequence
and showdown constraints.  We instead enforce a worst-case floor for every probe,
optimize target-specific reveal rate, and derive local sample-cost bounds;
response-refinement methods are complementary at deployment, not acquisition baselines
\citep{wang2011balancing}.

\paragraph{Identification under selective observation.}
Opponent models provide posteriors and response portfolios
\citep{southey2005bayes,ganzfried2011game,bard2013online}.  Recent consistency
guarantees assume identifiability and visitation \citep{ganzfried2025consistent};
passive models cannot remove missing-not-at-random selection without additional
structure \citep{little2019statistical}; our probes safely create revealing visitation.
Observation fibers connect this failure to partial
identification and monitoring \citep{manski2003partial,bartok2014partial}.
Selective-label methods account for decisions governing visibility
\citep{lakkaraju2017selective,wei2021decision,yang2022debiasing}.  Here interventions
must also preserve worst-case game value, and sequence-form flow recovers censored mass
on reveal-certified histories.

\paragraph{Pricing information under a value floor.}
Conservative learning, budgeted control, and safe policy-evaluation collection constrain
acquisition cost
\citep{wu2016conservative,yang2022reduction,wan2022safe,
boutilier2016budget,mukherjee2024saver}; pure exploration and controlled sensing price
information \citep{garivier2016optimal,camilleri2022active,
nitinawarat2015controlled,chernoff1959sequential}.  Our constraint is a global value
floor over a sequence-form polytope.  Under directional factorization,
reveal-certified flow and a local lower bound match the capacity upper rate when
parent flow is non-bottleneck.  Neither MNAR
conditioning nor parametric linear-program sensitivity alone links the floor to per-target sample
cost.

\section{Preliminaries and Problem Setup}
\label{sec:prelim}

We study repeated play in a finite two-player zero-sum extensive-form game with
perfect recall.  The agent starts from a certified blueprint and may adapt within a
fixed worst-case value floor.

\subsection{Sequence Form and the Safety Floor}
The sequence form \citep{vonstengel1996efficient} assigns a realization probability
to each action sequence.  Let $x\in\Xpoly$ and $y\in\Ypoly$ be the agent and opponent
realization plans in their sequence-form polytopes, and let $\Amat$ be the payoff
matrix, so $U(x,y)=x^\top\Amat y$.
The blueprint $x_{\mathrm{bp}}$ guarantees
$\vref=\min_{y\in\Ypoly}x_{\mathrm{bp}}^\top\Amat y$.  For safety budget
$\rok\ge0$ (per-policy value slack), define
$\Wval(x)=\min_{y\in\Ypoly}x^\top\Amat y$ and
\begin{equation}
\Sset=\{x\in\Xpoly: \Wval(x)\ge \vref-\rok\}.
\label{eq:safeset}
\end{equation}
Sequence-form duality makes $\Sset$ a polytope (\suppappendix{app:proofs}{A}).
Every deployed plan, including every probe, lies in $\Sset$.  Thus model error and
exploration may reduce exploitation value, but every selected plan retains worst-case
expected value at least $\vref-\rok$.

\subsection{Observation Protocol and Fiber}
Let $\Obs_x(y)$ denote the observation law under agent plan $x$ and opponent $y$.
The collection plan $x$ determines which opponent decisions the agent can reach.
Write $H$ for a public history and
$\mathcal I(H)=\{I_t=(H,t)\}$ for its type-specific opponent information sets.
Thus $\operatorname{pub}(I_t)=H$; throughout, $I$ includes the private type.
Full monitoring reveals $(I,a)$ at every opponent decision.  Under \emph{showdown
censoring}, the public history is observed, but a private type is revealed only at
showdown; a fold ends the hand and hides that type (Fig.~\ref{fig:mechanism}).

The \emph{public stream} records all public actions as reach-weighted frequencies.
The \emph{showdown stream} records private types only in hands that reach showdown.
Because the chosen action affects whether showdown occurs, normalizing the second
stream conditions on a selected event.  A passive fixed policy therefore need not
recover the uncensored behavior (\S\ref{sec:thy:passive}).
An active agent may instead use a plan in $\Sset$ to carry a selected line to
showdown.  Such a deviation is a \emph{floor-safe reveal probe}.

\begin{example}[running example]
\label{ex:run}
At public history $H$, the opponent's private type is $t\in\{H,L\}$ with prior
$\tfrac12$, giving information sets $I_H$ and $I_L$.  Facing a bet, it folds
($f$), hiding $t$, or calls ($c$), revealing $t$ at showdown.  Passive showdown
data therefore contain no folds.
\end{example}

\paragraph{Observation fiber.}
For a fixed protocol and collection plan $x$, define
$\mathcal F_x(y)=\{y'\in\Ypoly:\Obs_x(y')=\Obs_x(y)\}$.  Opponents in one fiber are
observationally indistinguishable and must receive the same deployed response.
The floor in \eqref{eq:safeset}, by contrast, quantifies over all of $\Ypoly$ and is
independent of the observation protocol.

\paragraph{Assumptions.}
The following conditions apply.
\par\noindent\textbf{(A1) Stationary sampling.} The opponent is fixed and nonadaptive.
Within each frozen collection batch, chance and behavioral randomization are
independent across hands.
\par\noindent\textbf{(A2) Known chance.} Chance probabilities and type priors are known,
as in a solved abstraction.
\par\noindent\textbf{(A3) Safe reveal-controllability.} A floor-safe suffix preserves the
target's maximal controlled reach and carries every non-fold continuation to a
type-revealing terminal, with no later non-revealing exit.  Last-decision and
all-in lines often satisfy this condition; other lines require a target-level
audit (Def.~\ref{def:class}).
\par\noindent\textbf{(A4) Residual payoff-nullity.} Every flow-preserving perturbation
$d$ invisible to the selected public-plus-active operator satisfies
$x^\top\Amat d=0$ for all $x\in\Sset$.  For fold-terminal residuals, it suffices that
fold payoff is type-independent and the operator fixes every safe plan's aggregate
fold mass.

Table~\ref{tab:guarantee_scope} summarizes the conditions for each guarantee.

\begin{table}[t]
\centering
\small
\begin{tabular}{@{}p{.25\columnwidth}p{.68\columnwidth}@{}}
\toprule
Claim & Conditions \\
\midrule
Floor & None beyond deployment in $\Sset$ (Thm.~\ref{thm:safety}). \\
Identification & (A2)--(A3) on prefix-closed certified targets; pooling also needs (A1)
(Thm.~\ref{thm:active_id}). \\
Per-target rate & Certified target, full-plan directional factorization, a
local public-preserving direction, and non-bottleneck parent flow
(Thms.~\ref{thm:capacity} and~\ref{thm:lower_bound}). \\
Universal audit & Certified rates with positive joint coverage
(Thm.~\ref{thm:portfolio}). \\
Value recovery & Nested exact-pin coverage at total hand budget $N_{\rm tot}$;
oracle comparison additionally needs payoff-relevant coverage or payoff-null residuals
(Thm.~\ref{thm:value_recovery}). \\
\bottomrule
\end{tabular}
\caption{Guarantees and their conditions.}
\label{tab:guarantee_scope}
\end{table}

\subsection{Public and Active Confidence Sets}
The agent maintains two nested confidence sets over $\Ypoly$.  The public set $\Cpub$
uses intervals for reach-weighted public action frequencies.
Because folds are public, $\Cpub$ constrains each aggregate frequency.  It cannot determine
how censored fold mass splits among the hidden types behind a public state.
A reveal probe aimed at $H$ pins non-fold mass only at those
$I\in\mathcal I(H)$ reached with positive reveal probability.  Adding these constraints yields
$\Cid\subseteq\Cpub$.  After $N$ hands, finite samples replace each exact realization-mass pin by a
simultaneous interval around its unbiased estimate.  For fixed opponent $\ystar$ and
failure probability $\delta$, $\Pr\{\ystar\in\Cid(N)\}\ge1-\delta$; the exact-pin
population set is
the zero-width limit.  Under a frozen probe, each reveal row is a Bernoulli mean over
all $N_R$ reveal hands, with expectation equal to known agent reach times opponent
realization mass.  Bonferroni empirical-Bernstein intervals over all rows give joint
coverage conditional on the independent pilot
(\suppappendix{app:pubconf}{B}).

\subsection{Problem Statement}
The agent repeatedly faces $\ystar$ and deploys $x_t\in\Sset$ each
round.  Any reveal probe also lies in $\Sset$.
The objective is to maximize $x_t^\top\Amat\ystar-\vref$ while certifying the value
recovered from active evidence.

\section{Theory}
\label{sec:theory}

The value floor is independent of the data-driven selector.  The statistical analysis
links passive bias, active identification, per-target observation cost, and universal
auditing.  Full proofs are in \suppappendix{app:proofs}{A}.

\subsection{Safety Is Selector-Independent}
\label{sec:thy:floor}
\begin{theorem}[Floor preservation]
\label{thm:safety}
Let $\mathcal A\subseteq\Sset$.  Any measurable selector that chooses $x\in\mathcal A$
using any model or observed data satisfies
\[
x^\top\Amat y\ge\vref-\rok\qquad\text{for every }y\in\Ypoly.
\]
The same guarantee holds in conditional and cumulative expectation against adaptive,
nonanticipating opponents, without stationarity or an opponent model.
\end{theorem}

Every reveal probe in $\Sset$ therefore preserves the floor.  For such probes, we
characterize the recoverable information and its observation rate.

\subsection{Passive Censoring Is Inconsistent}
\label{sec:thy:passive}
Fix a type-specific opponent information set $I=(H,t)$ and action $a$, with
$P(R\mid I)>0$, where $R$ denotes showdown.

\begin{theorem}[Passive showdown estimation is biased]
\label{thm:passive_mnar}
Under \textnormal{(A1)} and a fixed collection policy, assume the number of retained
observations at $I$ diverges.  The estimator that normalizes only over hands reaching
showdown converges almost surely to
\[
\begin{gathered}
\hat p_N(a\mid I)\longrightarrow p(a\mid I,R),\\[2pt]
p(a\mid I,R)=
\frac{P(R\mid I,a)}{P(R\mid I)}p(a\mid I).
\end{gathered}
\]
It recovers the uncensored action distribution on its positive support if and only if
$P(R\mid I,a)$ is constant over that support.  Any coordinate on which the selection
ratio $P(R\mid I,a)/P(R\mid I)$ differs from one and $p(a\mid I)>0$ has nonzero
asymptotic bias; a
two-sided interval centered on this estimator, whose random width converges to zero in
probability, then has coverage of the true coordinate tending to zero.
\end{theorem}

\emph{Proof.}  Bayes' rule gives the displayed selection ratio.  The strong law applies
to the retained observations; a shrinking interval centered at a distinct limit
eventually excludes the true coordinate.  \qed

More passive data only concentrate around the selected estimand.  In
Example~\ref{ex:run}, $R$ is the call, so
$\hat p(f\mid I_H),\hat p(f\mid I_L)\to0$, regardless of the true fold rates.

\subsection{Active Reveal Identifies Payoff-Relevant Behavior}
\label{sec:thy:active}
Active identification requires a target-level certificate.

\begin{definition}[Reveal certificate and certified class]
\label{def:class}
A target $I$ has a \emph{reveal certificate} at budget $\rok$ if a maximal-reach
floor-safe plan has a floor-safe extension preserving that reach and carrying every
non-fold continuation to a type-revealing terminal, with no later non-revealing exit.
The \emph{certified class} $\mathcal C_\rok$ contains these targets.  With (A2), a
safe-reachable member identifies non-fold children; prefix closure supplies the parent
and flow the fold mass.  Positive parent mass is needed only for normalization.
Pooling also requires (A1), and the lower bound a censored-fiber direction
(Def.~\ref{def:harddir}).
\end{definition}

For an opponent information set $I$, let $c_I\ge0$ denote the vector of known chance
weights for agent sequences leading to $I$, with opponent factors omitted.  The
\emph{controlled reach}
$\omega_x(I):=c_I^\top x$ is a linear reach weight determined by the agent and chance,
rather than the joint visitation probability.  Its safe observation capacity is
\[
\kcap(I)=\max_{x\in\Sset}\omega_x(I).
\]
Let $H_Ix=0$ denote whole-information-set flow equalities that force every
non-fold continuation at $I$ to a revealing terminal, and define
\[
\kappa^{\mathrm{rev}}_\rok(I)
=\max\{\omega_x(I):x\in\Sset,\ H_Ix=0\}.
\]
We set $\kappa^{\mathrm{rev}}_\rok(I)=0$ if this feasible set is empty.
For $I\in\mathcal C_\rok$, the certificate gives
$\kappa^{\mathrm{rev}}_\rok(I)=\kcap(I)$; otherwise $\kcap(I)$ is only
controlled-reach capacity.

Let $E_I$ be the event that the type is revealed after a non-fold continuation at $I$, and
$r_I(x,y)=\Pr_{x,y}(E_I)$.  \emph{Reveal factorization} (RF) means that some
$\pi_y(I)\in[0,1]$, independent of $x$, satisfies
$r_I(x,y)\le\omega_x(I)\pi_y(I)$ for every $x\in\Sset$, with equality for
reveal-forcing plans.  Here $\pi_y(I)$ is residual opponent reach and non-fold
mass.  Set $\lrate(I;y)=\max_{x\in\Sset}r_I(x,y)$.  Then
$\kappa^{\mathrm{rev}}_\rok(I)\pi_y(I)
\le\lrate(I;y)\le\kcap(I)\pi_y(I)$, with equality throughout for
$I\in\mathcal C_\rok$.

A reveal-certified target is \emph{safe-reachable} when $\kcap(I)>0$.  The identified
region $\mathcal R_\rok$ is the largest prefix-closed subset of
$\{I\in\mathcal C_\rok:\kcap(I)>0\}$.

\begin{lemma}[Flow recovers a censored action]
\label{lem:flow}
Let $\sigma(I)$ denote the opponent sequence reaching $I$, with child
$\sigma(I)a$.  Suppose the parent mass $y_{\sigma(I)}$ is identified and every non-fold mass
$y_{\sigma(I)a}$ is observed.  Then
\[
y_{\sigma(I)f}=y_{\sigma(I)}-\!\sum_{a\neq f}y_{\sigma(I)a}.
\]
If $y_{\sigma(I)}>0$, division by it recovers the behavior probabilities at $I$.
\end{lemma}

In Example~\ref{ex:run}, the probe reveals each call mass; flow supplies each fold.

\begin{theorem}[Active identification by flow]
\label{thm:active_id}
Under \textnormal{(A2)--(A3)}, a population collection mixture assigning positive mass
to a reveal certificate for every $I\in\mathcal R_\rok$ identifies any fixed opponent's
realization masses there, and its behavior probabilities on positive-parent support.
Pooling rounds additionally needs (A1).
\end{theorem}

\emph{Proof sketch.}  Let $d$ be the difference of two observationally equivalent
realization plans.  Each certified non-fold coordinate is a positive-coefficient row of
the active operator, hence vanishes in $d$.  Known root mass starts an induction:
prefix closure supplies each parent, and Lemma~\ref{lem:flow} forces its fold child.
Thus $d=0$ throughout $\mathcal R_\rok$.  \qed

Full identification at $H$ requires every relevant $I\in\mathcal I(H)$ in
$\mathcal R_\rok$; public rows pool over $H$, whereas active pins remain type-specific.
Outside this region, (A4) makes residuals payoff-null
(\suppappendices{app:necessity}{D.2}{app:residual}{D.3}).

\subsection{The Safe Observation Capacity Frontier}
\label{sec:thy:capacity}
For a reveal-certified target, capacity is computed by a single linear program (LP)
over $\Sset$: the certified suffix makes every reached non-fold continuation revealing
while preserving maximal controlled reach and the value floor.

\begin{theorem}[Safe per-target reveal rate]
\label{thm:capacity}
If $I\in\mathcal C_\rok$, \textnormal{(A1)--(A2)} and \textnormal{(RF)} hold, then
against an opponent with residual reach-and-continuation mass
$\pi_{\ystar}(I)$, the maximal reveal rate is
$\lrate(I;\ystar)=\kcap(I)\pi_{\ystar}(I)$.  On a local class with positive parent margin,
let $\lambda_{\rm par}(I)$ be any rate for which
a prefix certificate estimates $y_{\sigma(I)}$ to half-width
$O(\sqrt{\log(1/\delta)/(N\lambda_{\rm par}(I))})$, taking
$\lambda_{\rm par}(I)=\infty$ when it is known, and set
$\lambda_{\rm joint}(I)=\min\{\lrate(I;\ystar),\lambda_{\rm par}(I)\}$.  Conditional-probability
half-width $\varepsilon$ then holds with probability $1-\delta$ after
$N=O(\log(1/\delta)/(\lambda_{\rm joint}(I)\varepsilon^2))$ hands, where the constants
depend on the support margin.  In particular, the rate is governed by
$\lrate(I;\ystar)$ when
the parent certificate is no slower.
\end{theorem}

Every reveal-certified target with positive unconstrained reach is safe-reachable at
positive budget; a constructive mixture bound appears in
\suppapp{app:capacity_rate}{A.4}.

\begin{theorem}[Safe controlled-reach frontier]
\label{thm:capacity_law}
The map $\rok\mapsto\kcap(I)$ is non-decreasing, concave, and piecewise linear.  At
every budget, each optimal floor multiplier is a global supergradient.  The frontier
starts at $\kappa_0(I)=\max_{x\in\mathcal S(0)}\omega_x(I)$ and saturates at the
unconstrained reach
$\omega_{\max}(I):=\max_{x\in\Xpoly}\omega_x(I)$.  With
$\mu_I=\partial_+\kcap(I)|_{\rok=0}$,
\[
\kcap(I)\le
\min\{\omega_{\max}(I),\;\kappa_0(I)+\mu_I\rok\}.
\]
Equality holds on the maximal initial affine interval $[0,\rok_1]$, where
$\rok_1>0$ (and may be unbounded if the frontier is already saturated); later slopes
are non-increasing
(\suppapp{app:capacity_law}{A.6}).
\end{theorem}

Mixing feasible plans together with their budgets proves concavity.  Finitely many LP
bases give piecewise linearity, and each active floor dual supplies the supporting
slope of its face.
Pricing information across an interval requires the maximal reveal-forced capacity
$\kappa_\rho^{\mathrm{rev}}(I)$ to equal $\kcap(I)$ throughout that interval.

On the first face, a \emph{wall target} has $\kappa_0(I)=0$ and
$\kcap(I)=\mu_I\rok$.  Under (RF), a non-bottleneck parent, and capacity equality,
the floor dual determines both marginal reach and per-target information cost.
In Example~\ref{ex:run}, the blueprint does not reach $I$, so halving $\rok$ doubles
the required probes on the first face.

\subsection{Local Directional Observation Cost}
\label{sec:thy:lower}
\refstepcounter{definition}
\label{def:harddir}
\paragraph{Definition \thedefinition\ (Censored-fiber direction).}
Fix a public history $H$ and a base $\bar y$ with parent mass at least
$2\tau>0$ at each
affected $I\in\mathcal I(H)$.  A \emph{censored-fiber direction} is a feasible,
nonzero tangent $h$ on continue/fold children.  It fixes each parent flow, preserves
all public transcripts and prior-weighted continuation mass, and has positive
interior support on affected revealed labels.  Let $E_h$ be the event that an
affected non-fold label is observed.  For all sufficiently small $\varepsilon>0$, the
one-hand laws under $\bar y\pm\varepsilon h$ agree off $E_h$ for every
$x\in\Sset$.  Normalize
$\max_{\sigma(I)a\in\operatorname{supp}(h)}
|h_{\sigma(I)a}|/\bar y_{\sigma(I)}=1$.

Let $\Ypoly_\tau(H)\subseteq\Ypoly$ contain opponents whose affected parent masses
are at least $\tau$.  For $y\in\Ypoly_\tau(H)$, write
$p_y(a\mid I)=y_{\sigma(I)a}/y_{\sigma(I)}$ and let $p_H(y)$ collect the affected
conditional action probabilities.  Define
$d_H(y,y')=\lVert p_H(y)-p_H(y')\rVert_\infty$.  For
$C\subseteq\Ypoly_\tau(H)$, set
$\operatorname{wid}_H(C)=\tfrac12\sup_{y,y'\in C}d_H(y,y')$, and choose
$\varepsilon_0>0$ inside the stated support margin.  Thus $\varepsilon$ is a
conditional-action-probability half-width and
$d_H(\bar y+\varepsilon h,\bar y-\varepsilon h)=2\varepsilon$.

Define the direction's maximal safe informative-event rate
\[
\Lambda_\rok(h;y):=\max_{x\in\Sset}\Pr_{x,y}(E_h).
\]
The full-operator rate accommodates distinct reach weights across constituent
information sets.

\begin{theorem}[Local censored-fiber lower bound]
\label{thm:lower_bound}
For a censored-fiber direction $h$ at $H$, the alternatives satisfy
$d_H(\bar y+\varepsilon h,\bar y-\varepsilon h)=2\varepsilon$.  For
$0<\varepsilon\le\varepsilon_0$ and $0<\delta\le1/8$, any adaptive floor-safe observer
returning $C\subseteq\Ypoly_\tau(H)$ and satisfying
$\Pr_y\{y\in C,\operatorname{wid}_H(C)<\varepsilon\}\ge1-\delta$ at both alternatives
requires
$N=\Omega(\log(1/\delta)/(\Lambda_\rok(h;\bar y)\varepsilon^2))$.
If $\Pr_{x,y}(E_h)=\omega_x(I_0)\pi_y(I_0)$ for all $x\in\Sset$, then
$\Lambda_\rok(h;y)=\kcap(I_0)\pi_y(I_0)$, also equaling $\lrate(I_0;y)$
when $E_h$ and $E_{I_0}$ have identical rates.  If the parent flow is known
or certified at no slower rate, and $I_0$ is reveal-certified with
$\kappa^{\mathrm{rev}}_\rok(I_0)=\kcap(I_0)$, the bounds match in capacity,
accuracy, and confidence up to simultaneous-coverage and local-support factors,
which are not uniform at the support boundary.
\end{theorem}

\emph{Proof sketch.}  A valid set of half-width below $\varepsilon$ induces a test
between $\bar y+\varepsilon h$ and $\bar y-\varepsilon h$ with both errors at most
$\delta$.  Bretagnolle--Huber therefore requires Kullback--Leibler (KL) divergence at least
$\log(1/(4\delta))$.

For every safe one-hand plan, the laws agree off $E_h$ and, for a finite local
constant $C_h$, have KL at most $C_h\varepsilon^2\Pr(E_h)$.  Adaptive probes remain
in $\Sset$, so the KL chain rule gives total KL at most
$C_hN\Lambda_\rok(h;\bar y)\varepsilon^2$.  Combining the bounds proves
the rate.  The alternatives share their parent masses; hence the lower and upper rates
match when the parent certificate is no slower
(\suppapp{app:lowerbound}{A.5}).  \qed

Theorems~\ref{thm:capacity} and~\ref{thm:lower_bound} isolate the same bottleneck:
the rate at which a floor-safe plan generates informative reveals.  Under the
directional factorization in Theorem~\ref{thm:lower_bound} and a non-bottleneck
parent certificate, this rate is
$\kcap(I)\pi_{\ystar}(I)$; otherwise the full observation operator determines it
through $\Lambda_\rok(h;\ystar)$.

\subsection{Certified Value Recovery}
\label{sec:thy:value}
Identification becomes operational when a smaller opponent set improves the robust
response.  The next result separates the population value recovered by active
evidence from its finite-sample uncertainty.

\begin{theorem}[Certified value recovery]
\label{thm:value_recovery}
Let $C_\infty^{\mathrm{id}}$ be the exact-pin population active set containing
$\ystar$, and let
$\Gactivepop=\max_{x\in\Sset}\min_{y\in C_\infty^{\mathrm{id}}}x^\top\Amat y-\vref$.
Define
$V_N=\max_{x\in\Sset}\min_{y\in\Cid(N_{\rm tot})}x^\top\Amat y$, and let $x_N$
attain it.
If $C_\infty^{\mathrm{id}}\subseteq\Cid(N_{\rm tot})$ with probability at least
$1-\delta$, define
\[
\eta=\max_{x\in\Sset}\max_{y,y'\in\Cid(N_{\rm tot})}x^\top\Amat(y-y').
\]
The response $x_N$ is floor-safe; its certified gain $V_N-\vref$ lies in
$[\Gactivepop-\eta,\Gactivepop]$ on this nesting event, and its true-opponent
expected gain is at least
$\Gactivepop-\eta$.  If exact pins cover every payoff-relevant direction, or every
residual direction is payoff-null on $\Sset$ (in particular under (A4)), then
$\Gactivepop=\max_{x\in\Sset}x^\top\Amat\ystar-\vref$.
\end{theorem}

\subsection{Universal Auditing}
\label{sec:thy:portfolio}
Fix $m$ certified coordinates and $K$ floor-safe probes.  Let $R_{ij}\ge0$ be a
certified lower bound on the informative-event rate for coordinate $i$ under probe
$j$, including controlled reach, continuation, and reveal.
For $d\ge1$, let
$\Delta_d$ denote the probability simplex over $d$ coordinates.

\begin{theorem}[Safe universal audit]
\label{thm:portfolio}
For a probe mixture $\alpha\in\Delta_K$, the universal audit rate satisfies
\[
q_{\rm aud}:=\max_{\alpha\in\Delta_K}\min_i(R\alpha)_i
=\min_{u\in\Delta_m}\max_j u^\top R_{\cdot j}.
\]
Here $u\in\Delta_m$ weights coordinates, and the maximizing mixture is floor-safe.
If the audit receives fraction $\beta\in(0,1]$ of $N_{\rm tot}$ hands,
$q_{\rm aud}>0$, and simultaneous $(1-\delta)$-coverage intervals have radii at most
$C\sqrt{\log(2m/\delta)/(\beta N_{\rm tot}q_{\rm aud})}$ for some $C>0$, the audit
discovers every coordinate differing from its reference by at least a specified
separation $\Delta_{\rm sig}>0$ with probability $1-\delta$ once
$N_{\rm tot}>4C^2\log(2m/\delta)/(\beta q_{\rm aud}\Delta_{\rm sig}^2)$.
\end{theorem}

Convexity preserves the floor under mixing; LP duality gives the equality.

\section{Safe Active De-censoring}
\label{sec:scope}

Safe Active De-censoring (SAD) batches discovery, floor-safe collection,
confidence-set fitting, and robust deployment.  Public screening targets visible
deviations; a universal audit covers predeclared public-null targets.  All plans
remain in $\Sset$.

\subsection{Screen, Probe, and Respond}

\paragraph{Screening with public data.}
Let $\hat p_H$ be the action frequencies from a blueprint pilot and $p_H^{\rm bp}$ the
known blueprint reference.  With $n_H$ visits, action set $A_H$, $J$ tested action
coordinates, and screening failure budget $\delta_r$, SAD uses the simultaneous lower
bound
\[
\widehat D_{\rm scr}(H)=
\left[\tfrac12\lVert\hat p_H-p_H^{\rm bp}\rVert_1
-\tfrac{|A_H|}{2}\sqrt{\frac{\log(2J/\delta_r)}{2n_H}}\right]_+,
\]
where $[z]_+=\max\{z,0\}$.
Set $\widehat D_{\rm scr}(H)=0$ when $n_H=0$.
Under the reference opponent, this removes noise-only deviations simultaneously with
probability at least $1-\delta_r$.  Because public data cannot rank the type-specific
members of $\mathcal I(H)$, a positive score routes every reveal-certified member.
An audit-enabled instantiation reserves a fixed share for the max--min mixture in
Theorem~\ref{thm:portfolio}, whose active intervals can exclude the reference when all
public scores vanish.

\paragraph{Certified routing.}
For a floor-safe library $\{x_j\}_{j=1}^K$, let $M_{Ij}$ mark reveal certificates,
$\Omega_{Ij}=M_{Ij}\omega_{x_j}(I)$, and $R_{Ij}$ lower-bound the informative rate;
uncertified entries receive zero credit.  With screen weights
$w\ge0$ and library mixture $\alpha\in\Delta_K$, the route maximizes
$w^\top\Omega\alpha$.  Public-null coverage reserves a fixed
max--min $R\alpha$ share (Thm.~\ref{thm:portfolio}); a distinct width allocation is in
\suppapp{app:joint_algorithms}{C.2}.

\begin{algorithm}[b]
\caption{One batch of Safe Active De-censoring.}
\label{alg:sad}
\begin{algorithmic}[1]
\State \textbf{Discover.} Form $\widehat D_{\rm scr}$ from a public pilot; when
public-null coverage is required, use disjoint active reference exclusions from a
fixed max--min audit.
\State \textbf{Route.} Freeze the library, $M$, and $R$; mix the screened route
with any reserved max--min audit share.
\State \textbf{Probe.} Sample the mixture; refit only component-certified rows.
\State \textbf{Refit.} Add simultaneous one-hand reveal-indicator intervals under
the frozen portfolio to $\Cpub$, forming $\Cid$.
\State \textbf{Deploy.} Solve
$\max_{x\in\Sset}\min_{y\in\Cpub\cap\Cid}x^\top\Amat y$, reverting to $\Cpub$
and then the blueprint if necessary.
\end{algorithmic}
\end{algorithm}

Any selection audit is disjoint from the reveal batch constructing $\Cid$, and one
global failure budget is split across rows.  This preserves the coverage condition in
Theorem~\ref{thm:value_recovery} under data-dependent target selection.

\paragraph{Solving the robust response.}
Bucketed instances use one exact LP per confidence set.  On the unbucketed river,
cutting planes call best-response dynamic programs (DPs) without materializing
$\Amat$.  Weak-duality bounds and repair preserve exact floors
(\suppapp{app:fulldeck_solving}{E.1}); solver gaps affect only exploitation
certificates.

\section{Experiments}
\label{sec:exp}

\begin{table*}[t]
\centering
\small
\begin{tabular}{@{}lrr|rr|rrr@{}}
\toprule
& \multicolumn{2}{c}{$\Cpub$} & \multicolumn{2}{c}{Random reveal}
& \multicolumn{2}{c}{SAD} & \\
Opponent & Cert.\ gain & Pop.\ gain & Cert.\ gain & Pop.\ gain & Cert.\ gain & Pop.\ gain & Tgt. \\
\midrule
Equilibrium
& $0.000$ & $0.000$ & $0.000$ & $0.000$ & $0.000$ & $0.000$ & $0$ \\
River over-fold
& $.485\;(.003)$ & $.688\;(<.001)$
& $.588\;(.027)$ & $.765\;(.029)$
& $\mathbf{.692}\;(.002)$ & $.810\;(<.001)$ & $356$ \\
Turn over-fold
& $.045\;(.001)$ & $.080\;(.001)$
& $.071\;(.009)$ & $.138\;(.003)$
& $\mathbf{.093}\;(.001)$ & $.126\;(.002)$ & $90$ \\
Revealed call
& $.232\;(.001)$ & $.310\;(<.001)$
& $.294\;(.024)$ & $.397\;(.098)$
& $\mathbf{.297}\;(.001)$ & $.377\;(.001)$ & $315$ \\
\bottomrule
\end{tabular}
\caption{Matched-budget deployment on the coarse turn--river endgame
($N=10^6$, $\rho=0.5$).  Cert./Pop.\ gains are the robust lower bound/true-opponent
value minus $\vref=-0.0434$
(floor $-0.5434$); parentheses are 95\% Student-$t$ half-widths over 10 matched
seeds, with per-run simultaneous certificate coverage $0.9$.
$\Cpub$ uses $N$ public hands; active arms use a $0.2N/0.8N$ split; \emph{Tgt.}
is SAD's mean retained count of $2160$.  SAD beats $\Cpub$ in all six gain
contrasts (Holm $p\le0.012$) and Random in both certified over-fold contrasts
($p=0.004$); all $164$ cells pass the $10^{-6}$ floor audit.}
\label{tab:sad_e2e}
\end{table*}

\subsection{Experimental Protocol}

The evaluation tests passive-bias correction, sample-split SAD, public-null discovery,
and capacity/reveal calibration in bucketed turn--river endgames, a fixed-board
unbucketed river, and controlled Leduc, bandit, censored-chain, and MDP instances.
The coarse endgame has $5{,}221$ agent and
$9{,}721$ opponent sequences and a sequence-form Nash blueprint solved numerically
to tolerance.  Values are chips with ante $1$; \emph{certified}/\emph{population}
denote the robust-bound/true-opponent gain over $\vref$, except
\S\ref{sec:exp:g1} reports absolute subgame values.

Every arm receives $N$ hands.  Public-only uses $N$ blueprint hands; Random and SAD
split $0.2N$ for a public pilot and $0.8N$ for an independent reveal batch.  SAD
freezes its screened route before reveal; Random uses the same split with a
data-independent route over all $2160$ targets before the same safe weighted-reach
solve.  Arm-local confidence sets use failure budget $\delta=0.1$.  Active rows retain only showdown-committing
continuations, while public rows bound their parent flow.  This is the public-screened
branch of Algorithm~\ref{alg:sad}, with no audit batch;
\S\ref{sec:exp:g1} evaluates the disjoint audit branch.  An independent binary64
best-response audit uses tolerance $10^{-6}$; its minimum signed slack over $574$
controller and ablation records is $-6.53{\times}10^{-8}$.  Unless stated otherwise,
results report two-sided 95\% Student-$t$ confidence intervals (CIs) over 10 matched seeds; paired tests use
exact Wilcoxon signed ranks with Holm correction within each comparison family.
Solver and opponent details are in
\suppappendices{app:fulldeck}{E}{app:setup}{F}.

\subsection{End-to-End Safe Active De-censoring}
\label{sec:exp:deploy}

Table~\ref{tab:sad_e2e} reports the sample-split controller at $N=10^6$.
On the equilibrium control, the simultaneous screen retains no target and SAD returns
the blueprint.  On the river over-fold, it retains about $356$ of $2160$ type-specific
targets and raises the public-only certificate by $0.207$; its population gain rises by
$0.122$ and reaches $0.810$ versus the population safe oracle's $0.828$.  The turn
over-fold and revealed-call certificates rise by $0.048$ and $0.065$.

On both censored over-fold families, SAD's certificate exceeds Random on every matched
seed; both reveal arms improve on the already exposed revealed-call family.  All
$164$ primary and $270$ sweep/ablation cells are feasible, complete, and satisfy the
stated $10^{-6}$ floor-audit tolerance.  At $N=3{\times}10^5$, SAD's river
certificate exceeds $\Cpub$/Random by $0.111/0.134$ (Holm-adjusted paired
$p=0.004$) and its turn certificate exceeds $\Cpub$ by $0.027$.
At the tighter $\rho=0.1$ and $N=10^6$, SAD certifies $0.271$ versus
$0.174/0.234$ on the river and $0.040$ versus $0.021/0.021$ on the turn
(four adjusted comparisons, $p\le0.008$; \suppapp{app:deploy}{D.7}).

With the same split, unscreened capacity-only weights all $2158$ positive-capacity
targets by $\kcap(I)$ and matches Random (Holm-adjusted $p\ge0.316$).  SAD
exceeds it by $0.126\;(.003)/0.093\;(.003)$ at
$N=3{\times}10^5/10^6$ (95\% CI half-widths; adjusted $p=0.006$), isolating
public-anomaly routing.

\subsection{Public-Twin Identification and Discovery}
\label{sec:exp:g1}

\begin{figure}[t]
\centering
\includegraphics[width=\columnwidth]{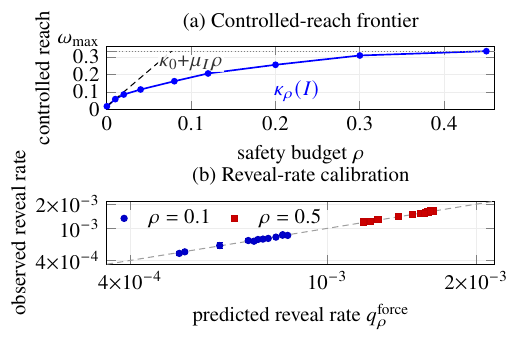}
\caption{Controlled reach and reveal-rate calibration.
(a)~Leduc frontier and origin tangent.
(b)~Five-seed observed and predicted rates for $12$ hold'em suffix targets at
$N=10^6$ (means and 95\% Student-$t$ CIs); dashed is equality.}
\label{fig:price_of_safety}
\end{figure}

The public twin preserves every public aggregate while transferring fold and call
mass between weak and strong hands (\suppapp{app:twin}{E.2}).  This unbucketed
subgame has $1081$ combinations per player and about $1.7\times10^4$ sequences;
a plan in $\Sset$ attains $V=0.815$.

Public evidence certifies only $\vref=-0.065$ at every budget.  Bucketed reveal data
certify $0.684$, $0.756$, and $0.794$ as $N$ grows from $10^5$ to $10^7$, and the
population reveal set recovers at least $96\%$ of the safe-exploitable gap
$V-\vref$.

A prespecified three-program root mixture excludes the twin in $100/100$ seeds
and $0/100$ controls at $N=3{,}000$ (95\% Wilson intervals $[96.3,100]\%$ and
$[0,3.7]\%$; \suppapp{app:public_null_audit}{E.3}).
At $N=10^6$, the same mixture, with a disjoint $0.2N/0.4N/0.4N$
public/audit/refit split and refitting only after exclusion, fires in $30/30$
twin seeds and $0/30$ controls.  It certifies $0.655\;(.008)$ at population
value $0.760\;(.009)$ (95\% CI half-widths), and all $60$
acquisition/response floor audits pass (\suppapp{app:audit_refit}{E.4}).

\subsection{Mechanism Analysis}
\label{sec:exp:passive}
\label{sec:exp:kappa}

\paragraph{Passive bias.}
At $N=10^5$, the showdown-only box conflicts with $\Cpub$ in all $40$ cells,
whereas the reveal-mass set remains feasible in all $120$
(\suppapp{app:coverage}{D.1}), as predicted by Theorem~\ref{thm:passive_mnar}.

\paragraph{Opponent population.}
Across $54$ procedural or undertrained counterfactual-regret-minimization
(CFR\textsuperscript{+}) opponents, public
fibers certify a median $73\%$ of the safe-exploitable gap; solved grouped reveal
fibers give $91\%$ among $51$ opponents (\suppapp{app:zoo}{E.6}).

\paragraph{Capacity and reveal rate.}
The Leduc frontier leaves its origin-dual tangent after the first face
(Fig.~\ref{fig:price_of_safety}a).  Across the $12$ nontrivial-suffix targets in
the predeclared set and $\rho\in\{0.1,0.5\}$, mean rates track
$q_\rho^{\mathrm{force}}:=\kappa^{\mathrm{rev}}_\rho\pi_{\ystar}$: log--log slope
$1.010$, $R^2=0.999$, and median relative error $1.1\%$ at $N=10^6$
(Fig.~\ref{fig:price_of_safety}b).  Simulator-only $\pi_{\ystar}$ is used post hoc,
not by SAD.\@ Across
$N\in\{10^5,3{\times}10^5,10^6\}$, conditional-interval half-width has slope
$-0.501$ against $Nq_\rho^{\mathrm{force}}$ ($R^2=0.959$).

\section{Discussion and Conclusion}
\label{sec:conclusion}

Safe observation capacity quantifies floor-safe access to censored behavior;
reveal-certified flow and SAD turn it into valid evidence and stronger floor-safe
response certificates.

\paragraph{Generative AI Use Disclosure.}
Generative AI assisted manuscript editing and code, figure, and diagnostic refinement.
The authors verified all mathematics, experiments, claims, citations, and conclusions
and retain full responsibility.

\section*{Acknowledgments}
We acknowledge computational resources and support provided by the Imperial College
Research Computing Service (\url{https://doi.org/10.14469/hpc/2232}).

\bibliography{references}

\appendix
\section{Proofs}
\label{app:proofs}

The proofs below concern the sequence-form objects of \S\ref{sec:prelim}; we write
$E_2y=e_2$ for the opponent realization constraints.  Roadmap: A.1 floor
(Thm.~\ref{thm:safety}); A.2 MNAR bias (Thm.~\ref{thm:passive_mnar}); A.3 identification
(Thm.~\ref{thm:active_id}, Prop.~\ref{prop:necessity}, and the A3 boundary); A.4
achievable capacity (Thm.~\ref{thm:capacity}); A.5 the local lower bound and conditional match
(Thm.~\ref{thm:lower_bound}); A.6 the frontier and decoupling
(Thm.~\ref{thm:capacity_law}, Prop.~\ref{prop:decoupling}); and A.7 recovery
(Thm.~\ref{thm:value_recovery}); A.8 proves universal auditing and joint routing
(Thm.~\ref{thm:portfolio}, Prop.~\ref{thm:joint_routing}).

\subsection{Safe Set and Floor Preservation}
The function $\Wval(x)=\min_{y\in\Ypoly}x^\top\Amat y$ is a pointwise minimum of
linear functions of $x$, hence concave and piecewise linear.  The condition
$\Wval(x)\ge\vref-\rok$ is equivalent to the opponent-side minimization LP
\[
\min\{(\Amat^\top x)^\top y:E_2y=e_2,\ y\ge0\}
\]
having value at least $\vref-\rok$.  By LP duality, this holds iff there exists a free
dual vector $\nu$ such that
\[
E_2^\top\nu\le \Amat^\top x,
\qquad
 e_2^\top\nu\ge \vref-\rok .
\]
Thus $\Sset$ is the projection of a polyhedron in $(x,\nu)$ intersected with $\Xpoly$.
The blueprint is feasible at $\rok=0$, so the safe set is nonempty for every budget used
in the paper.

\paragraph{Proof of Theorem~\ref{thm:safety}.}
If $x\in\mathcal A\subseteq\Sset$, then
\[
\min_{y\in\Ypoly}x^\top\Amat y=\Wval(x)\ge\vref-\rok .
\]
Therefore for every opponent plan $y\in\Ypoly$, $x^\top\Amat y\ge\vref-\rok$.  The
selector only chooses a member of $\mathcal A$, so the same bound holds for the deployed
plan.  More generally, let $\mathcal F_{t-1}$ contain the transcript before round $t$,
let $x_t\in\mathcal A$ be any history-measurable selection, and let the opponent choose
any history-dependent $y_t\in\Ypoly$, possibly after observing $x_t$.  If $Z_t$ is the
realized payoff, then
\[
\mathbb E[Z_t\mid\mathcal F_{t-1},x_t,y_t]
=x_t^\top\Amat y_t\ge\vref-\rok .
\]
Conditioning again if $y_t$ is randomized, then summing and applying the tower
property gives
$\mathbb E[\sum_{t=1}^T Z_t]\ge T(\vref-\rok)$.  The argument uses neither an
opponent model nor stationarity. \qed

\subsection{Passive MNAR Estimation}
\paragraph{Proof of Theorem~\ref{thm:passive_mnar}.}
Let $R$ be the reveal event.  By the law of total probability and Bayes' rule,
\[
P(a\mid I,R)=\frac{P(R\mid I,a)P(a\mid I)}{P(R\mid I)}.
\]
The passive showdown estimator averages actions only over samples with $R=1$, so by
the law of large numbers it converges to $P(a\mid I,R)$.  This equals $P(a\mid I)$ on
the positive action support exactly when $P(R\mid I,a)$ is action-independent there.
Otherwise the limit is
biased.  Any confidence interval centered at the estimator whose radius converges to
zero concentrates around the biased limit, so it eventually excludes the true value
$P(a\mid I)$ on the affected coordinate with probability tending to one. \qed

\subsection{Active Identification}
\paragraph{Proof of Lemma~\ref{lem:flow}.}
Sequence-form flow at $I$ gives
\[
y_{\sigma(I)}=\sum_{a\in A(I)}y_{\sigma(I)a}.
\]
If the parent mass is already identified and every non-fold child realization
$y_{\sigma(I)a}$ is observed, the fold mass is the remaining child mass,
\[
y_{\sigma(I)f}=y_{\sigma(I)}-\sum_{a\ne f}y_{\sigma(I)a}.
\]
When $y_{\sigma(I)}>0$, behavioral probabilities follow by normalizing by the parent. \qed

\paragraph{Proof of Theorem~\ref{thm:active_id}.}
Let $y\in\mathbb{R}^{|\Sigma_2|}$ be an opponent realization plan with sequence-form
flow $Ey=e$, $y\ge0$.  The identified region $\mathcal R_\rok$ is the prefix-closed set
defined in \S\ref{sec:thy:active}.
The argument is pointwise in this fixed population plan and uses (A2)--(A3);
(A1) is needed only to estimate the same plan from pooled rounds.
For each non-fold sequence $\sigma(I)a$ with $I\in\mathcal R_\rok$, membership in
$\mathcal C_\rok$ supplies a floor-safe suffix that carries the entire mass
$y_{\sigma(I)a}$ to a type-revealing terminal.  Safe reachability supplies a
collection plan with known coefficient $\omega_x^{\mathrm{rev}}(I)>0$, so the
corresponding event probability is
$\omega_x^{\mathrm{rev}}(I)y_{\sigma(I)a}$.  Hence that realization coordinate is
pinned by a row of the active observation operator $M$.

For two observationally indistinguishable opponents, let
$d$ be the difference of their realization plans.  Then $Ed=0$ and $Md=0$.  Induct on
opponent-sequence depth within the prefix-closed region.  Chance and type priors fix the root.
If $d$ vanishes on the parent sequence $\sigma(I)$, every non-fold child is a row of
$Md$ and hence zero.  Flow then gives
\[
d_{\sigma(I)f}=d_{\sigma(I)}-\sum_{a\ne f}d_{\sigma(I)a}=0.
\]
Thus $d$ vanishes on every realization coordinate in $\mathcal R_\rok$.  When the true
parent mass is zero, nonnegativity and flow pin all children as well.  This proves
identification on the stated region.  Outside it, (A4) makes every global residual
direction payoff-null.  On fold-terminal support, a type-independent conditional
fold payoff together with reach-weighted aggregate invariance for every
$x\in\Sset$ is sufficient.  A public
history $H$ is fully identified only if the induction covers every relevant
$I\in\mathcal I(H)$. \qed

\begin{proposition}[Reveal-controllability cannot be omitted for point identification]
\label{prop:necessity}
There is a finite censored game without reveal-controllability in which two distinct
opponent realization plans induce the same law of public actions, showdown labels,
and realized payoff under every agent plan.
\end{proposition}

\begin{proof}[Proof of Proposition~\ref{prop:necessity}]
Let the opponent type be $t\in\{H,L\}$ with prior $1/2$.  It emits a public signal
$b\in\{A,B\}$ and then terminates at payoff zero without revealing $t$; the agent
cannot force a reveal.

Define $y^A$ by $H\mapsto A,L\mapsto B$ and $y^B$ by
$H\mapsto B,L\mapsto A$.  Under either opponent the public signal is uniform and
there is no showdown label; payoff is also identically zero.  Hence the complete
monitored laws coincide for every agent plan although the two type-conditioned
realization plans differ.  A type-revealing suffix would distinguish the pairs
$(A,H),(B,L)$ from $(A,L),(B,H)$, which is the role of (A3).  Payoff recovery can
require less than point identification precisely when the remaining fiber is
payoff-null, as formalized below.
\end{proof}

\begin{proposition}[Public-history scores rank states, not types]
\label{prop:public_rank}
Any score computed solely from the local public counts at $H$ is $H$-measurable,
hence equal for targets sharing that public state but holding different hidden types.
Let $V(I)$ be any fixed target-level exploitation score and let $P$ be any distribution
over the finite target set.  Then
\[
\begin{aligned}
\operatorname{Var}_P[V(I)]
&=\operatorname{Var}_P\!\big(\mathbb E_P[V(I)\mid H]\big)\\
&\quad+\mathbb E_P\!\big[\operatorname{Var}_P(V(I)\mid H)\big].
\end{aligned}
\]
The first term is between-state variation accessible to an $H$-measurable ranking;
the second is within-state variation that no such ranking can order.
\end{proposition}
\begin{proof}[Proof of Proposition~\ref{prop:public_rank}]
The local public counts at a state sum over its hidden types, so any statistic
restricted to those counts is measurable with respect to $H$.  The displayed
identity is the law of total variance.  Conditional variation is invisible to
any $H$-measurable ranking.
\end{proof}

\subsection{Safe Observation Capacity}
\label{app:capacity_rate}
\begin{proposition}[Constructive safe-reach bound]
\label{prop:capacity_mix}
Fix $I$ and $q\in\Xpoly$.  Set $L_I=(\vref-\Wval(q))_+$ and let
$\alpha_I(\rok)=1$ if $L_I=0$, otherwise
$\alpha_I(\rok)=\min\{1,\rok/L_I\}$.  Then
\[
\begin{aligned}
\kcap(I)&\ge \kappa_0(I)
 +\alpha_I(\rok)\bigl(\omega_q(I)-\kappa_0(I)\bigr)_+,\\
\kappa_0(I)&=\max_{x\in\mathcal S(0)}\omega_x(I).
\end{aligned}
\]
\end{proposition}
\begin{proof}
Let $x_0\in\arg\max_{x\in\mathcal S(0)}\omega_x(I)$, so
$\omega_{x_0}(I)=\kappa_0(I)$.  If $\omega_q(I)\le\kappa_0(I)$, deploying $x_0$
already proves the bound.  Otherwise, since $\Wval$ is concave,
\[
\begin{aligned}
\Wval((1-\alpha)x_0+\alpha q)
&\ge (1-\alpha)\Wval(x_0)+\alpha\Wval(q)\\
&\ge \vref-\alpha L_I.
\end{aligned}
\]
For $\alpha\le\alpha_I(\rok)$ the mixture lies in $\Sset$ (when $L_I=0$, the whole
segment is floor-safe).  Linearity of reach gives
\[
\omega_{(1-\alpha)x_0+\alpha q}(I)
=\kappa_0(I)+\alpha(\omega_q(I)-\kappa_0(I)).
\]
Taking $\alpha=\alpha_I(\rok)$ proves the stated bound.  If $\kappa_0(I)=0$ and
$q$ has positive reach, this is positive for every $\rok>0$; if
$\kappa_0(I)>0$, safe reachability is immediate.
\end{proof}

\paragraph{Proof of Theorem~\ref{thm:capacity}.}
Assume $I\in\mathcal C_\rok$ and
deploy the certified extension of a capacity-achieving safe probe for $N$ hands.
Under (A1)--(A2) and (RF), a hand produces a revealed non-fold label at $I$ with
probability $\lrate(I)=\kcap(I)\pi_{\ystar}(I)$.  Thus the successful-reveal count
$S_N$ has mean $N\lrate(I)$, and a multiplicative Chernoff bound gives
$S_N=\Theta(N\lrate(I))$ with probability at least $1-\delta$ once
$N\lrate(I)\gtrsim\log(1/\delta)$.  Concentrating the action-specific reveal
indicators and dividing by the known controlled reach estimates each revealed child
realization mass with half-width
$O(\sqrt{\log(1/\delta)/(N\lrate(I))})$ on the positive-support neighborhood of
Definition~\ref{def:harddir}; support margins absorb fixed label shares.
Indeed, if $Z_{n,a}$ is the reveal-row indicator for non-fold child $a$ and
$b_{I,a}(x)>0$ is its known agent--chance coefficient, then
\[
\begin{aligned}
\mathbb E Z_{n,a}&=b_{I,a}(x)y_{\sigma(I)a},\\
\widehat y_{\sigma(I)a}&=\bar Z_a/b_{I,a}(x),\\
\left|\widehat y_{\sigma(I)a}-y_{\sigma(I)a}\right|
&=O\!\left(\sqrt{\frac{\log(1/\delta)}{N\lrate(I)}}\right),
\end{aligned}
\]
where the second relation uses the fixed interior label share
$\mathbb E Z_{n,a}=\Theta(\lrate(I))$ and empirical-Bernstein concentration.

Conditional probabilities also divide by $y_{\sigma(I)}$, which the target reveal
events alone need not identify.  Define $\lambda_{\rm par}(I)$ by requiring a
prefix certificate for this parent mass with half-width
$O(\sqrt{\log(1/\delta)/(N\lambda_{\rm par}(I))})$, and take
$\lambda_{\rm par}(I)=\infty$ when the parent is known.  On a neighborhood where the
parent mass is bounded away from zero, the ratio map
$p_y(a\mid I)=y_{\sigma(I)a}/y_{\sigma(I)}$ is Lipschitz: if
$y_{\sigma(I)}\ge m_I>0$ and $|\Delta_{\rm par}|<m_I$, then
\[
|\widehat p(a\mid I)-p(a\mid I)|
\le \frac{|\Delta_a|+p(a\mid I)|\Delta_{\rm par}|}
{m_I-|\Delta_{\rm par}|}.
\]
Thus fixed positive label shares and $m_I$ enter only the local constant.  A union
bound gives
\[
\operatorname{wid}_I(C)
=O\!\left(\sqrt{\frac{\log(1/\delta)}
 {N\min\{\lrate(I),\lambda_{\rm par}(I)\}}}\right).
\]
Hence conditional half-width $\varepsilon$ requires
\[
\begin{aligned}
N&=\tilde O\!\left(\frac{1}{\lambda_{\rm joint}(I)\varepsilon^2}\right),\\
\lambda_{\rm joint}(I)&=\min\{\lrate(I),\lambda_{\rm par}(I)\}.
\end{aligned}
\]
The simpler
$\lrate(I)$ rate holds exactly when the parent is known or its prefix certificate is
non-bottleneck. \qed

\subsection{Safe De-censoring Rate Lower Bound}
\label{app:lowerbound}
We prove Theorem~\ref{thm:lower_bound}: no floor-safe observer---however it allocates
its probes---can certify a censored-fiber direction faster than its maximal
informative-event rate $\Lambda_\rok(h;\bar y)$.  Under directional reveal
factorization this becomes the scalar rate $\lrate(I)$.  We make the fiber conditions precise, prove the bound, add
the exploitation corollary, and verify the conditions for unequal priors, several non-fold
actions, and deep public-null continuations.

\paragraph{Operational and global certificates.}
The certificate of Definition~\ref{def:class} is realized \emph{operationally} at a
capacity-achieving plan: fix $x_I^\star\in\arg\max_{x\in\Sset}\omega_x(I)$; the target carries an
\emph{operational reveal certificate} if $x_I^\star$ extends by a suffix to a floor-safe plan of
equal reach ($\omega(I)$ unchanged, worst-case value still $\ge\vref-\rok$) whose every non-fold
continuation at $I$ is forced to a type-revealing terminal, with no later opponent move to a
non-revealing terminal and known chance/type weights along the suffix (A2).  A stronger \emph{global
certificate}---a suffix operator $T_I:\Xpoly\to\Xpoly$ with the same forcing property and
$\omega_{T_Ix}(I)=\omega_x(I)$, $\Wval(T_Ix)=\Wval(x)$ for \emph{every} $x\in\Sset$---is a
sufficient condition we state but do not require.
The operational certificate makes $x_I^\star$ feasible for the reveal-forcing
program at the same reach, so
$\kappa^{\mathrm{rev}}_\rok(I)=\kcap(I)$ and the reach LP of
Theorem~\ref{thm:capacity}: reveal-controllability (A3) as a per-instance certificate rather
than a structural claim about poker.  The lower bound uses the full informative-event
rate $\Lambda_\rok$; $\kappa^{\mathrm{rev}}_\rok$ supplies a constructive rate, and the
two match under full-plan factorization and capacity equality.
Membership is verified by the three-part audit of Appendix~\ref{app:certaudit}.  All
$60$ reported targets receive an audited reveal-forcing plan.  For $48$
direct-showdown targets the suffix is empty, so
$\kappa^{\mathrm{rev}}_\rok(I)=\kcap(I)$ and membership in $\mathcal C_\rok$ is
certified.  At the remaining $12$, a non-fold action
can reopen betting; a constrained LP with an outward-rounded upper certificate
over the same fixed-action treeplex forces call at every
chance-compatible player-$0$ information set after that action and computes
$\kappa^{\mathrm{rev}}_\rok(I)$
directly.  Its audit checks the forced-action flow equalities, showdown reach for
every non-fold child, and the value floor; these $12$ probes are not claimed to
preserve unconstrained maximal reach.  The lower-bound hold'em targets below
also meet the global certificate: they are last-decision river nodes whose non-fold
children lead directly to showdown.

\paragraph{How far A3 can be weakened: the payoff-nullity boundary.}
For a complete collection portfolio, let $M_{\mathrm w}$ be its global observation
operator, including public rows and every forced coordinate, and define
$\Nres=\{d:E_2d=0,\ M_{\mathrm w}d=0\}$.  A3 on one target annihilates residual
directions supported on that target's prefix closure; global point identification
requires the portfolio to cover every relevant closure.  A partial portfolio leaves a
larger residual nullspace, but its uncertainty can still be payoff-null.  Proposition
\ref{prop:necessity} shows that point identification can otherwise fail.  Write the
\emph{safe payoff spread} of a direction as
\[
\eta(d)=\max_{x\in\Sset}x^\top\Amat d-\min_{x\in\Sset}x^\top\Amat d,
\]
and, for a fixed opponent $\ystar$, let
$\mathcal F_M(\ystar)=\{y\in\Ypoly:M_{\mathrm w}y=M_{\mathrm w}\ystar\}$ be the observation fiber and
\[
\Delta_M^-(\ystar)=\max_{x\in\Sset}\ \max_{y\in\mathcal F_M(\ystar)}\ x^\top\Amat(\ystar-y)\ \ge\ 0
\]
its \emph{one-sided safe-payoff width}.  Because $\mathcal F_M(\ystar)$ is a bounded polytope the maximum
is attained and finite---unlike $\sup_{d\in\Nres}\eta(d)$, which over a nonzero subspace is $0$ or
$+\infty$ and so cannot bound a residual loss---and
$\Delta_M^-(\ystar)\le\Delta_M:=\max_{x\in\Sset}\max_{y,y'\in\Ypoly:\,M_{\mathrm w}y=M_{\mathrm w}y'}|x^\top\Amat(y-y')|$,
the opponent-independent protocol diameter.

\begin{proposition}[Payoff-nullity identifies the safe-payoff functional]
\label{prop:payoff_boundary}
Fix a global forcing portfolio with observation operator $M_{\mathrm w}$ and residual
nullspace $\Nres$.
\textnormal{(i)} If every $d\in\Nres$ is \emph{payoff-null on the safe set}, $x^\top\Amat d=0$ for all
$x\in\Sset$, then observations determine the safe-payoff functional $x\mapsto x^\top\Amat y$ on $\Sset$,
so the robust response recovers the safe oracle value,
\[
\max_{x\in\Sset}\,\min_{y:\,M_{\mathrm w}y=M_{\mathrm w}\ystar}x^\top\Amat y
=\max_{x\in\Sset}x^\top\Amat\ystar,
\]
even though $\ystar$ is identified only modulo $\Nres$.
\textnormal{(ii)} Conversely, if $d\in\Nres$, $\bar y\in\Ypoly$, and $t>0$ satisfy
$\bar y\pm td\in\Ypoly$ and $x^\top\Amat d\neq0$ for some $x\in\Sset$, then two observationally indistinguishable
opponents give different payoffs to that safe plan, so the observations do not
determine the full safe-payoff functional.  This condition alone need not change the
maximized safe oracle value, because the affected plan may be suboptimal.
Writing $R=\max_{x\in\Sset}\min_{y\in\mathcal F_M(\ystar)}x^\top\Amat y$ for the robust-over-fiber value,
the shortfall is bounded by the one-sided width,
\[
0\le\max_{x\in\Sset}x^\top\Amat\ystar-R\le\Delta_M^-(\ystar),
\]
while the floor $R\ge\vref-\rok$ holds unconditionally.  The exact shortfall is the gap between two
linear programs---an $\Sset$-constrained best response to $\ystar$ and the robust response over
$\mathcal F_M(\ystar)$---so it is computable per target; the uniform diameter $\Delta_M$ is bilinear in
$(x,y)$ and is evaluated by vertex enumeration or as a certified bound.
Full reveal-controllability over every payoff-relevant closure gives
$\Nres=\{0\}$ there and is sufficient but not necessary: a partial global portfolio
also recovers the safe-payoff functional whenever all of its residual directions are
payoff-null.
\end{proposition}

\begin{proof}
(i) $M_{\mathrm w}y=M_{\mathrm w}\ystar$ iff $y-\ystar\in\Nres$, so payoff-nullity gives
$x^\top\Amat y=x^\top\Amat\ystar$ for every $x\in\Sset$ and every observationally equivalent $y$: the
inner minimum over the fiber equals $x^\top\Amat\ystar$ pointwise, and the outer maximum is the oracle.
(ii) The two feasible plans $\bar y\pm td$ share
$M_{\mathrm w}\bar y$ but their payoffs to $x$ differ by
$2t x^\top\Amat d\neq0$, proving failure of payoff-functional identification.
For the oracle maximizer $x^\star$ of $\ystar$, taking $x=x^\star$ in the outer maximum of $R$,
\[
\begin{aligned}
R&\ge\min_{y\in\mathcal F_M(\ystar)}(x^\star)^\top\Amat y\\
&=(x^\star)^\top\Amat\ystar-\max_{y\in\mathcal F_M(\ystar)}(x^\star)^\top\Amat(\ystar-y)\\
&\ge(x^\star)^\top\Amat\ystar-\Delta_M^-(\ystar)
=\max_{x\in\Sset}x^\top\Amat\ystar-\Delta_M^-(\ystar),
\end{aligned}
\]
and $R\le\max_{x\in\Sset}x^\top\Amat\ystar$ since $\ystar\in\mathcal F_M(\ystar)$; the floor follows because every
response lies in $\Sset$.  A partial portfolio can satisfy the premise, for example
when each unforced hidden split lies in the kernel of
$x\mapsto x^\top\Amat d$ over $\Sset$; the fold-terminal sufficient condition in
(A4) is one such case. \qed
\end{proof}

Two distinctions are essential.  First, payoff-nullity is the \emph{global} safe-payoff condition and is
strictly stronger than the local kink $\gamma$ of Prop.~\ref{prop:kink}: since the base-optimal face
$F_0\subseteq\Sset$, one has $\gamma(d)\le\eta(d)$, with strict gap exactly when a direction is smooth at
the reference opponent ($\gamma=0$) yet moves value elsewhere in $\Sset$ ($\eta>0$).  The kink governs
the \emph{lower bound}---it suffices to make a split hard to learn---whereas payoff-nullity governs
\emph{identification}---it is what licenses leaving a split unforced; conflating them is unsound, since
directions can have $\gamma=0$ at a smooth reference while remaining non-null on
$\Sset$.  Second, the condition $x^\top\Amat d=0$ for every global residual $d$
characterizes the exact payoff-null enlargement.  Because payoff-nullity is a measure-zero linear condition and unforced
value-relevant exits are generic in multi-street poker, this characterization also
explains why (A3) is the practical sufficient condition.
Equality of robust and safe-oracle optima can nevertheless occur without full
payoff-nullity when
residual directions affect only suboptimal plans.  A sufficient condition for uniform
oracle recovery to fail is that one fiber contain two opponents with disjoint
safe-optimal response sets.

\paragraph{Robust response is minimax-optimal given the forced coordinates.}
Whatever a protocol forces, the robust-over-fiber response is the best a floor-safe agent can do with that
evidence.  Fix an exact observation $\sigma=M_{\mathrm w}\ystar$ and its fiber
$\mathcal F_\sigma=\{y\in\Ypoly:M_{\mathrm w}y=\sigma\}$.

\begin{proposition}[Fiber minimax optimality]
\label{thm:fiber_minimax}
Among all floor-safe deployment rules that depend on the data only through $\sigma$---each committing a
single plan $x(\sigma)\in\Sset$ to the whole fiber---the largest guaranteed value over $\mathcal F_\sigma$ is
\[
R_\sigma=\max_{x\in\Sset}\ \min_{y\in\mathcal F_\sigma}\ x^\top\Amat y,
\]
attained by the robust-over-fiber response
$x_\sigma\in\arg\max_{x\in\Sset}\min_{y\in\mathcal F_\sigma}x^\top\Amat y$.  No floor-safe
$\sigma$-measurable rule guarantees more, and every such rule is floor-safe.
\end{proposition}

\begin{proof}
A $\sigma$-measurable rule selects some $x(\sigma)\in\Sset$; since $M_{\mathrm w}y=\sigma$ for every
$y\in\mathcal F_\sigma$, the rule cannot separate members of the fiber, so its worst-case value there is
$\min_{y\in\mathcal F_\sigma}x(\sigma)^\top\Amat y\le\max_{x\in\Sset}\min_{y\in\mathcal F_\sigma}x^\top\Amat y=R_\sigma$,
with equality at $x_\sigma$.  Floor-safety holds because $x(\sigma)\in\Sset$ gives
$x(\sigma)^\top\Amat y\ge\vref-\rok$ for all $y\in\Ypoly\supseteq\mathcal F_\sigma$
(Thm.~\ref{thm:safety}). \qed
\end{proof}

The optimality is \emph{relative to the protocol}: $R_\sigma$ is the ceiling for the observation operator
$M_{\mathrm w}$ in force, and a richer protocol that forces more continuations can only raise it, shrinking
$\mathcal F_\sigma$.  Under A3 the fiber collapses to $\{\ystar\}$ on the closure and
$R_\sigma=\max_{x\in\Sset}x^\top\Amat\ystar$ is the oracle; in general
$R_\sigma\ge\max_{x\in\Sset}x^\top\Amat\ystar-\Delta_M^-(\ystar)$ by Prop.~\ref{prop:payoff_boundary}.
The finite-sample analogue adds the statistical width to the structural one.

\begin{proposition}[Residual-width recovery]
\label{thm:residual_recovery}
Let $C_N$ be a confidence set built from the partial-forcing reveals that contains $\ystar$ with
probability at least $1-\delta$ and whose fiber-robust value
$J_N=\max_{x\in\Sset}\min_{y\in C_N}x^\top\Amat y$ obeys $R_\sigma-J_N\le\zeta_N$, and put
$x_N\in\arg\max_{x\in\Sset}\min_{y\in C_N}x^\top\Amat y$.  Then $x_N$ is floor-safe and, on the coverage
event,
\[
\max_{x\in\Sset}x^\top\Amat\ystar-x_N^\top\Amat\ystar\ \le\ \Delta_M^-(\ystar)+\zeta_N .
\]
This proposition is deterministic conditional on the stated bound for $\zeta_N$.
A statistical rate additionally requires simultaneous contraction of every
value-relevant observation row and a finite reconstruction constant, as in the
proof of Theorem~\ref{thm:value_recovery}.  Under A3 or payoff-nullity
$\Delta_M^-(\ystar)=0$ and the bound reduces to its statistical term; otherwise
the structural residual persists.
\end{proposition}

\begin{proof}
On coverage $\ystar\in C_N$, so $x_N^\top\Amat\ystar\ge\min_{y\in C_N}x_N^\top\Amat y=J_N\ge R_\sigma-\zeta_N$.
By Proposition~\ref{thm:fiber_minimax} and Prop.~\ref{prop:payoff_boundary}(ii),
$R_\sigma\ge\max_{x\in\Sset}x^\top\Amat\ystar-\Delta_M^-(\ystar)$.  Combining the two inequalities gives the
bound; floor-safety is Theorem~\ref{thm:safety}. \qed
\end{proof}

\paragraph{Structural fragments that satisfy A3.}
The first two structures below need no additional floor slack.  The capped construction
requires enough unused slack for its call-down suffix.

\begin{lemma}[Reveal-controllable fragments]
\label{lem:fragments}
A target $I$ with $\kcap(I)>0$ lies in $\mathcal C_\rok$ if either of the following holds.
\textnormal{(a) Immediate showdown:} the opponent's non-fold action at $I$ itself
reaches a type-revealing terminal, so no agent suffix is required.
\textnormal{(b) Chance-only runout:} after the continuation, only chance acts before a
forced showdown (e.g.\ an all-in runout).
\textnormal{(c) Capped call-down (budgeted):} a maximal-reach plan retains at least $\bar c$ of
unused floor slack, where $\bar c$ bounds the additional worst-case loss of a
never-betting check/call suffix.
\end{lemma}

\begin{proof}
In each case every non-fold continuation reaches a type-revealing terminal with no opponent
exit, placing those coordinates in $M$; (a) and (b) require no additional agent action,
while in (c) the retained slack covers the suffix loss $\bar c$.  The identification
induction of Theorem~\ref{thm:active_id} then applies. \qed
\end{proof}

The complementary \emph{bad set} $\mathcal B_\rok$ collects raw-reachable targets with a non-fold
continuation that has no floor-safe reveal-forcing suffix and leaves a residual direction of nonzero
spread ($\eta>0$).  Theorem~\ref{thm:active_id} covers the prefix-closed certified region;
the residual-width results above handle $\mathcal B_\rok$.  Fragment (c) gives the multi-street implication: in
fixed-limit or capped-bet subgames a larger slice of the tree is reveal-controllable at a given
budget.  In deep no-limit play call-down can cost more than $\rok$, producing a
nonempty $\mathcal B_\rok$ and an explicit residual-width term.

\paragraph{Forcing cost as a budget tradeoff.}
The binary ``A3 holds or not'' coarsens a continuous tradeoff already priced by the reveal capacity
$\kappa^{\mathrm{rev}}_\rok(I)=
\max\{\omega_x(I):x\in\Sset,\ H_Ix=0\}$ of
Prop.~\ref{prop:decoupling}.  The budget needed to reveal-force reach $\alpha$ at $I$ is its inverse,
\[
\begin{aligned}
c_I(\alpha)&=\min\{\rok\ge0:\kappa^{\mathrm{rev}}_\rok(I)\ge\alpha\}\\
&=\left[\vref-\sup_{\substack{x\in\Xpoly,\ H_Ix=0\\
\omega_x(I)\ge\alpha}}\Wval(x)\right]_+,
\end{aligned}
\]
with $c_I(\alpha)=+\infty$ when the constrained set is empty.  The inner problem is
a single LP, and $c_I$ is non-decreasing in $\alpha$.  Full reveal forcing at
reach $\alpha$ is affordable iff $c_I(\alpha)\le\rok$; otherwise this constrained
construction cannot attain $\alpha$, and the full observation operator determines
the remaining ambiguity.  This gives a budget-indexed version of A3 without
excluding alternative partially informative plans.

\paragraph{Censored-fiber directions.}
\begin{definition}[Detailed censored-fiber conditions]
Fix a public history $H$, a base opponent $\bar y\in\Ypoly$ whose affected parent
masses are at least $2\tau>0$, and a nonzero tangent
$h$ such that $\bar y\pm uh\in\Ypoly$ for $0\le u\le u_0$.  Let
$\mathcal I_h(H)\subseteq\mathcal I(H)$ be its affected type-specific information
sets.  Call $h$ a \emph{censored-fiber direction at $H$} if it is supported on their
continue/fold children, preserves every public transcript
($B_Hh=0$, where $B_H$ is the public-transcript operator) and prior-weighted
continuation mass, and has a positive interior margin on
the affected revealed labels.  Let $E_h$ denote observation of an affected non-fold
label.  For every $x\in\Sset$, every sufficiently small $\varepsilon$, and every
one-hand observation $z\notin E_h$, also require
\[
\Obs_x(\bar y+\varepsilon h)(z)=\Obs_x(\bar y-\varepsilon h)(z).
\]
The public constraint $B_Hh=0$ is necessary for this equality but need not imply it
for an arbitrary observation protocol.  For every $I\in\mathcal I_h(H)$ require
$h_{\sigma(I)}=0$ and normalize
\[
\lVert h\rVert_{H,\bar y}
=\max_{\substack{I\in\mathcal I_h(H),\,a:\\
\sigma(I)a\in\mathrm{supp}(h)}}
\frac{|h_{\sigma(I)a}|}{\bar y_{\sigma(I)}}=1.
\]
Restrict the local parameter space and confidence sets to
$\Ypoly_\tau(H)=\{y\in\Ypoly:y_{\sigma(I)}\ge\tau\
\forall I\in\mathcal I_h(H)\}$.  Consequently the conditional-action metric in
Definition~\ref{def:harddir} is well-defined and satisfies
$d_H(\bar y+\varepsilon h,\bar y-\varepsilon h)=2\varepsilon$.  The explicit
off-event equality, interior margin, and finite observation alphabet give, uniformly
over $x\in\Sset$,
\[
\mathrm{KL}\!\left(\Obs_x(\bar y+\varepsilon h)\,\middle\|\,
\Obs_x(\bar y-\varepsilon h)\right)
\le C_h\varepsilon^2\Pr_{x,\bar y}(E_h).
\tag{L}
\]
Indeed, off-event equality forces the two laws to assign the same mass $q_x$ to
$E_h$.  Conditional on $E_h$, their finite label laws are affine perturbations
$p_x^\pm=p_x^0\pm\varepsilon\dot p_x$.  The interior margin and normalization of
$h$ give $|\dot p_x(z)|\le K_h p_x^0(z)$ on every affected label, uniformly in
$x$ (a zero base coefficient also has zero perturbation).  The inequality
$\log(1+u)\le u+u^2$ for sufficiently small $|u|$ gives
$\mathrm{KL}(p_x^+\|p_x^-)\le C_h\varepsilon^2$.
Multiplication by $q_x=\Pr_{x,\bar y}(E_h)$ yields (L); when $q_x=0$ it is
immediate.
The direction is \emph{locally value-relevant} if the safe-response kink diagnostic
$\gamma(h)$ in Proposition~\ref{prop:kink} is positive.  Then, for
$g(y)=\max_{x\in\Sset}x^\top\Amat y-\vref$,
$c_\varepsilon(h)=\tfrac{1}{2\varepsilon}\big(g(\bar y{+}\varepsilon h)+
g(\bar y{-}\varepsilon h)-2g(\bar y)\big)\ge\gamma(h)/2$.
\end{definition}
The body's two-type instance is the simplest case: when the two types have equal
controlled reach, equal and opposite continue/fold perturbations preserve the public
law, and a committing continuation gives a Bernoulli revealed label.

\paragraph{Proof of Theorem~\ref{thm:lower_bound}.}
Set $y^0=\bar y+\varepsilon h$ and $y^1=\bar y-\varepsilon h$.
Their conditional-action vectors are $2\varepsilon$ apart.  Define a test that selects
$y^0$ if $y^0\in C$ and $y^1$ otherwise.  Under $y^0$, joint coverage and width success
make the test correct.  Under $y^1$, the same event makes it correct because a set of
half-diameter below $\varepsilon$ cannot contain both alternatives.  Thus each testing
error is at most $\delta$.  The Bretagnolle--Huber inequality gives
\[
2\delta\ge
P^0_N(\textnormal{error})+P^1_N(\textnormal{error})
\ge\tfrac12\exp\{-\mathrm{KL}(P^0_N\,\|\,P^1_N)\},
\]
and hence $\mathrm{KL}(P^0_N\,\|\,P^1_N)\ge\log(1/(4\delta))$.

For any floor-safe one-hand plan $x$, (L) and the definition of
$\Lambda_\rok(h;\bar y)$ give
\[
\mathrm{KL}\big(\mathrm{law}^0_x\,\big\|\,\mathrm{law}^1_x\big)
\le C_h\Lambda_\rok(h;\bar y)\varepsilon^2.
\]
The probes are adaptive: $x_n$ is measurable in the data transcript
$\mathcal D_{n-1}$ and lies in $\Sset$
almost surely, so the conditional law of hand $n$ obeys the same per-hand bound pointwise, and
the KL chain rule gives
$\mathrm{KL}(P^0_N\,\|\,P^1_N)\le
C_hN\Lambda_\rok(h;\bar y)\varepsilon^2$.

Combining the two KL bounds and using $\delta\le1/8$ gives
$N=\Omega(\log(1/\delta)/(\Lambda_\rok(h;\bar y)\varepsilon^2))$.  Under directional (RF),
$\Pr_{x,y}(E_h)=\omega_x(I_0)\pi_y(I_0)$, maximizing over $x$ gives
$\Lambda_\rok(h;y)=\kcap(I_0)\pi_y(I_0)$.  When $E_h$ and $E_{I_0}$
have identical rates, this also equals $\lrate(I_0;y)$.  Because
$h_{\sigma(I)}=0$, the two lower-bound alternatives share their parent masses; the
bound does not charge for learning those denominators.  It therefore matches the
conditional-probability upper bound when parent flow is known or certified at no
slower rate. \qed

\paragraph{Reach-and-reveal mass.}
For the achievable one-information-set reveal-mass rate, $I=(H,t)$ is type-specific and
under (RF),
$\kappa_\rok^{\mathrm{rev}}(I)\pi_y(I)\le\lrate(I;y)\le\kcap(I)\pi_y(I)$,
with equality for $I\in\mathcal C_\rok$; conditional normalization additionally
uses the prefix rate $\lambda_{\rm par}(I)$.  For a direction spanning several
$I\in\mathcal I(H)$, the governing quantity is instead
$\Lambda_\rok(h;y)$ from the full observation operator.  These quantities coincide
only under the directional factorization stated in Theorem~\ref{thm:lower_bound}.
Public statistics are indexed by $H$; capacities, active pins, and their plug-in
rates remain indexed by the constituent $I$.

\begin{corollary}[Exploitation floor on value-relevant directions]
\label{cor:exploit_floor}
If the censored-fiber direction $h$ is locally value-relevant
($\gamma(h)>0$), then $c_\varepsilon(h)\varepsilon$ is a data-free modulus and
beating it requires the rate of Theorem~\ref{thm:lower_bound}.
\end{corollary}
\begin{proof}
Write $g(y)=\max_{x\in\Sset}x^\top\Amat y-\vref$ and set the Jensen gap
$2c\varepsilon:=g(y^0)+g(y^1)-2g(\bar y)\ge0$ (convexity), so $c=c_\varepsilon(h)$.  For any
floor-safe $x$ the regrets $\mathrm{reg}_b(x)=g(y^b)-x^\top\Amat y^b+\vref\ge0$ sum to
$\mathrm{reg}_0(x)+\mathrm{reg}_1(x)=g(y^0)+g(y^1)+2\vref-2x^\top\Amat\bar y\ge2c\varepsilon$, since
$x^\top\Amat\bar y\le g(\bar y)+\vref$.  For any target $s<c\varepsilon$ the good sets
$G_b=\{x:\mathrm{reg}_b(x)\le s\}$ are disjoint, so deploying within $s$ of optimal against the
realized twin with probability at least $1-\delta$ tests $y^0$ versus $y^1$ with both
errors at most $\delta$; the Bretagnolle--Huber argument above gives the same
confidence-dependent lower bound.  When $\gamma(h)=0$, $g$ is
locally linear along $h$ (Prop.~\ref{prop:kink}); for sufficiently small
$\varepsilon$, a base-optimal response is optimal against both twins, so no
de-censoring is needed.
\end{proof}

\paragraph{The proof uses only the local observation bound.}
The estimation argument uses only equality of the off-event law and the quadratic
bound (L).  Any censored-fiber direction satisfying these conditions therefore yields
Theorem~\ref{thm:lower_bound} by the identical KL/Le~Cam steps, and
Corollary~\ref{cor:exploit_floor} adds the exploitation floor wherever $\gamma(h)>0$.  We now
verify the censored-fiber conditions for the extensions the body claims.

\paragraph{Extensions (same mechanism).}
The construction touches the binary equal-prior toy only through the local bound (L), so it extends directly,
each case value-relevant where $\gamma(h)>0$ (Cor.~\ref{cor:exploit_floor}).  \emph{Unequal priors}
$p,1-p$: perturb the per-type split by $h_{t_1}=\pm(1-p)$, $h_{t_2}=\mp p$, leaving the public action
mass fixed while the revealed-type law moves by $O(\varepsilon)$, so $\mathrm{KL}\le
C(p,\alpha)\varepsilon^2$ (constants degrading as $p\to\{0,1\}$ or $\alpha\to0$).  \emph{Multiple
non-fold actions}: for $v$ with $\sum_a v_a=0$ and $v\!\restriction\!A^+(I)\neq0$, set
$h_{t_1,a}=(1-p)v_a$, $h_{t_2,a}=-p\,v_a$; the reveal law is categorical rather than Bernoulli but the
same quadratic bound holds.  \emph{Deep public-null directions}: any tangent $h$ in
the reveal-closed subtree with $Eh=0$, preserved public transcript ($B_Hh=0$),
interior base label law, and the stated off-event equality is a censored-fiber
direction.  Each yields Theorem~\ref{thm:lower_bound} by the same KL/Le~Cam steps.
If instead $B_Hh\neq0$, the public channel also carries information and the governing
rate is that of the full observation operator, outside this construction.

\begin{proposition}[LP-verifiable kink]
\label{prop:kink}
Let $G_0=g(\bar y)+\vref$ and let $F_0=\{x\in\Sset:x^\top\Amat\bar y=G_0\}$ be the base optimal
safe face.  Set $\gamma(h)=\max_{x\in F_0}x^\top\Amat h-\min_{x\in F_0}x^\top\Amat h$, two linear
programs.  If $\gamma(h)>0$ then $g$ has a directional kink at $\bar y$ along $h$ and
$c_\varepsilon(h)\ge\gamma(h)/2$ for every feasible $\varepsilon>0$; if $\gamma(h)=0$ the
base-optimal face has a common slope, $g$ is locally linear along $h$, and the bound is
value-vacuous.
\end{proposition}
\begin{proof}
Let $x^+\in F_0$ attain the max and $x^-\in F_0$ the min.  Since $x^\pm\in\Sset$ and
$(x^\pm)^\top\Amat\bar y=G_0$, feasibility in the program defining $g$ gives
$g(\bar y\pm\varepsilon h)+\vref\ge(x^\pm)^\top\Amat(\bar y\pm\varepsilon h)=G_0\pm\varepsilon
(x^\pm)^\top\Amat h$.  Adding the two and subtracting $2g(\bar y)+2\vref=2G_0$ yields
$g(\bar y{+}\varepsilon h)+g(\bar y{-}\varepsilon h)-2g(\bar y)\ge\varepsilon\big((x^+)^\top\Amat h
-(x^-)^\top\Amat h\big)=\varepsilon\gamma(h)$, i.e.\ $c_\varepsilon(h)\ge\gamma(h)/2$.  For the
converse, Danskin's theorem gives the one-sided directional derivatives
$g'(\bar y;\pm h)=\max_{x\in F_0}x^\top\Amat(\pm h)$, whose sum is
$\max_{x\in F_0}x^\top\Amat h-\min_{x\in F_0}x^\top\Amat h=\gamma(h)$.  If $\gamma(h)=0$ the two
one-sided derivatives are negatives, so $g$ is differentiable at $\bar y$ along $h$; being convex
piecewise linear it is then affine on a neighborhood of $\bar y$ along $h$, so $c_\varepsilon(h)=0$
for all small $\varepsilon$.
\end{proof}
Thus $\gamma(h)$ is an $\varepsilon$-independent response-switch diagnostic: $\gamma>0$ certifies
that the base opponent has two safe-best responses valuing the hidden split oppositely, so the
optimal safe play switches as the over-folding type flips.  Both $c_\varepsilon$ and $\gamma$ are
computed by the same safe-response LPs that produce the deployment oracle, so the diagnostic is
available per target; the Kuhn and hold'em instances below realize the full construction end to
end.

\paragraph{Scope of the directional result.}
The theorem establishes a local minimax rate for one censored-fiber direction.  A
uniform rate over all of $\mathcal C_\rok$ has a different geometry: most targets have
$c_\varepsilon=0$ away from the equilibrium base, and deeper subgames may expose the
split through another public channel.  Multi-target routing therefore forms a distinct
portfolio problem, evaluated through SAD's public-screened controller.

\paragraph{Numerical checks.}
On Kuhn poker---a private card, then fold (censoring) or call (showdown)---we perturb the two extreme
cards' per-card fold rate by $\pm\varepsilon$ and verify the proof's structure exactly (population LPs,
no sampling): the public-twin total variation and the \emph{passive} blueprint reveal-mass difference
are $0$ to machine precision for every $\varepsilon$, while the active reveal-mass difference scales as
$2\kcap\varepsilon$.  At $\rok=0.3$ the safe response switches type, giving a linear value gap
($\mathrm{gap}/\varepsilon=0.022$).  By symmetry, the group rate
$\Lambda_\rok$ is twice either constituent's per-target rate, a fixed factor, so
the conditional bounds retain the same
$\Theta((\kcap(I)\varepsilon^2)^{-1})$ dependence.  At $\rok=0.1$ the floor pins
the response and the bound is vacuous
($c=0$).  The kink is not
Kuhn-specific: on the coarse turn--river instance, eight river targets sharing two hands whose call
ends at a showdown have LP-computed $c_\varepsilon$ positive at the equilibrium base ($0.002$ to
$0.038$ over $\rok\in\{0.1,0.3,0.5\}$), each an exact public twin (total variation
$\le2\times10^{-15}$), so the value gap of Theorem~\ref{thm:lower_bound} is non-vacuous on real
hold'em targets.

\paragraph{Prevalence of value-relevant directions.}
Positive kink diagnostics concentrate near response-indifference boundaries.  Computing $\gamma(h)$ across eight public-twin targets and $31$ base
opponents (the equilibrium, $15$ simplex perturbations, $15$ independently sampled leaky opponents) at
each $\rok\in\{0.1,0.3,0.5\}$: at the equilibrium base every feasible target has $\gamma>0$, since
indifference makes the optimal safe face high-dimensional; under perturbation, the fraction of targets
with $\gamma>0$ increases from $0.15$ to $0.38$ as $\rok$ grows, and among strongly exploitable sampled opponents
($g>0.1$) only $1.4\%$ ($5$ of $357$) have $\gamma>0$.  A leaky opponent's exploitable value almost
always lies in publicly visible directions that $\Cpub$ already captures (Appendix~\ref{app:ablation});
the censored split is value-relevant on the structured, near-indifferent, high-budget subset, where
the factorized rate $\lrate$ is tight.  The estimation lower bound
(Theorem~\ref{thm:lower_bound}) needs none of
this---it holds for every censored-fiber direction regardless of $\gamma$.

\subsection{The Safe Observation Frontier}
\label{app:capacity_law}
We prove Theorem~\ref{thm:capacity_law}: the capacity is a concave piecewise-linear frontier,
supported by the floor dual at every budget, with equality on the first face.  Throughout, $\kcap(I)$ is the
value of the parametric linear program
\[
\textnormal{(P$_\rok$)}\qquad
\kcap(I)\;=\;\max_{x\in\Xpoly}\;\big\{\,\omega_x(I)\;:\;\Wval(x)\ge\vref-\rok\,\big\},
\]
where $\omega_x(I)=c_I^\top x$ is linear, $\Xpoly$ is the (compact, convex) realization polytope,
and $\Wval(x)=\min_{y\in\Ypoly}x^\top\Amat y$ is concave and piecewise linear, so the floor
constraint is a finite intersection of half-spaces and (P$_\rok$) is an LP, feasible for every
$\rok\ge0$ (the blueprint $x_{\mathrm{bp}}$, with $\Wval(x_{\mathrm{bp}})=\vref$, is feasible).

\paragraph{Monotone, concave, piecewise linear.}  If $\rok\le\rok'$ the feasible set of (P$_\rok$) is
contained in that of (P$_{\rok'}$), so $\kcap$ is non-decreasing.  For concavity, let $x_1,x_2$ be
optimal at $\rok_1,\rok_2$ and $\theta\in[0,1]$; by concavity of $\Wval$,
$\Wval(\theta x_1+(1-\theta)x_2)\ge\theta\Wval(x_1)+(1-\theta)\Wval(x_2)\ge
\vref-(\theta\rok_1+(1-\theta)\rok_2)$, so $\theta x_1+(1-\theta)x_2$ is feasible at the mixed
budget, and linearity of $\omega$ gives
$\kappa_{\theta\rok_1+(1-\theta)\rok_2}\ge\theta\kappa_{\rok_1}+(1-\theta)\kappa_{\rok_2}$.  Piecewise
linearity is the standard parametric-LP fact: the value of an LP as a function of a right-hand-side
parameter is piecewise linear, with breakpoints where the optimal basis changes and slope equal to the
constraint's dual on each piece.

\paragraph{Upper bound (Lagrangian).}  For any multiplier $\mu\ge0$ and any $x$ feasible in
(P$_\rok$), the slack $\mu(\Wval(x)-\vref+\rok)\ge0$, so
\[
\begin{aligned}
\omega_x(I)&\le\omega_x(I)+\mu(\Wval(x)-\vref+\rok)\\
&\le\max_{x'\in\Xpoly}\big[\omega_{x'}(I)+\mu(\Wval(x')-\vref)\big]+\mu\rok\\
&=:g(\mu)+\mu\rok,
\end{aligned}
\]
with $g(\mu)$ \emph{independent} of $\rok$.  Maximizing the left side gives
$\kappa_\rho(I)\le g(\mu)+\mu\rho$ for every $\mu\ge0$.  At any budget $\rho_0$,
LP strong duality gives
\[
\kappa_{\rho_0}(I)=\min_{\mu\ge0}\{g(\mu)+\mu\rho_0\}.
\]
Thus every optimal floor multiplier $\mu_I(\rho_0)$ obeys, for all $\rho\ge0$,
\[
\kappa_\rho(I)\le g(\mu_I(\rho_0))+\mu_I(\rho_0)\rho
=\kappa_{\rho_0}(I)+\mu_I(\rho_0)(\rho-\rho_0).
\]
This uses only feasibility and boundedness of the underlying LP; no Slater condition is
needed.  At $\rho_0=0$, take the smallest optimal multiplier
$\mu_I=\partial_+\kappa_\rho(I)|_{\rho=0}$ to obtain
\[
\kcap(I)\;\le\;\kappa_0(I)+\mu_I\,\rok\qquad\text{for all }\rok\ge0.
\]
Combined with the trivial bound $\kcap(I)\le\omega_{\max}(I)$, this gives the upper bound of the
theorem.

\paragraph{First-face equality.}  Because the frontier has finitely many affine pieces,
there is a maximal nontrivial initial affine interval
$[0,\rok_1]$, with $\rok_1>0$ and possibly $\rok_1=\infty$ when the frontier is
already saturated.  Its slope is the right derivative
$\mu_I=\partial_+\kcap(I)|_{0}$, equivalently the smallest optimal origin
multiplier when the dual is non-unique.  Hence
$\kcap(I)=\kappa_0(I)+\mu_I\rok$ throughout this interval.  Beyond it, concavity
makes successive value slopes non-increasing, so the frontier remains weakly below
the origin supporting line and eventually reaches $\omega_{\max}(I)$.  Degenerate
basis changes with the same value slope need not end the affine interval.  This proves
Theorem~\ref{thm:capacity_law}. \qed

\begin{lemma}[Blueprint shift]
\label{lem:bp_shift}
For every $I$, $\rok\ge0$, and $\varepsilon\ge0$,
\[
\kappa_{\rok}(I;v^\star-\varepsilon)
=\kappa_{\rok+\varepsilon}(I;v^\star).
\]
\end{lemma}
\begin{proof}
Both sides optimize over the same safe set because
$v^\star-\varepsilon-\rok=v^\star-(\rok+\varepsilon)$.
\end{proof}

\begin{corollary}[Safe certification cost at fixed local support]
\label{cor:price_of_safety}
For every censored-fiber direction $h$, Theorem~\ref{thm:lower_bound} gives
\[
N=\Omega\!\left(
\frac{\log(1/\delta)}
{\Lambda_\rok(h;\ystar)\varepsilon^2}
\right).
\]
A reveal-forcing plan at anchor $I$ attains the constructive rate
$\kappa^{\mathrm{rev}}_\rok(I)\pi_{\ystar}(I)$.  If the informative event
satisfies full-plan directional factorization,
$\Pr_{x,\ystar}(E_h)=\omega_x(I)\pi_{\ystar}(I)$ for every $x\in\Sset$, then
$\Lambda_\rok(h;\ystar)=\kcap(I)\pi_{\ystar}(I)$.  When parent flow is known or
certified at no slower rate and
$\kappa^{\mathrm{rev}}_\rok(I)=\kcap(I)$, the bounds match:
\[
N=\tilde\Theta\!\left(
\frac{1}{\kcap(I)\pi_{\ystar}(I)\varepsilon^2}
\right).
\]
All hidden constants here depend on the fixed interior label law; this is not
a uniform scaling statement as $\pi_{\ystar}(I)$ approaches the support boundary.
On a shared first face this replaces $\kcap(I)$ by
$\kappa_0(I)+\mu_I\rok$; for a wall target it gives
$N=\tilde\Theta(1/(\mu_I\rok\pi_{\ystar}(I)\varepsilon^2))$ under the same
parent and fixed-support conditions.
Without capacity equality, $\kappa^{\mathrm{rev}}_\rok\pi_{\ystar}$ remains an
achievable construction, while $\Lambda_\rok$ governs the unavoidable cost.
\end{corollary}

\paragraph{Domain-generality and its one structural condition.}  The frontier result uses only a
convex policy set, a concave floor, and a linear reach functional; nothing is specific to sequence
form.  The \emph{coupling} to the statistical rate needs (RF) and one more condition:
$\kappa_\rho^{\mathrm{rev}}(I)=\kcap(I)$ on the budget interval, so the same
frontier prices reach and information; conservative-bandit and constrained-MDP instances verify both the
frontier and the coupling numerically (Appendix~\ref{app:nongame}).  Showdown censoring is the
instance where the floor is a worst-case exploitability value and the certificate makes reveal
equal reach.

\paragraph{When reach and reveal decouple.}
Corollary~\ref{cor:price_of_safety} matches the reach and reveal frontiers only
under (RF) and capacity equality; otherwise the reveal frontier prices the explicit
reveal-forcing construction rather than an information-theoretic optimum.  Let the reveal
suffix be encoded by the flow equalities $H_Ix=0$, with safe reveal capacity
$\kappa^{\mathrm{rev}}_\rok=\max\{\omega_x(I):x\in\Sset,\ H_Ix=0\}$ and floor dual
$\mu^{\mathrm{rev}}_I=\left.\partial^+_\rho
\kappa^{\mathrm{rev}}_\rho(I)\right|_{\rho=0}$, alongside the reach program with capacity
$\kcap(I)$ and dual $\mu_I$.

\begin{proposition}[Reveal-forcing frontier]
\label{prop:decoupling}
Under (RF), suppose the reveal-forcing LP is feasible on a neighborhood of
$\rho=0$.  For a reveal-wall target with
$\kappa^{\mathrm{rev}}_0(I)=0$, a supported revealed
realization-mass coordinate can be estimated to half-width $\varepsilon$ on
its first face after
\[
N=\tilde O\!\left(
\frac{1}{\mu^{\mathrm{rev}}_I\rok\pi\varepsilon^2}
\right).
\]
For a conditional-action certificate, this rate applies when the parent is
known or certified no more slowly; otherwise the target--parent joint rate
of Theorem~\ref{thm:capacity} governs.
If $\kappa_\rho^{\mathrm{rev}}(I)=\kcap(I)$ on a neighborhood of zero, then
$\mu^{\mathrm{rev}}_I=\mu_I$; under the full-plan factorization and parent
conditions of Corollary~\ref{cor:price_of_safety}, this upper rate is matching.
If the zero-budget safe set attains the unconstrained maximum reveal-forcing
reach, then $\mu^{\mathrm{rev}}_I=0$.
\end{proposition}

\begin{proof}[Proof of Proposition~\ref{prop:decoupling}]
On the first face,
$\kappa^{\mathrm{rev}}_\rok=\mu^{\mathrm{rev}}_I\rok$.  Its maximizing
reveal-forcing plan produces informative events at rate
$\kappa^{\mathrm{rev}}_\rok\pi$, so the same concentration argument as
Theorem~\ref{thm:capacity} gives the stated upper bound.  The controlled program
optimizes over $\Sset$, while the implemented
reveal program optimizes over the smaller polytope
$\Sset\cap\{x:H_Ix=0\}$ of reveal-forcing flow equalities.  Both are
parametric LPs whose only $\rok$-dependent right-hand side is the floor; each
optimal floor multiplier is therefore the corresponding frontier's
supergradient.
Equality of the two value functions on a neighborhood of zero gives equal right
derivatives and hence $\mu^{\mathrm{rev}}_I=\mu_I$.  Corollary~\ref{cor:price_of_safety}
then supplies the matching lower bound under its stated conditions.  If
$\kappa^{\mathrm{rev}}_0$ already equals the unconstrained maximum over
$\{x\in\Xpoly:H_Ix=0\}$, the reveal
frontier is locally constant and has right derivative zero.
\end{proof}

A two-arm conservative bandit makes the dichotomy concrete (verified in
Appendix~\ref{app:nongame}): with a costed target arm of security loss $L_I$ and a separate reveal arm
of loss $L_{\mathrm{rev}}>0$, the certificate costs
$N\rok=\sigma^2/(\Delta^2\mu^{\mathrm{rev}}_I)$ with
$\mu^{\mathrm{rev}}_I=1/L_{\mathrm{rev}}$; coupling
($\mu^{\mathrm{rev}}_I=\mu_I$) holds exactly when
$L_{\mathrm{rev}}=L_I$.  If $L_{\mathrm{rev}}=0$, the reveal arm has positive
zero-budget capacity, so the reveal-wall formula does not apply.  For globally certified last-decision showdown targets,
the reveal equalities add no restriction, so $L_{\mathrm{rev}}=L_I$ holds by
construction and one dual governs both.  An operational certificate at a single
budget establishes equality of capacity values only at that budget, not equality of
the two objectives on neighboring faces.

\subsection{Certified Value Recovery}
\paragraph{Proof of Theorem~\ref{thm:value_recovery}.}
Write
\[
\begin{aligned}
V_N&=\max_{x\in\Sset}\min_{y\in\Cid(N_{\rm tot})}x^\top\Amat y,\\
V_\infty&=\max_{x\in\Sset}\min_{y\in C_\infty^{\mathrm{id}}}x^\top\Amat y.
\end{aligned}
\]
On the simultaneous-coverage event
$C_\infty^{\mathrm{id}}\subseteq\Cid(N_{\rm tot})$, set inclusion gives $V_N\le V_\infty$.
Let $x_\infty$ attain $V_\infty$.  Every $y\in\Cid(N_{\rm tot})$ and
$z\in C_\infty^{\mathrm{id}}$ belong to $\Cid(N_{\rm tot})$, so the definition of $\eta$ and
symmetry in $(y,z)$ imply
$x_\infty^\top\Amat y\ge x_\infty^\top\Amat z-\eta$.  Taking the two minima and then
maximizing over the finite-sample response gives $V_N\ge V_\infty-\eta$.
Moreover, coverage implies
$x_N^\top\Amat\ystar\ge\min_{y\in\Cid(N_{\rm tot})}x_N^\top\Amat y=V_N$.
The response is floor-safe because $x_N\in\Sset$, independently of coverage
(Theorem~\ref{thm:safety}).  By definition,
$\Gactivepop=V_\infty-\vref$, so subtracting $\vref$ from these inequalities yields
the certificate and realized-gain statements in the theorem.

For a row-level rate consequence, let $M$ collect every finite-width public
or reveal row, scaled to a realization-mass functional, and index these rows
by $i$.  Run one fixed collection portfolio for $N_{\rm tot}$ hands.  Let
$q_i>0$ be a certified lower bound on the effective one-hand observation rate
appearing in row $i$'s confidence radius, and suppose simultaneous
concentration gives
\[
w_i\le\widetilde O\!\left(\frac{h_i}{\sqrt{q_iN_{\rm tot}}}\right).
\]
Here a public row's $q_i$ includes its pilot allocation and public visitation,
whereas an action-specific reveal row uses its own label-event rate rather
than the total reveal rate of its target; the fixed $h_i$ absorbs row scaling
and support constants.  Standard Bernoulli or multinomial concentration gives
this display for the confidence rows used here.  Under (A4), for
every $x\in\Sset$ there exist $\theta,\nu$ satisfying
$\Amat^\top x=M^\top\theta+E_2^\top\nu$: indeed, (A4) says that
$\Amat^\top x$ annihilates $\ker E_2\cap\ker M$, whose orthogonal complement is
$\operatorname{range}(E_2^\top)+\operatorname{range}(M^\top)$.  Hence, for
$d=y-y'$ with $E_2d=0$ and $|(Md)_i|\le2w_i$,
\[
\eta
\le 2\max_{x\in\Sset}
\min_{\substack{\theta,\nu:\\
\Amat^\top x=M^\top\theta+E_2^\top\nu}}
\sum_i|\theta_i|w_i
\le\widetilde O\!\left(\frac{A_h(q)}{\sqrt{N_{\rm tot}}}\right),
\]
where
\[
A_h(q):=\max_{x\in\Sset}
\min_{\substack{\theta,\nu:\\
\Amat^\top x=M^\top\theta+E_2^\top\nu}}
\sum_i\frac{|\theta_i|h_i}{\sqrt{q_i}}.
\]
In particular,
$A_h(q)\le A_1h_{\max}/\sqrt{\underline q}$, where
\[
A_1:=\max_{x\in\Sset}
\min_{\substack{\theta,\nu:\\
\Amat^\top x=M^\top\theta+E_2^\top\nu}}\|\theta\|_1 .
\]
Here $h_{\max}:=\max_i h_i$ and $\underline q:=\min_iq_i$.
This flow-aware constant accounts for uncertainty propagated from observed
children to censored coordinates.  If a dedicated collection plan for row
$i$ has certified row-specific rate $r_i>0$ and is played for share
$\alpha_i$, then other plans can only add observations, so
$q_i\ge\alpha_i r_i$.  Thus the display is a total-budget bound; assigning
$n$ hands to each of $m$ rows gives $N_{\rm tot}=mn$ and
$q_i\ge r_i/m$.  This also explains the unbucketed-river
granularity effect: per-combination pins divide the portfolio's reveal mass among
coordinates with much smaller jointly achieved rates.  Finally, the set-stability
argument uses only $C_\infty^{\mathrm{id}}\subseteq\Cid(N_{\rm tot})$ and
$\ystar\in C_\infty^{\mathrm{id}}$, so it applies unchanged to any nested population
and finite-sample uncertainty sets.  Theorem~\ref{thm:active_id} identifies the
exact-pin population value with the oracle restricted to $\mathcal R_{\rok}$; when the
residual is payoff-null, $V_\infty-\vref=\Goracle$. \qed

\begin{corollary}[Portfolio allocation]
\label{cor:portfolio_allocation}
Fix a linear flow reconstruction $T$ satisfying
$\Amat^\top x-M^\top Tx\in\operatorname{range}(E_2^\top)$ for every
$x\in\Sset$; such a map exists under (A4) by a linear right inverse on the
relevant row space.  For every finite-width row $i$, let $r_i>0$ be a
certified lower bound on that row's observation rate under its dedicated
collection plan, including any public-pilot allocation or within-target
label share.  Let $b_i=\max_{x\in\Sset}|(Tx)_i|$, retain the confidence
constant $h_i$ from the preceding display, and discard rows with $b_i=0$.
Among mixtures of these dedicated plans, the separable envelope
\[
\eta\le\widetilde O\!\left(
\frac{1}{\sqrt{N_{\rm tot}}}
\sum_i\frac{b_ih_i}{\sqrt{\alpha_i r_i}}
\right)
\]
is minimized by
\[
\alpha_i^\star=
\frac{(b_i^2h_i^2/r_i)^{1/3}}
{\sum_j(b_j^2h_j^2/r_j)^{1/3}}.
\]
The resulting envelope is
\[
\eta\le\widetilde O\!\left(
\frac{\bigl[\sum_i(b_i^2h_i^2/r_i)^{1/3}\bigr]^{3/2}}
{\sqrt{N_{\rm tot}}}
\right).
\]
\end{corollary}

\begin{proof}
Taking coordinatewise maxima in the row-level display gives
$\eta\le 2\sum_i b_iw_i$.  Since
$w_i=\widetilde O(h_i/\sqrt{q_iN_{\rm tot}})$ and
$q_i\ge\alpha_i r_i$, with $h_i$ retained in its numerator, this yields
the stated separable envelope.
Set $c_i=b_ih_i/\sqrt{r_i}$.  The strictly convex objective
$\sum_i c_i\alpha_i^{-1/2}$ on $\sum_i\alpha_i=1$ has stationarity condition
$\alpha_i\propto c_i^{2/3}$, which gives $\alpha_i^\star$.
Substitution yields
$\sum_i c_i/\sqrt{\alpha_i^\star}
=(\sum_i c_i^{2/3})^{3/2}$.
\end{proof}

The shares in Corollary~\ref{cor:portfolio_allocation} optimize the displayed separable
envelope.  The following results remove its dedicated-probe restriction and also cover
target discovery when the public screen is null.

\subsection{Universal Auditing and Joint Routing}
\label{app:joint_routing}

Fix $m$ reveal-certified targets and $K$ floor-safe probes
$x^{(1)},\ldots,x^{(K)}$.  Let $R_{ij}\ge0$ be a certified lower bound on
the one-hand informative-event probability for target $i$ under probe $j$
over the opponent class under consideration.  These bounds may incorporate
controlled reach, opponent continuation, and the reveal-forcing suffix.

\paragraph{Proof of Theorem~\ref{thm:portfolio}.}
Let $\alpha^\star$ attain the max--min objective and set
$x_{\rm aud}=\sum_j\alpha_j^\star x^{(j)}$.
Convexity of $\Sset$ makes $x_{\rm aud}$ floor-safe.  If $p_i(x)$ denotes the
actual informative-event rate, linearity of root mixing gives
$p_i(x_{\rm aud})=\sum_j\alpha_j^\star p_i(x^{(j)})
\ge(R\alpha^\star)_i\ge q_{\rm aud}$.
The max--min objective is an LP;
its equality follows from LP duality, with $u$ a distribution over targets.
The confidence radius under the audit fraction is at most
$C\sqrt{\log(2m/\delta)/(\beta Nq_{\rm aud})}$.  The displayed strict
inequality makes this radius smaller than $\Delta_{\rm sig}/2$, so the reference lies
outside the interval of every $\Delta_{\rm sig}$-separated coordinate.  The
simultaneous confidence event gives the stated probability. \qed

\begin{corollary}[Public-null target discovery]
\label{cor:public_null_audit}
The discovery guarantee of Theorem~\ref{thm:portfolio} requires no public
anomaly.  Thus a predeclared public-null coordinate is detected whenever it is
$\Delta_{\rm sig}$-separated and $q_{\rm aud}>0$.  When $q_{\rm aud}=0$, the library and
certified lower bounds do not establish uniform discovery; impossibility follows
only when $R$ is exact or its zero entries are certified support impossibilities.
\end{corollary}
\begin{proof}
Apply Theorem~\ref{thm:portfolio} to the predeclared target set; its confidence
test uses reveal coordinates rather than the public screen. \qed
\end{proof}

The audit objective protects the least-observed target.  Once target
importance weights are available, the full cross-reveal matrix yields a
different globally solvable allocation.

\begin{proposition}[Globally optimal flow-aware width for a certified rate matrix]
\label{thm:joint_routing}
Let $a_i>0$, assume some mixture has $(R\alpha)_i>0$ for every $i$, and set
\[
\Phi_R(\alpha):=\sum_{i=1}^m\frac{a_i}{\sqrt{(R\alpha)_i}},
\qquad \alpha\in\Delta_K.
\]
Set $\Phi_R(\alpha)=+\infty$ if any $(R\alpha)_i=0$.
Then $\Phi_R$ is convex, its optimal reveal vector $q^\star=R\alpha^\star$
is unique, and an optimal mixture exists with at most $m$ probes; the
optimizing mixture need not be unique.  A
covered mixture $\alpha$ with $q=R\alpha$ is optimal if and only if, with
$s_j=\sum_i a_iR_{ij}/q_i^{3/2}$,
\[
\begin{aligned}
s_j&\le\Phi_R(\alpha) &&\text{for every }j,\\
s_j&=\Phi_R(\alpha) &&\text{when }\alpha_j>0.
\end{aligned}
\]
The optimal value has the exact dual representation
\[
\min_{\alpha\in\Delta_K}\Phi_R(\alpha)
=
\left[
\max_{\substack{u\ge0\\R^\top u\le\mathbf 1}}
\sum_{i=1}^m a_i^{2/3}u_i^{1/3}
\right]^{3/2}.
\]
More generally, if each reveal rate is a nonnegative linear functional
$r_i^\top x$ on $\Sset$ and these are the rates used by the confidence
bound, then
\[
\min_{x\in\Sset}\sum_i\frac{a_i}{\sqrt{r_i^\top x}}
\]
is a convex program.  At a covered iterate $x$, its Frank--Wolfe linear
oracle is the safe weighted-reach LP
\[
\max_{z\in\Sset}\sum_i
\frac{a_i\,r_i^\top z}{(r_i^\top x)^{3/2}},
\]
and
\[
G(x):=\frac12\left[
\max_{z\in\Sset}\sum_i
\frac{a_i\,r_i^\top z}{(r_i^\top x)^{3/2}}
-\sum_i\frac{a_i}{\sqrt{r_i^\top x}}
\right]
\]
satisfies $0\le\Phi(x)-\Phi^\star\le G(x)$.  An additive-$\epsilon$
upper error in the linear oracle adds $\epsilon/2$ to this certificate.
\end{proposition}

\begin{proof}
The Hessian of the finite-library objective is
\[
\nabla^2\Phi_R(\alpha)
=\frac34R^\top
\operatorname{diag}\!\left(
\frac{a_i}{(R\alpha)_i^{5/2}}
\right)R\succeq0.
\]
The extended objective is lower semicontinuous on the compact polytope
$R\Delta_K$, so the positive-coverage assumption gives an optimum.
Strict convexity in the positive reveal vector gives uniqueness of $q^\star$;
The simplex KKT conditions give the displayed inequalities,
because
$\nabla_j\Phi_R=-s_j/2$ and
$\sum_j\alpha_j^\star s_j=\Phi_R(\alpha^\star)$.
Writing $c_i=a_i/(q_i^\star)^{3/2}$, every column used by an optimal mixture
lies in the supporting hyperplane
$c^\top R_{\cdot j}=\Phi_R(\alpha^\star)$.  Thus $q^\star$ lies in the convex
hull of columns in an affine space of dimension at most $m-1$, and
Carath\'eodory's theorem gives an optimal mixture supported on at most $m$
probes.

For any $u\ge0$ with $R^\top u\le\mathbf1$, every feasible
$q=R\alpha$ satisfies $u^\top q\le1$.  Direct minimization over $q>0$
under this single constraint gives
\[
\inf_{\substack{q>0\\u^\top q\le1}}
\sum_i a_iq_i^{-1/2}
=\left(\sum_i a_i^{2/3}u_i^{1/3}\right)^{3/2},
\]
so the dual display is a lower bound.  At $q^\star$, set
$u_i^\star=a_i/[\Phi_R(\alpha^\star)(q_i^\star)^{3/2}]$.
The KKT inequalities imply $R^\top u^\star\le\mathbf1$, and substitution
makes the lower bound equal $\Phi_R(\alpha^\star)$, proving strong duality.
For the continuous safe set, composition of the convex decreasing map
$q\mapsto\sum_i a_iq_i^{-1/2}$ with linear positive rates gives convexity.
Its negative gradient supplies the stated linear oracle; the standard
Frank--Wolfe gap bounds global suboptimality by convexity.
\end{proof}

For the fixed flow reconstruction in
Corollary~\ref{cor:portfolio_allocation}, this objective is its full-matrix form.
If the coordinate confidence widths
satisfy $w_i\le h_i/\sqrt{Nq_i}$, with simultaneous-coverage factors absorbed
into $h_i$, then setting $a_i=2b_ih_i$ gives
$\eta\le N^{-1/2}\Phi_R(\alpha)$.  Thus the diagonal allocation is recovered
when probe $j$ reveals only coordinate $j$, whereas general $R$ prices all
cross-reveals.

\section{Confidence Sets under Censoring}
\label{app:pubconf}

\subsection{Public Confidence Set}
\label{app:c_pub}
For a public class $S$ and action $a$, the reach-weighted public frequency is
\begin{equation}
\pi_y^\omega(a\mid S)=
\frac{\sum_{I\in\mathcal I(S)}\omega_r(I)y_{\sigma(I)a}}
     {\sum_{I\in\mathcal I(S)}\omega_r(I)y_{\sigma(I)}} ,
\label{eq:pubfreq}
\end{equation}
where $\omega_r(I)$ is chance times the agent's reach, known to the agent
(reach-weighted evaluation in the tradition of \citealt{bowling2008strategy}).  A confidence
interval $[\ell,u]$ on \eqref{eq:pubfreq} linearizes into the public constraints
\begin{equation}
\ell\sum_I\omega_r(I)y_{\sigma(I)}
\le \sum_I\omega_r(I)y_{\sigma(I)a}
\le u\sum_I\omega_r(I)y_{\sigma(I)} ,
\label{eq:pub}
\end{equation}
and stacking them over public classes and actions yields the public confidence set
$\Cpub$.  Folds are public, so $\Cpub$ constrains every action's aggregate frequency;
what it cannot resolve is how a censored fold mass is distributed across the hidden types
behind a public state.

In the implemented batch controller, a fixed blueprint generates a pilot of
$N_P$ independent hands.  If public class $S$ is visited $n_S$ times and action $a$
is observed $C_{S,a}$ times, then
$\widehat\pi^\omega(a\mid S)=C_{S,a}/n_S$; conditional on $n_S$, this is a Bernoulli
mean with expectation \eqref{eq:pubfreq}.  When $n_S=0$ its interval is $[0,1]$.

\subsection{Active Reveal Confidence Set}
\label{app:c_obs}
The active set adds the reveal constraints to $\Cpub$, $\Cid=\Cpub\cap\{\text{reveal
constraints}\}$, and we use it in two forms.  The \emph{population} object pins each
showdown-revealed non-fold mass $y_{\sigma(I)a}$ with positive portfolio reveal rate
by its exact value (zero width).  The ladder (Table~\ref{tab:ladder}) is the
conditional full-coverage object: an ideal portfolio covers every relevant
type-specific $I\in\mathcal R_\rok$, and its restricted-oracle value numerically
equals $G^\star$ on the tested families by pooling reveal support across that
portfolio.
The \emph{finite-sample} object uses a reveal plan frozen after the pilot and a
disjoint batch of $N_R$ hands.  For each committing revealed label
$i=(I,a)$, define
\[
Z_{k,i}=\mathbf 1\{\text{hand }k\text{ reveals label }i\}.
\]
Here revealing label $i$ means reaching $I$, observing action $a$, and completing
the type-revealing suffix.
The variables $Z_{1,i},\ldots,Z_{N_R,i}$ are Bernoulli with
\[
\mathbb E Z_{k,i}=m_i=w_i y_{\sigma(I)a},
\qquad
\widehat m_i=\frac{1}{N_R}\sum_{k=1}^{N_R}Z_{k,i},
\]
where $w_i$ is the known chance-and-agent showdown-reach coefficient.  Thus the
denominator is every hand in the reveal batch, not only successful showdowns.
An interval $[\ell_i,u_i]$ for $m_i$ becomes the two linear rows
$\ell_i\le w_i y_{\sigma(I)a}\le u_i$.

For either public or reveal Bernoulli mean $\widehat q$ from $n\ge1$ samples, the
implementation uses the two-sided empirical-Bernstein radius
\begin{equation}
\begin{aligned}
b_n(\widehat q,\delta_j)
&=\sqrt{\frac{2\widehat v\log(3/\delta_j)}{n}}
+\frac{3\log(3/\delta_j)}{n},\\
\widehat v&=\widehat q(1-\widehat q),
\end{aligned}
\label{eq:eb_radius}
\end{equation}
clipped to $[0,1]$; for $n=0$ it uses $[0,1]$.  This is Theorem~1 of
\citet{audibert2009exploration} with $x=\log(3/\delta_j)$ and unit range.

\begin{proposition}[Fixed-batch simultaneous coverage]
\label{prop:fixed_batch_coverage}
Suppose \textnormal{(A1)--(A2)} hold.  Let $J_P$ and $J_R$ be the numbers of public
and reveal intervals.  Before observing the pilot, a public-only arm assigns
$\delta$ to its public family, while every active arm reserves $\delta/2$ for
public rows and $\delta/2$ for reveal rows; each reservation is divided uniformly
within its family and any unused share remains unused.  If the reveal plan and its
$J_R$ labels are frozen from the independent pilot before the reveal batch is sampled, the
intersection of \eqref{eq:pub} and the reveal rows contains $\ystar$ with
probability at least $1-\delta$.  On the same event, the corresponding
exact-row population set is contained in the finite-sample set:
$C_\infty^{\mathrm{id}}\subseteq\Cid$.
\end{proposition}
\begin{proof}
Conditional on each public visit count, \eqref{eq:eb_radius} covers its Bernoulli
action probability with failure at most its preassigned $\delta_j$.  Conditional on
the pilot, the reveal plan and labels are fixed and each $\widehat m_i$ is the
Bernoulli mean above, so the same bound covers $m_i$.  A union bound within and
across the pre-budgeted families gives total failure probability at most $\delta$.
Substitution into the linear rows proves membership of $\ystar$.  Every
$y\in C_\infty^{\mathrm{id}}$ has the same frozen public and reveal row values
as $\ystar$, so it satisfies those covered intervals as well, proving the
set inclusion. \qed
\end{proof}

By construction $\Cid\subseteq\Cpub$, so active evidence never enlarges the feasible
opponent region.  The sample split is what permits data-dependent target selection
without post-selection confidence correction.  Under a fixed portfolio with
informative-event rate $q_i$, the successful count is
$\Theta(q_iN_R)$ with high probability and conditional-action error scales as
$\widetilde O(1/\sqrt{q_iN_R})$, matching Theorem~\ref{thm:capacity}.

\paragraph{Calibration diagnostic.}
The weighting and union bounds are empirically necessary.  On a heads-up no-limit hold'em
river endgame under showdown monitoring (60 rounds of 200 episodes, 25 seeds), the
reach-weighted, union-bounded public set holds anytime coverage $1.000$ across all four
diagnostic opponents and levels $\delta\in\{0.05,0.1,0.2\}$.  Dropping the reach weighting
collapses coverage to $0.04$--$0.12$ on the equilibrium and censored-fold opponents (where
showdown reach is most uneven); dropping the union bound degrades it as $\delta$ tightens
($0.96$ at $\delta{=}0.05$ to $0.64$ at $\delta{=}0.20$).  These failures fall on the
low-reach tail the MNAR effect (Theorem~\ref{thm:passive_mnar}) makes hardest.

\section{Algorithms}
\label{app:algos}

\subsection{Safe Active De-censoring Loop}
\label{app:alg:etc}
Algorithm~\ref{alg:sad} gives one acquisition batch.  Every policy it forms is solved
inside $\Sset$, so the floor holds independently of all data and screening decisions.

\paragraph{Simultaneous public screen.}
For each public history $H$, let $\hat p_H$ be its pilot action frequencies and
$p_H^{\rm bp}$ the known blueprint reference.  Conditional on $n_H$ visits, a union
bound over the $M=\sum_H|A_H|$ action coordinates gives
\[
\lVert \hat p_H-p_H\rVert_\infty
\le r_H:=\sqrt{\frac{\log(2M/\delta_r)}{2n_H}}
\quad\text{for every }H
\]
with probability at least $1-\delta_r$.  The triangle inequality then yields
\[
\tfrac12\lVert p_H-p_H^{\rm bp}\rVert_1
\ge \tfrac12\lVert\hat p_H-p_H^{\rm bp}\rVert_1
-\tfrac{|A_H|}{2}r_H.
\]
SAD routes only histories with a positive lower bound.  Under the reference opponent,
no history survives on the simultaneous event; under an alternative, the positive
remainder weights every reveal-certified private target behind that public history.
The screen uses no hidden label or opponent realization.

\paragraph{Matched acquisition budget.}
Every empirical cell is assigned a total budget $N$.  Public-only collection uses all
$N$ blueprint hands.  An active arm reserves $0.2N$ blueprint hands for the public
pilot, freezes its route, and draws the remaining $0.8N$ hands under the resulting
floor-safe portfolio.  Random probing uses the same split; oracle-targeting is reported only
as a population-assisted upper reference.  The pilot alone defines $\Cpub$ with known
blueprint reach weights.  The independent reveal batch estimates the joint
reach-and-reveal masses of committing non-fold sequences, from which flow recovers fold
mass.  One global confidence budget is split across the public and reveal families and
then across their statistical rows.

Each library component has a fixed certificate mask $M_{Ij}$, and
$\Omega_{Ij}=M_{Ij}\omega_{x_j}(I)$; a component receives no routing credit for
uncertified reach.  Singleton certificates always supply a valid fallback library,
although the linear screened objective may select only one component.  Universal
coverage comes from the separately reserved max--min audit share, not from that
screened objective.  The coarse-endgame rows commit to showdown under one compatible
weighted plan, so its mask covers every routed row.
Routing and response computation use the assigned empirical streams.  If the joint
empirical set is empty, deployment falls back to its empirical public set and then to
the blueprint.  The evaluated SAD route uses the screened public anomaly as its value
weight and maximizes weighted controlled reach in $\Sset$.  Capacity separately prices
the attainable observation rate; the reported capacity-only arm in
App.~\ref{app:ablation} isolates that role.

\paragraph{Hyperparameters.}
The matched-budget experiment fixes confidence level $\delta=0.1$, screening level
$\delta_r=0.1$, safety budget $\rho=0.5$, public fraction $0.2$, and probe tie-break
weight $B=10^6$ over ten matched seeds.  A target enters the route only when its screened
weight exceeds $10^{-6}$.  These quantities are fixed across opponents and budgets.

\subsection{Public-Null Auditing and Joint Routing}
\label{app:joint_algorithms}

When the public screen selects no target, a data-independent audit mixture can preserve
discovery.  For a fixed safe probe library, we form its certified cross-reveal matrix
$R$, setting $R_{Ij}=0$ whenever component $j$ lacks a reveal certificate for target
$I$, and solve
$\max_{\alpha\in\Delta_K}\min_i(R\alpha)_i$.  A direct LP gives the exact solution for
an enumerated library.  For an implicit safe set, Hedge maintains a distribution over
targets and calls the safe weighted-reach LP oracle at each iteration; averaging the
returned realization plans gives a feasible primal plan, while the smallest weighted
oracle value gives a dual upper bound.  In exact arithmetic these quantities sandwich
the max--min audit rate; the implementation reports their binary64 numerical
counterparts.  The root audit tag is sampled privately before each hand.  Thus a mixture
component may be analyzed as its own reveal-certified stratum while all blueprint
filler hands remain charged to the acquisition budget; this stratified interpretation
requires a stationary opponent that cannot condition on the tag.

Fixed positive row scales $s_i$ may normalize heterogeneous targets by replacing
$R_{ij}$ with $R_{ij}/s_i$.  A normalized minimum $\gamma$ then certifies physical
rate at least $s_i\gamma$ for target $i$; confidence radii use these physical rates.
The four-target benchmark uses solo reveal-rate upper enclosures as $s_i$.

The screened linear objective, the max--min audit, and the following width objective
are distinct allocations.  In the joint-routing benchmark, a fixed audit share is
reserved first and the remaining share is optimized for width.  Let
$q_i=r_i^\top x$ be the complete reveal rate of target
$i$, including reveals contributed by probes aimed elsewhere.  We minimize
$\sum_i a_i/\sqrt{q_i}$ by Frank--Wolfe.  At iteration $t$, the oracle weights are
$a_i/q_{i,t}^{3/2}$; one weighted-reach LP over $\Sset$ supplies the search
vertex, and one-dimensional line search updates the mixture.  We reserve a fixed audit
share before this optimization, so every iterate retains positive coverage.  The final
Frank--Wolfe residual is recomputed with a fresh oracle call.  In exact arithmetic it
is a global suboptimality bound for the fixed-audit objective; the implementation
reports its numerical value using binary64 outer arithmetic.  For each oracle call, the
safety-dual opponent is quantized to an exactly normalized dyadic behavioral
plan.  Outward rounding encloses its realization weights, the payoff products,
and the player treeplex backward induction, producing a floating-point upper
enclosure of the exact Lagrangian bound for the supplied binary64 weights.  A
$H_Ix=0$ reveal program uses the same recursion with each specified child fixed
rather than maximizing over its siblings, so the enclosure is taken over the
same constrained treeplex.  A
$10^{-5}$ blueprint share interiorizes each returned search vertex.  Each oracle LP imposes the
worst-case floor; independent best responses additionally check the solo probes,
audit, separable baseline, and returned joint mixture.

\subsection{Evidence Updates}
The primary protocol keeps the evidence streams disjoint.  The \emph{public stream}
contains the blueprint pilot and drives both $\Cpub$ and the frozen route.  The
\emph{reveal stream} contains only the subsequent probe hands; it records a private type
when a non-fold continuation reaches showdown and estimates the corresponding
realization mass with the total reveal-batch denominator.  Public rows are keyed by
$H$, whereas reveal rows remain keyed by $I\in\mathcal I(H)$.  The agent knows the
collection policy and therefore the controlled reach coefficients exactly.

\subsection{Verification}
\label{app:certaudit}
The certified floor is audited directly: for a deployed plan $x$ the worst-case value
$\Wval(x)=\min_{y\in\Ypoly}x^\top\Amat y$ is the opponent's best response to a fixed $x$,
computed by a deterministic bottom-up pass over the opponent sequence form---no linear
program and no confidence set are involved---so every reported plan is checked against
$\Wval(x)\ge\vref-\rok$ independently of the model used to form it.  The implementation
uses binary64 arithmetic and records a violation only below
$\vref-\rok-10^{-6}$.  Across the $574$ floor-safe controller and ablation
records, the minimum computed response slack is $-6.53{\times}10^{-8}$ and the minimum
acquisition-policy slack is $-3.20{\times}10^{-13}$; hence every record passes
the stated numerical audit.  We use
\emph{numerically floor-certified to $10^{-6}$} for this deterministic,
model-independent audit and
\emph{statistically certified} for the high-probability exploitation lower bound a robust
response attains over a confidence set (valid at level $1-\delta$ under coverage); every
deployed plan is numerically floor-certified to $10^{-6}$ regardless of model correctness,
while its exploitation
value is statistically certified only when the confidence set covers $\ystar$.  The exploitation
plans themselves are the robust responses $\max_{x\in\Sset}\min_{y\in C}x^\top\Amat y$ over
$C\in\{\Cpub,\Cid\}$, each one sequence-form LP (HiGHS) by strong duality: writing
$C=\{y\in\Ypoly:Gy\le h\}$ and $\Ypoly=\{y\ge0:E_2y=e_2\}$,
\[
\begin{gathered}
\max_{x\in\Xpoly,\,z,\,\eta,\,\lambda\ge0,\,\nu}\ z
\quad\text{s.t.}\quad
z\le e_2^\top\eta-h^\top\lambda,\\
E_2^\top\eta-G^\top\lambda\le \Amat^\top x,\quad
E_2^\top\nu\le \Amat^\top x,\quad
e_2^\top\nu\ge \vref-\rok,
\end{gathered}
\]
whose last two rows are the floor's dual certificate.  Reveal-forcing policies
are audited per target by the same machinery---the reach LP for
$\kcap(I)$ or constrained reveal LP for $\kappa^{\mathrm{rev}}_\rok(I)$,
a structural reveal-forcing-suffix check, and the floor verifier on the
resulting plan.  Of the $60$ reported targets, $48$ are direct-showdown lines
with $\kappa^{\mathrm{rev}}_\rok=\kcap$ and hence satisfy
Def.~\ref{def:class}; the other $12$ use whole-information-set
call equalities on every chance-compatible response information set and are
reported only with $\kappa^{\mathrm{rev}}_\rok$.  All fixed-action, showdown-reach,
optimality-enclosure, and floor checks pass.  The end-to-end
controller instead assigns weights over $2160$ type-specific river labels ($2158$ have
positive capacity).  Its active set admits only showdown-committing sequence rows with
positive controlled reach; sequence-form flow then recovers the associated censored
mass.  Thus the routing-label count is neither the number of statistical rows nor an
additional per-label certificate claim.  In the
primary $N=10^6$ matched-budget grid, all $164$ probe and response cells pass the
floor check; the two lower-budget runs bring the unique total to $404$, with no
confidence-set infeasibility or timeout.  The $30$ reported-protocol capacity-only
cells also pass every floor and feasibility check.  All $488$ plans in the
complementary population-assisted routing grid also pass.  When a finite-sample
confidence set is empty---the naive $\Cpub\cap\Cid$ box can be infeasible under the MNAR
selection effect---the solver reports infeasibility and the controller falls back to the
feasible public set.

\section{Additional Experiments}
\label{app:moreexp}

This appendix collects supporting diagnostics for Safe Active De-censoring.

\subsection{Active Confidence-Set Coverage}
\label{app:coverage}
Theorem~\ref{thm:value_recovery} assumes the active set $\Cid(N_{\rm tot})$ contains $\ystar$ with
high probability and has payoff-relevant width $\eta\to0$.  Figure~\ref{fig:cid_coverage}
measures both directly on the coarse endgame, sample-split so that targets are chosen from
public statistics independently of the showdown counts that form the intervals.  Panel~(a)
reports the \emph{per-coordinate} containment fraction---the share of pinned coordinates whose
interval traps the truth; since $\Cid$ is the intersection of these coordinate intervals over
the spatial union bound of Appendix~\ref{app:pubconf}, every-coordinate containment in a cell
implies $C_\infty^{\mathrm{id}}\subseteq\Cid(N_{\rm tot})$ and hence
$\ystar\in\Cid(N_{\rm tot})$ there.  The plotted $N$ is this total batch size.
The active set holds at $1$ in every cell at every budget
(solid; the three families coincide), so the set covers $\ystar$ throughout, whereas the
passive showdown box (dashed) \emph{loses} coordinate containment as $N$ grows---from full at
$N\le10^4$ to $0.69$--$0.74$ at $N=10^6$.  Coverage that
degrades with more data is the finite-sample signature of
Theorem~\ref{thm:passive_mnar}: the passive intervals concentrate around the
showdown-conditioned estimand, which excludes the truth.  Panel~(b) shows the active
interval width shrinking at the $1/\sqrt N$ rate (guide line), so $\eta\to0$ as the theorem
requires.  Together the panels support the coverage and width hypotheses of the
certified-recovery theorem, and the passive curve is why the active protocol is needed to
meet them.
\begin{figure*}[t]
\centering
\includegraphics[width=0.82\textwidth]{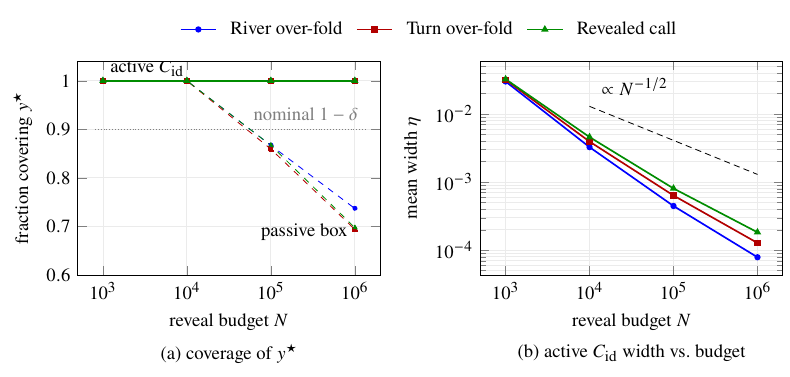}
\caption{Finite-sample behavior of the active confidence set $\Cid$ on the coarse
turn--river endgame ($\rho=0.5$, $\delta=0.1$, $5$ matched seeds).  (a)~Fraction of
pinned coordinates whose interval contains the truth $\ystar$ (per-coordinate containment; the
set covers $\ystar$ in a cell when this is $1$, via the union bound of
Appendix~\ref{app:pubconf}): the active set (solid) holds
at $1$ for all three deviating families and budgets, while the passive showdown box (dashed)
loses containment as $N$ grows.  (b)~Active interval width versus budget, shrinking at
the $1/\sqrt N$ rate.  Together these support the coverage and width hypotheses of
Theorem~\ref{thm:value_recovery}.}
\label{fig:cid_coverage}
\end{figure*}

\subsection{A Necessity Counterexample}
\label{app:necessity}
Proposition~\ref{prop:necessity} shows that (A3) cannot be omitted from a general
point-identification claim.
Table~\ref{tab:necessity} instantiates the proof's minimal two-type construction:
without reveal-controllability the two opponents have monitoring gap $0$ and remain in
one nontrivial fiber; adding a forced showdown separates them and collapses that fiber,
recovering Theorem~\ref{thm:active_id}'s identification.
\begin{table}[t]
\centering\small

\begin{tabular}{lrr}
\toprule
Setting & Mon.\ gap & Fiber diam. \\
\midrule
Non-reveal-controllable        & $0.000$ & $1.000$ \\
Reveal-controllable (restored) & $1.000$ & $0.000$ \\
\bottomrule
\end{tabular}
\caption{Point-identification boundary in the two-type construction of
Proposition~\ref{prop:necessity}.  \emph{Monitoring gap} is total variation between
the public-action, showdown-label, and payoff laws; \emph{fiber diameter} is total
variation between opponent plans still compatible with one monitored law.}
\label{tab:necessity}
\end{table}

\subsection{Residual Safe-Payoff Width Beyond A3}
\label{app:residual}
Reveal-controllability is sufficient but not necessary (Prop.~\ref{prop:payoff_boundary}):
when forcing is only partial, the operative quantity is the one-sided safe-payoff width
$\Delta_M^-(\ystar)$, which bounds the exploitation lost to the residual fiber
(Prop.~\ref{thm:residual_recovery}) and is computed exactly by two linear programs---an
$\Sset$-constrained best response to $\ystar$ and the robust response over the observation
fiber $\mathcal F_M(\ystar)$.  We audit it directly on the bucketed turn--river endgame
($5{,}221$ sequences, $\rok=0.5$) against four deployment opponents.  Forcing every
non-fold continuation collapses the fiber to a point and recovers the safe oracle value
$V(\ystar)=\max_{x\in\Sset}x^\top\Amat\ystar$---the point-identification regime of
Theorem~\ref{thm:active_id}---while leaving continuations unforced one at a time prices each
per-continuation shortfall $V-R_M$ by a single fiber LP.  The $2627$ continuations with no
fold sibling are payoff-null by flow and need no LP; the remaining $1198$ fold-adjacent
continuations are each priced exactly.

Table~\ref{tab:residual_width} reports the result: for every leak $92$--$99\%$ of
continuations are payoff-null---leaving them unforced changes the certified value by
nothing---so the quotient A3 must force is a small, identifiable fraction of the tree
(the equilibrium control is uniformly null).  The residual concentrates with censoring
depth (the deep turn over-fold carries $52\%$ of it on its top eight continuations), yet
forcing \emph{only} the top-$k$ recovers essentially nothing: the robust value is set by
the worst \emph{unforced} value-relevant direction, so the quotient must be forced as a
whole.  It is also \emph{structurally localized}
(Table~\ref{tab:residual_localize}): the value-relevant continuations collapse onto
$3$--$18$ public lines with $90\%$ of the mass on $2$--$7$ of them, identifiable in
advance---so ``reveal-controllable after every continuation'' becomes ``force a few
identifiable lines,'' with the bounded residual of Prop.~\ref{prop:payoff_boundary} on
any line left unforced.
\begin{table}[t]
\centering\small

\setlength{\tabcolsep}{2.3pt}
\begin{tabular}{@{}lrrrrr@{}}
\toprule
Opponent & $V{-}\vref$ & val.-rel. & \%\,null & worst\,$\Delta^-$ & \%\,expl \\
\midrule
Equilibrium (control) & $0.000$ & $0$   & $99.9$ & $0.000$ & $0.0$ \\
River over-fold       & $0.595$ & $300$ & $91.8$ & $0.021$ & $3.5$ \\
Revealed call         & $1.434$ & $213$ & $94.4$ & $0.040$ & $2.8$ \\
Turn over-fold (deep) & $0.517$ & $49$  & $98.7$ & $0.062$ & $12.0$ \\
\bottomrule
\end{tabular}
\caption{Residual safe-payoff width on the bucketed turn--river endgame
($5{,}221$ sequences, $\rok=0.5$), per deployment opponent.  $V{-}\vref$ is the safe
exploitation available; \emph{val.-rel.}\ counts continuations with positive single-drop
shortfall (of $3825$ non-fold continuations; the rest are payoff-null, $2627$ by flow
alone); worst\,$\Delta^-$ is the largest single-continuation one-sided width
(Prop.~\ref{prop:payoff_boundary}), absolute and as \% of exploitation.  The equilibrium
control carries no exploitable value and is uniformly null; across genuine leaks
$92$--$99\%$ of continuations are payoff-null, so the payoff-relevant quotient A3 must
force is a small fraction of the tree.}
\label{tab:residual_width}
\end{table}

\begin{table}[t]
\centering\small

\setlength{\tabcolsep}{4pt}
\begin{tabular}{@{}lrrr@{}}
\toprule
Opponent & val.-rel.\ cont. & public lines & lines for $90\%$ \\
\midrule
River over-fold       & $300$ & $18$ & $7$ \\
Revealed call         & $213$ & $9$  & $2$ \\
Turn over-fold (deep) & $49$  & $3$  & $2$ \\
\bottomrule
\end{tabular}
\caption{Structural localization of the payoff-relevant residual ($\rok=0.5$).
The value-relevant continuations---those with positive single-drop shortfall---do
not merely number few; they collapse onto a small set of \emph{public betting
lines} (line${}\times{}$action groups, the unit an agent can target by driving
that line to a showdown).  ``Lines for $90\%$'' is the number of public lines
carrying $90\%$ of the residual exploitation.  A partial-A3 agent therefore knows
\emph{a priori} which lines to force, and value concentrates as censoring deepens
(the turn leak rides on two lines).}
\label{tab:residual_localize}
\end{table}

\subsection{Depth Generality}
\label{app:depth}
A synthetic censored chain swept over depths $D\in\{1,\dots,6\}$ with four opponent
types (Table~\ref{tab:depth}) checks the mechanism outside poker: the censored fiber is
not opponent-controlled, so full-coverage active reveal closes the payoff-relevant ambiguity
($\Gactivepop=G^\star$) with the floor met at every depth; the added value over the
public aggregate is non-monotonic (largest at $D{=}3$, zero at $D{=}4$), so the claim is
coverage, not a trend.
\begin{table}[t]
\centering\small

\begin{tabular}{crrrr}
\toprule
Depth $D$ & $G_{\mathrm{pub}}$ & $\Gactivepop$ & $G^\star$ & $\Delta_{\mathrm{rec}}$ \\
\midrule
1 & $0.244$ & $0.352$ & $0.352$ & $0.108$ \\
2 & $0.303$ & $0.367$ & $0.367$ & $0.065$ \\
3 & $0.276$ & $0.396$ & $0.396$ & $0.120$ \\
4 & $0.417$ & $0.417$ & $0.417$ & $0.000$ \\
5 & $0.417$ & $0.438$ & $0.438$ & $0.021$ \\
6 & $0.417$ & $0.438$ & $0.438$ & $0.021$ \\
\bottomrule
\end{tabular}
\caption{Depth generality on the synthetic censored chain
($\rho=0.1$, 4 opponent types).  Full-coverage active reveal reaches the oracle at every censoring depth
($\Gactivepop=G^\star$) and the verified safety value meets the floor.  $\Delta_{\mathrm{rec}}$
is non-monotonic in depth---largest where the public line hides the leak ($D{=}3$), zero
where the public aggregate already captures it ($D{=}4$).}
\label{tab:depth}
\end{table}

\subsection{Capacity-Frontier Distribution}
\label{app:kappa}
The study selects the $60$ largest non-fold realization-mass deviations of the
deep-turn over-fold opponent.  Forty-eight continuations lead directly to showdown.
At the other $12$, a non-fold action can reopen betting, so the reveal program fixes
call at every chance-compatible player-$0$ information set after that action.  The
resulting LP maximizes $\kappa_\rho^{\mathrm{rev}}$ over the smaller floor-safe
treeplex.  Its audit checks every fixed-flow equality, verifies that each non-fold
child's showdown reach equals the target's controlled reach, and independently
recomputes the value floor.

Table~\ref{tab:kappa} summarizes these $12$ nontrivial-suffix targets.  Raising
$\rho$ from $0.1$ to $0.5$ increases median
$\kappa_\rho^{\mathrm{rev}}$ from $0.0044$ to $0.0087$ and median joint reveal
rate from $7.34{\times}10^{-4}$ to $1.51{\times}10^{-3}$.  Here
$\pi_{\ystar}(I)=\sum_{a\ne f}y^\star_{\sigma(I)a}$ includes residual opponent
reach as well as continuation mass.

We then simulate each of the $24$ target--budget policies at
$N\in\{10^5,3{\times}10^5,10^6\}$ over five independent seeds, totaling
$168$ million hands.  At $N=10^6$, the regression of log mean observed frequency
on log predicted $\kappa_\rho^{\mathrm{rev}}\pi_{\ystar}$ has slope $1.010$ and
$R^2=0.999$; median and maximum relative errors are $1.1\%$ and $2.6\%$.
Across all sample sizes, the log conditional-interval half-width has slope
$-0.501$ against $Nq_\rho^{\mathrm{force}}$ ($R^2=0.959$), where
$q_\rho^{\mathrm{force}}:=\kappa_\rho^{\mathrm{rev}}\pi_{\ystar}$, and the realized Wilson coverage
is $97.5\%$ over the $360$ simulation cells.  Thus the measured reveal frequency
and interval contraction separately recover the two factors in the rate law.
If a prefix certificate is slower, its joint rate replaces this target-level
diagnostic.
\begin{table}[t]
\centering
\small

\setlength{\tabcolsep}{5pt}
\begin{tabular}{@{}rrrrr@{}}
\toprule
$\rho$ & $\widetilde\kappa^{\mathrm{rev}}_\rho$
& $\kappa^{\mathrm{rev},(10)}_\rho$
& $\widetilde q_\rho^{\mathrm{force}}$
& $\widehat N_{\mathrm{cert}}$ \\
\midrule
$0.1$ & $0.0044$ & $0.0030$ & $7.34{\times}10^{-4}$ & $5.5{\times}10^5$ \\
$0.5$ & $0.0087$ & $0.0072$ & $1.51{\times}10^{-3}$ & $2.6{\times}10^5$ \\
\bottomrule
\end{tabular}
\caption{Reveal-forcing capacity on the $12$ selected deep-turn targets whose
non-fold continuation can reopen betting.  Each target--budget value solves the
floor-safe reach LP with whole-information-set call equalities.  Here
$q_\rho^{\mathrm{force}}=\kappa_\rho^{\mathrm{rev}}\pi_{\ystar}$ and
$\widehat N_{\mathrm{cert}}=1/(\widetilde q_\rho^{\mathrm{force}}\varepsilon^2)$ at
$\varepsilon=0.05$.  All fixed-flow, showdown-reach, and floor audits pass.}
\label{tab:kappa}
\end{table}

\subsection{Non-Game Instances of the Frontier}
\label{app:nongame}
The safe observation frontier (Thm.~\ref{thm:capacity_law}) and the decoupling characterization
(Prop.~\ref{prop:decoupling}) use only a convex policy body, a concave floor, and a linear reach, so
they are not specific to sequence form.  We confirm both on two non-game convex bodies by solving the
relevant LPs with their floor duals exposed; Figure~\ref{fig:price_of_safety} reports the game frontier.

\emph{Conservative bandit (the simplex).}  Arms with known means under the floor
``expected reward $\ge\vref-\rok$'': a below-baseline wall arm has exactly
$\kcap=\min\{1,\mu_I\rok\}$ with $\mu_I=1/L_I$ and $L_I$ as in
Proposition~\ref{prop:capacity_mix}; the achievable
slope equals the dual to numerical precision, and $N_{\mathrm{cert}}\rok$ is constant;
the root parent is known.
Splitting the costed arm from the reveal arm exhibits the two cases of
Prop.~\ref{prop:decoupling}:
$N_{\mathrm{cert}}\rok=50/25/100$ as the reveal costs $1\times/0.5\times/2\times$ the
target, collapsing toward $0$ when the reveal is free---coupling iff co-active.

\emph{Constrained MDP (the occupancy polytope).}  A discounted safe self-loop versus a
risky depth-$d$ chain under a reward floor: the target is a wall, and
Table~\ref{tab:nongame} shows $s_I=\mu_I$ to four digits, constant
$N_{\mathrm{cert}}\rok$, and $\mu_I=\gamma^{d}$ with known root parent---the discount is the occupancy analogue
of reach decay, so deeper targets cost more, exactly as the deep-turn targets do in
hold'em (\S\ref{sec:exp:kappa}).

\begin{table}[t]
\centering
\caption{Constrained-MDP instance ($\gamma=0.9$, $\sigma^2/\Delta^2=100$): the safe observation frontier and safe
certification cost hold on the occupancy polytope, with floor dual $\mu_I=\gamma^{d}$.}
\label{tab:nongame}
\small
\begin{tabular}{lcccc}
\toprule
target depth $d$ & $0$ & $1$ & $2$ & $4$\\
\midrule
floor dual $\mu_I$ & $1.000$ & $0.900$ & $0.810$ & $0.656$\\
slope $(\kcap-\kappa_0)/\rok$ & $1.000$ & $0.900$ & $0.810$ & $0.656$\\
$N_{\mathrm{cert}}\cdot\rok$ & $100.0$ & $111.1$ & $123.5$ & $152.4$\\
\bottomrule
\end{tabular}
\end{table}

\subsection{Matched-Budget Controller and Complementary Routing Diagnostic}
\label{app:deploy}
\begin{table}[tb]
\centering\small

\begin{tabular*}{\columnwidth}{@{\extracolsep{\fill}}l*{4}{r}@{}}
\toprule
Opponent family
& $G_{\mathrm{pub}}$ & $\Gactivepop$ & $G^\star$ & $\Delta_{\mathrm{rec}}$ \\
\midrule
Equilibrium (control)      & $0.000$ & $0.000$ & $0.000$ & $0.000$ \\
River over-fold (uniform)  & $0.323$ & $0.479$ & $0.479$ & $0.156$ \\
River over-fold (strong)   & $0.434$ & $0.595$ & $0.595$ & $0.161$ \\
Turn over-fold (deep)      & $0.039$ & $0.517$ & $0.517$ & $0.478$ \\
Revealed call (strong)     & $0.313$ & $1.434$ & $1.434$ & $1.121$ \\
\bottomrule
\end{tabular*}
\caption{Full-coverage population benchmark on the coarse bucketed endgame (gains
over $\vref=-0.043$, $\rok=0.5$).  Exact pins match full monitoring, and
$G_{\mathrm{pub}}$ uses the ideal portfolio's public reach weights; the
blueprint-stream reference in Table~\ref{tab:deploy_pf} therefore differs.
Fine results are in \suppapp{app:deploy_fine}{D.8}.}
\label{tab:ladder}
\end{table}

The sample-split controller in Table~\ref{tab:sad_e2e} charges all empirical arms the
same total budget, freezes its route from a public pilot, and fits the reveal constraints
from an independent batch.  Table~\ref{tab:sad_budget_sweep} expands the result across
budgets.  At $3{\times}10^5$ hands, SAD has crossed the public-only certificate on both
over-fold families while untargeted reveal remains below it.  At $10^5$, SAD already
improves on Random for every deviating family, reducing the statistical cost of the
sample split (Holm-adjusted paired $p\le0.010$ for certified value).
\begin{table*}[t]
\centering
\small

\setlength{\tabcolsep}{3.7pt}
\begin{tabular}{@{}llrr|rr|rrr@{}}
\toprule
& & \multicolumn{2}{c}{$\Cpub$} & \multicolumn{2}{c}{Random reveal}
& \multicolumn{2}{c}{SAD} & \\
Opponent & $N$ & Cert. & Real. & Cert. & Real. & Cert. & Real. & Tgt. \\
\midrule
River over-fold
& $10^5$ & $.413\;(.004)$ & $.645\;(.001)$
& $.295\;(.013)$ & $.563\;(.022)$
& $.322\;(.034)$ & $.605\;(.061)$ & $288$ \\
& $3{\times}10^5$ & $.460\;(.003)$ & $.684\;(.001)$
& $.437\;(.022)$ & $.697\;(.023)$
& $\mathbf{.571}\;(.009)$ & $.802\;(<.001)$ & $324$ \\
& $10^6$ & $.485\;(.003)$ & $.688\;(<.001)$
& $.588\;(.027)$ & $.765\;(.029)$
& $\mathbf{.692}\;(.002)$ & $.810\;(<.001)$ & $356$ \\
\midrule
Turn over-fold
& $10^5$ & $.020\;(.003)$ & $.068\;(.001)$
& $.003\;(.001)$ & $.026\;(.004)$
& $.016\;(.002)$ & $.058\;(.003)$ & $90$ \\
& $3{\times}10^5$ & $.034\;(.002)$ & $.075\;(.001)$
& $.026\;(.005)$ & $.110\;(.007)$
& $\mathbf{.061}\;(.002)$ & $.120\;(.001)$ & $90$ \\
& $10^6$ & $.045\;(.001)$ & $.080\;(.001)$
& $.071\;(.009)$ & $.138\;(.003)$
& $\mathbf{.093}\;(.001)$ & $.126\;(.002)$ & $90$ \\
\midrule
Revealed call
& $10^5$ & $.201\;(.003)$ & $.304\;(.002)$
& $.172\;(.006)$ & $.286\;(.003)$
& $.178\;(.002)$ & $.291\;(.003)$ & $315$ \\
& $3{\times}10^5$ & $.222\;(.002)$ & $.309\;(.001)$
& $.235\;(.004)$ & $.323\;(.003)$
& $\mathbf{.236}\;(.001)$ & $.336\;(.003)$ & $315$ \\
& $10^6$ & $.232\;(.001)$ & $.310\;(<.001)$
& $.294\;(.024)$ & $.397\;(.098)$
& $\mathbf{.297}\;(.001)$ & $.377\;(.001)$ & $315$ \\
\bottomrule
\end{tabular}
\caption{Sample-split matched-budget sweep on the coarse turn--river
endgame (parentheses give two-sided 95\% Student-$t$ CI half-widths over ten
matched seeds).  Each row uses the
same total hand budget $N$ per arm.  The equilibrium control is omitted: every method
has zero certified and realized gain at all three budgets, and SAD retains no target.
All $404$ unique cells in the primary and budget grids pass the independent floor
audit, with no confidence-set infeasibility or timeout.}
\label{tab:sad_budget_sweep}
\end{table*}

\begin{figure}[t]
\centering
\includegraphics[width=\columnwidth]{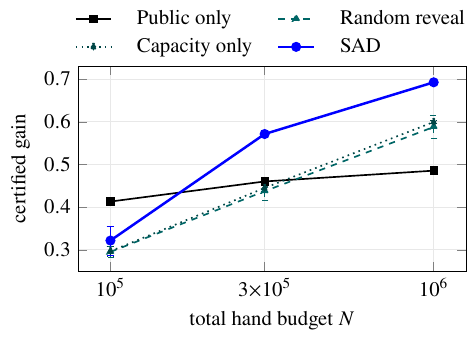}
\caption{Certified gain by acquisition strategy and hand budget on the river
over-fold.
Points are means over ten matched seeds; bars show two-sided 95\% Student-$t$
CIs.  Unscreened Capacity-only follows Random; SAD exceeds both from
$N=3{\times}10^5$.}
\label{fig:sad_budget}
\end{figure}

Table~\ref{tab:sad_tight_floor} repeats the matched controller comparison at
$\rho=0.1$, where the safe set admits substantially less probing slack.  SAD
selects a mean of $355.5$ river and $90$ turn targets, versus all $2{,}160$
positive-capacity targets under Random, and improves the certificate over both
public-only and Random acquisition on each deviation.  The paired improvements
remain significant after Holm correction across the four comparisons
($p=0.0078$); all $60$ acquisition and response plans pass the independent
$10^{-6}$ floor audit.
\begin{table*}[t]
\centering
\small

\begin{tabular}{@{}lrrrrr@{}}
\toprule
Opponent & $C^{\mathrm{pub}}$ & Random & SAD
& SAD $-$ $C^{\mathrm{pub}}$ & SAD $-$ Random \\
\midrule
River over-fold & $.174\;(.001)$ & $.234\;(.009)$ & $.271\;(.001)$ & $.097\;(.001)$ & $.037\;(.009)$ \\
Turn over-fold & $.021\;(<.001)$ & $.021\;(.001)$ & $.040\;(<.001)$ & $.019\;(<.001)$ & $.019\;(.001)$ \\
\bottomrule
\end{tabular}
\caption{Matched-budget deployment under the tighter safety budget
$\rho=0.1$ ($N=10^6$). Entries are certified-gain means with two-sided 95\%
Student-$t$ CI half-widths over ten matched seeds; the final two columns use
paired differences. Exact paired Wilcoxon tests remain significant after Holm
correction across the four comparisons ($p\le .008$).}
\label{tab:sad_tight_floor}
\end{table*}

Table~\ref{tab:sad_split_sensitivity} varies the public/reveal allocation while
holding the total budget at $10^6$.  The gain persists from $0.2/0.8$ through
$0.8/0.2$: every split improves both certified and realized value over $\Cpub$ on
all three deviations, while the equilibrium screen remains empty.  The near-identical
river result under $0.2/0.8$ and $0.5/0.5$ also shows that the result is not tied to
a single allocation.
\begin{table*}[t]
\centering
\small

\setlength{\tabcolsep}{4.5pt}
\begin{tabular}{@{}lrr|rr|rr|r@{}}
\toprule
& \multicolumn{2}{c}{River over-fold}
& \multicolumn{2}{c}{Turn over-fold}
& \multicolumn{2}{c}{Revealed call} & \\
Public/reveal & Cert. & Real. & Cert. & Real. & Cert. & Real. & Eq.\ tgt. \\
\midrule
$0.2/0.8$ & $.692\;(.002)$ & $.810\;(<.001)$
& $.093\;(.001)$ & $.126\;(.002)$
& $.297\;(.001)$ & $.377\;(.001)$ & $0$ \\
$0.5/0.5$ & $.692\;(.003)$ & $.810\;(<.001)$
& $.089\;(.001)$ & $.120\;(.002)$
& $.288\;(.001)$ & $.365\;(.001)$ & $0$ \\
$0.8/0.2$ & $.656\;(.003)$ & $.802\;(<.001)$
& $.082\;(.001)$ & $.116\;(.001)$
& $.267\;(.001)$ & $.341\;(.001)$ & $0$ \\
\bottomrule
\end{tabular}
\caption{Public/reveal allocation sensitivity at $N=10^6$
(parentheses give two-sided 95\% Student-$t$ CI half-widths over ten matched
seeds).  Each row uses the final
sample-split protocol with the same total budget.  Every allocation improves both
certified and realized value over $\Cpub$ on all three deviations, retains no
equilibrium target, and passes all floor, feasibility, and runtime audits
($120$ SAD cells).}
\label{tab:sad_split_sensitivity}
\end{table*}

A complementary controlled diagnostic fixes the public
object and routing signal at their population values while varying the reveal budget,
separating convergence of the reveal constraints from route-estimation variance.

Table~\ref{tab:deploy} aggregates that diagnostic at $N=10^5$ reveal hands, and
Table~\ref{tab:deploy_pf} expands it by family and reveal budget.  Its oracle-targeting
column is a finite-sample upper reference; $G^\star$ in Table~\ref{tab:ladder} is the
full-monitoring ceiling.
\begin{table}[!tb]
\centering\small

\setlength{\tabcolsep}{4pt}
\begin{tabular}{@{}lrrrr@{}}
\toprule
Method & Cert. & Real. & Cap. & Viol. \\
\midrule
$\Cpub$ (no reveal)       & $0.270$ & $0.363$ & $0.00$  & $0/3$ \\
Random probe              & $0.217$ & $0.337$ & $-0.07$ & $0/30$ \\
Active reveal ($\Dpub$)    & $0.322$ & $0.413$ & $0.17$  & $0/30$ \\
Oracle-target             & $0.412$ & $0.581$ & $0.26$  & $0/30$ \\
\midrule
\textbf{Routed diagnostic} & $\mathbf{0.329}$ & $\mathbf{0.409}$ & $\mathbf{0.17}$ & $0/30$ \\
Oracle                    & $0.999$ & $0.999$ & $1.00$  & $0/3$ \\
\bottomrule
\end{tabular}
\caption{Budgeted Safe Active De-censoring on the coarse turn--river endgame ($\rho=0.5$,
$\delta=0.1$, $N=10^5$; mean over the three deviating families, 10 matched seeds).
\emph{Cert.}/\emph{Real.}: robust certificate and realized value against $\ystar$;
the public baseline and routing signal use population public frequencies, while reveal
constraints use $N$ hands.  \emph{Cap.}: mean per-family share of the
oracle$\,$--$\,\Cpub$ gap captured;
\emph{Viol.}: floor violations (Theorem~\ref{thm:safety}).  \emph{Oracle-target} probes with
hidden leak labels (upper reference); \emph{Routed diagnostic}: per-cell maximum with
$\Cpub$.  The controlled design fixes public quantities to isolate finite-sample reveal
constraints; Table~\ref{tab:deploy_pf} gives the per-family breakdown.}
\label{tab:deploy}
\end{table}

The diagnostic certificates increase monotonically with reveal budget and approach
their population references from below, as predicted by
Theorem~\ref{thm:active_id}.  Two-sided 95\% Student-$t$ CI half-widths are at
most $0.014$.  The shallow river
over-fold tightens fastest ($0.663$ versus the oracle-target reference $0.678$ at
$N=10^5$); the deep turn and revealed-call coordinates remain cost-limited at that
budget.  Public targeting never exceeds hidden-label oracle targeting.
\begin{table*}[t]
\centering\small

\setlength{\tabcolsep}{5pt}
\begin{tabular}{@{}llrrrrrr@{}}
\toprule
Family & $N$ & $\Cpub$ & Passive & Random & SAD ($\Dpub$) & Oracle-tgt & Oracle \\
\midrule
River over-fold & $10^{3}$ & $0.510$ & $0.001_{\pm.002}$ & $0.081_{\pm.011}$ & $0.206_{\pm.011}$ & $0.207_{\pm.014}$ & $0.828$ \\
                & $10^{4}$ & $0.510$ & $0.230_{\pm.009}$ & $0.315_{\pm.007}$ & $0.495_{\pm.007}$ & $0.507_{\pm.009}$ & $0.828$ \\
                & $10^{5}$ & $0.510$ & $0.412_{\pm.005}$ & $0.476_{\pm.007}$ & $0.663_{\pm.002}$ & $0.678_{\pm.002}$ & $0.828$ \\
\midrule
Turn over-fold  & $10^{3}$ & $0.058$ & $0.000$ & $0.000$ & $0.000$ & $0.000$ & $0.735$ \\
                & $10^{4}$ & $0.058$ & $0.001_{\pm.000}$ & $0.004_{\pm.000}$ & $0.036_{\pm.002}$ & $0.025_{\pm.007}$ & $0.735$ \\
                & $10^{5}$ & $0.058$ & $0.023_{\pm.005}$ & $0.053_{\pm.005}$ & $0.082_{\pm.002}$ & $0.135_{\pm.002}$ & $0.735$ \\
\midrule
Revealed call   & $10^{3}$ & $0.242$ & $0.001_{\pm.000}$ & $0.015_{\pm.005}$ & $0.081_{\pm.005}$ & $0.078_{\pm.007}$ & $1.434$ \\
                & $10^{4}$ & $0.242$ & $0.117_{\pm.005}$ & $0.070_{\pm.002}$ & $0.152_{\pm.002}$ & $0.258_{\pm.009}$ & $1.434$ \\
                & $10^{5}$ & $0.242$ & $0.200_{\pm.002}$ & $0.121_{\pm.000}$ & $0.222_{\pm.000}$ & $0.424_{\pm.005}$ & $1.434$ \\
\bottomrule
\end{tabular}
\caption{Per-family robust objective by reveal budget in the controlled
population-public diagnostic (coarse turn--river endgame, $\rho=0.5$,
$\delta=0.1$; mean${}\pm{}$two-sided 95\% Student-$t$ CI half-width over $10$
matched seeds).  $\Cpub$ and Oracle are population references, with $\Cpub$
weighted by the blueprint collection stream rather than the ideal portfolio of
Table~\ref{tab:ladder}; the active columns use finite reveal constraints.  At
$N=10^5$, SAD's lift on both over-fold families is significant under a one-sided
Wilcoxon test ($p=0.001$).  Active and public objectives are certificates on
coverage; the passive column is statistically uncertified when its
showdown-conditioned box is nonempty and otherwise reports the empirical-$\Cpub$
fallback.  Every plan passes the floor audit.}
\label{tab:deploy_pf}
\end{table*}

\subsection{Deployment on the Fine Abstraction}
\label{app:deploy_fine}
The fine endgame ($4{,}590$ opponent information sets, $\vref=-0.053$) reproduces both
primary results: the conditional full-coverage ladder
(Table~\ref{tab:ladder_fine}) gives $\Gactivepop=G^\star$ on every family, and the budgeted deployment
(Table~\ref{tab:deploy_fine}) preserves the coarse structure.

\begin{table}[t]
\centering\small
\caption{Conditional full-coverage population ladder on the \emph{fine} endgame
($\vref=-0.053$, $\rho=0.5$); the fine-abstraction counterpart of
Table~\ref{tab:ladder}.}
\label{tab:ladder_fine}
\begin{tabular*}{\columnwidth}{@{\extracolsep{\fill}}l*{4}{r}@{}}
\toprule
Opponent family
& $G_{\mathrm{pub}}$ & $\Gactivepop$ & $G^\star$ & $\Delta_{\mathrm{rec}}$ \\
\midrule
Equilibrium (control)      & $0.000$ & $0.000$ & $0.000$ & $0.000$ \\
River over-fold (uniform)  & $0.304$ & $0.590$ & $0.590$ & $0.286$ \\
River over-fold (strong)   & $0.388$ & $0.718$ & $0.718$ & $0.330$ \\
Turn over-fold (deep)      & $0.459$ & $0.542$ & $0.542$ & $0.083$ \\
Revealed call (strong)     & $0.650$ & $1.080$ & $1.080$ & $0.430$ \\
\bottomrule
\end{tabular*}
\end{table}

The passive box is infeasible in every cell.  All $480$ completed responses clear the
floor; eight cells reach the solver wall-time cap.  Routed SAD
lifts the aggregate certificate ${\sim}18\%$---capturing $88\%$ of the oracle-targeting
gain on the censored river over-fold.  The finer public partition exposes more of each
leak, so SAD probes the turn over-fold sparingly (certifying $0.620$, near the $\Cpub$
baseline $0.592$), while naive oracle targeting deviates to low-anomaly lines and shifts
its own reach-weighted box away from the value it certifies ($0.233$): probing only where
revealing is informative spends budget and reach efficiently.
\begin{table*}[t]
\centering\small

\setlength{\tabcolsep}{5pt}
\begin{tabular}{@{}lrrrrrr@{}}
\toprule
Family & $\Cpub$ & Passive & Random & SAD ($\Dpub$) & Oracle-tgt & Oracle \\
\midrule
River over-fold & $0.466$ & $0.388_{\pm.003}$ & $0.542_{\pm.006}$ & $0.720_{\pm.003}$ & $0.753_{\pm.004}$ & $0.977$ \\
Turn over-fold  & $0.592$ & $0.495_{\pm.005}$ & $-0.000_{\pm.000}$ & $0.620_{\pm.004}$ & $0.233_{\pm.004}$ & $0.778$ \\
Revealed call   & $0.473$ & $0.324_{\pm.009}$ & $0.100_{\pm.002}$ & $0.371_{\pm.005}$ & $0.514_{\pm.004}$ & $1.080$ \\
\bottomrule
\end{tabular}
\caption{Deployment on the \emph{fine} turn--river endgame ($4{,}590$ opponent information
sets; $\rho=0.5$, $\delta=0.1$, $N=10^5$, mean${}\pm{}$two-sided 95\%
Student-$t$ CI half-width over $10$ matched seeds, except $n=9$ for SAD on the
two over-fold families)---the
fine-abstraction counterpart of Table~\ref{tab:deploy}.  The passive box is infeasible in
every cell ($30/30$, MNAR signature of Theorem~\ref{thm:passive_mnar}) while the
active-reveal sets are feasible throughout ($0/30$); all $28$ completed responses clear
the floor, and two SAD cells reach the solver wall-time cap.  Public-targeted SAD lifts the certificate on the river and turn over-folds; on
the revealed-call family the fine public aggregate already certifies more than the
finite-sample probe ($\Cpub=0.473$ versus SAD $0.371$), so Routed SAD---the deployed
maximum with $\Cpub$---falls back to $\Cpub$ there and never certifies below it.  On the deep
turn over-fold SAD ($0.620$) exceeds naive oracle targeting ($0.233$), and on the
high-capacity river it nearly matches it ($0.720$ versus $0.753$): the turn public anomaly
is weak (the fine public aggregate already exposes the leak), so SAD correctly probes little
and stays near $\Cpub$, whereas oracle targeting over-probes a low-anomaly line and
self-loosens its reach-weighted box.  Routed
SAD lifts the aggregate certificate from $0.510$ (no reveal) to $0.600$ with $0$ floor
violations.}
\label{tab:deploy_fine}
\end{table*}

\subsection{Comparison with Opponent-Modelling Baselines}
\label{app:baselines}
The fixed uniformly random non-fold collection policy of
\citet{davis2019solving} is a direct acquisition comparator.  We pair $0.2N$
blueprint pilot hands with $0.8N$ hands from this fixed policy, then apply the
same empirical-Bernstein confidence construction and robust responder used by
SAD.  Its downstream certificates can be strong (Table~\ref{tab:fixed_nonfold}),
but the collector's independently audited worst-case value is $-2.313$, which
misses the $-0.543$ floor by $1.769$ in every cell.  The comparison therefore
separates the statistical utility of broad non-fold collection from the
acquisition guarantee studied here; every downstream response remains
floor-safe.
\begin{table}[t]
\centering
\small

\begin{tabular}{@{}lrrr@{}}
\toprule
Opponent & Cert. & Pop. & $\Delta_{\rm acq}$ \\
\midrule
Equilibrium     & $.000\;(.000)$ & $.000\;(.000)$ & $-.743$ \\
River over-fold & $.565\;(.003)$ & $.788\;(.001)$ & $-.420$ \\
Turn over-fold  & $.372\;(.002)$ & $.624\;(.001)$ & $+.155$ \\
Revealed call   & $.864\;(.008)$ & $1.379\;(.001)$ & $-.827$ \\
\bottomrule
\end{tabular}
\caption{Davis-style fixed non-fold collection on the coarse endgame
($N=10^6$, $\rho=0.5$).  The downstream responder uses the same sampled
confidence construction as SAD.  Parentheses are two-sided 95\% Student-$t$ CI
half-widths over ten seeds.  $\Delta_{\rm acq}$ is the exact collection-phase
expected-value change relative to the blueprint, so a sampling CI is not
applicable.}
\label{tab:fixed_nonfold}
\end{table}

The deployment tables position our arms against public and oracle targeting; we now compare
against named opponent-modeling methods on the same opponents under one protocol
($\rho=0.5$, $N=10^5$ probe hands, ten matched seeds).  \emph{EM-BR} is a safety-constrained
best response to the censored-EM opponent reconstruction (the maximum-likelihood per-card model
under showdown censoring); \emph{RNR} is the matched-safety restricted Nash response
\citep{johanson2007robust}, the largest interpolation toward that model whose worst case still
clears the floor; \emph{DBR} is the matched-safety data-biased response
\citep{johanson2009data}, whose interpolation weight is per-infoset data confidence---the
most locally adaptive of these model-based baselines; all three are safe.
\emph{RNR$_{p=1}$} is the unconstrained best response to the same model.
Table~\ref{tab:baselines} reads in three parts.  On the equilibrium control the safe
model-based responses \emph{lose} value because they fit noise in the censored data and infer
a spurious exploit that reduces realized value, whereas the public-robust core stays at zero
(DBR's per-infoset confidence nearly repairs this).  On the censored leaks the core recovers
the most at both granularities---significantly ahead of even DBR, the best model-based arm
(one-sided Wilcoxon over matched seeds; $p$-values in the caption)---the realized form of
Theorem~\ref{thm:passive_mnar}: a model fit to censored data tracks a biased estimand, and
per-infoset confidence weighting mitigates but does not remove the bias.  On the
\emph{uncensored} revealed call the model-based responses are competitive, delineating
the boundary of the censoring effect.  The unconstrained RNR$_{p=1}$ leads on realized
value but breaches the floor on every opponent, so its gains carry no guarantee.
\begin{table}[t]
\centering\small

\setlength{\tabcolsep}{2pt}
\begin{tabular}{@{}lrrrrr@{}}
\toprule
Opp. (oracle) & core & EM-BR & RNR & DBR & RNR$_1$ \\
\midrule
\multicolumn{6}{@{}l}{\emph{Coarse endgame} (floor $-0.543$)} \\
Equil. (ctrl) ($0.00$) & $+0.000$ & $-0.075$ & $-0.055$ & $-0.004$ & $-0.059$ \\
River OF ($0.83$) & $+0.645$ & $+0.151$ & $+0.142$ & $+0.488$ & $+0.708$ \\
Turn OF ($0.73$) & $+0.069$ & $+0.012$ & $-0.004$ & $+0.027$ & $-0.032$ \\
Rev. call ($1.43$) & $+0.304$ & $+0.581$ & $+0.608$ & $+0.346$ & $+1.462$ \\
\midrule
\multicolumn{6}{@{}l}{\emph{Fine endgame} (floor $-0.553$)} \\
Equil. (ctrl) ($-0.00$) & $-0.000$ & $-0.242$ & $-0.242$ & $-0.025$ & $-0.280$ \\
River OF ($0.98$) & $+0.541$ & $+0.371$ & $+0.362$ & $+0.183$ & $+0.133$ \\
Turn OF$^{(8)}$ ($0.78$) & $+0.700$ & $+0.323$ & $+0.185$ & $+0.721$ & $+0.651$ \\
Rev. call$^{(9)}$ ($1.08$) & $+0.585$ & $+0.172$ & $+0.170$ & $+0.714$ & $+0.295$ \\
\midrule
Floor held & yes & yes & yes & yes & \textbf{no} \\
\bottomrule
\end{tabular}
\caption{Named baselines vs.\ the public-robust core on the matched
deployment opponents ($\rho=0.5$, $N=10^5$, mean realized value vs.\
$\ystar$ over $10$ matched seeds; two-sided 95\% Student-$t$ CI half-width
$\le 0.043$).  \emph{EM-BR} and
\emph{RNR} are safety-constrained censored-EM / restricted-Nash responses;
\emph{DBR} is the matched-safety data-biased response with per-infoset
confidence \citep{johanson2009data}; RNR$_{p{=}1}$ is unconstrained.
Parenthesized values are safety-constrained oracle gains; superscripts mark
cells with fewer than $10$ seeds after solver stalls.  The core recovers the
most wherever censoring hides the value (river at both granularities, turn
at the coarse one); at the \emph{fine} turn over-fold---whose leak the
finer public partition largely exposes (App.~E.9)---DBR slightly exceeds the core on
realized value, with no certificate attached.  One-sided Wilcoxon over
matched seeds: b2 river: core $>$ DBR, $p=0.0010$; b2 turn: core $>$ DBR, $p=0.0010$; b4 river: core $>$ EM, $p=0.0010$; b4 turn: DBR $>$ core, $p=0.0039$.  All safe arms hold the floor; RNR$_{p{=}1}$
breaches it on every opponent.}
\label{tab:baselines}
\end{table}

\subsection{Capacity-Routing Ablation}
\label{app:ablation}
A reported-protocol ablation replaces SAD's sampled public screen by an unscreened,
data-independent route that weights all $2158$ positive-capacity river targets by
$\kcap(I)$.  It otherwise uses the same $0.2N/0.8N$ public/reveal split, confidence
construction, safety set, and ten matched seeds (Table~\ref{tab:ablation}).  Capacity-only
tracks Random at every budget (Holm-adjusted paired $p\ge0.316$).  At
$N=3{\times}10^5$ and $10^6$, SAD exceeds Capacity-only by
$0.126\;(.003)$ and $0.093\;(.003)$, respectively (paired 95\% CI half-widths;
adjusted $p=0.006$).  Thus capacity prices attainable evidence, whereas the
sampled public anomaly locates the directions worth observing.
\begin{table}[t]
\centering\small

\setlength{\tabcolsep}{3.5pt}
\begin{tabular}{@{}lrrrr@{}}
\toprule
$N$ & $\Cpub$ & Random & Capacity-only & SAD \\
\midrule
$10^5$ &
$.413\;(.004)$ & $.295\;(.013)$ & $.297\;(.012)$ & $.322\;(.034)$ \\
$3{\times}10^5$ &
$.460\;(.003)$ & $.437\;(.022)$ & $.445\;(.008)$ & $.571\;(.009)$ \\
$10^6$ &
$.485\;(.003)$ & $.588\;(.027)$ & $.599\;(.003)$ & $.692\;(.002)$ \\
\bottomrule
\end{tabular}
\caption{Reported-protocol routing ablation on the river over-fold
($\rho=0.5$, $\delta=0.1$).  Entries are certified mean and two-sided 95\%
Student-$t$ CI half-width over ten matched seeds.  Active arms share the
$0.2N/0.8N$ public/reveal split; Capacity-only routes all $2158$
positive-capacity targets without using the sampled public anomaly.}
\label{tab:ablation}
\end{table}

This is the quantitative form of Proposition~\ref{prop:public_rank}.  Partitioning the
per-target leak value by public state and decomposing its variance, the within-public-state
component---the part no public score can rank---is $0.57$, $0.51$, and $0.42$ of the total
for the river over-fold, deep turn over-fold, and revealed call, respectively (roughly half
in every family), while the public anomaly $\Dpub$ has zero spread within any public state.
This decomposition explains why ordering public states captures most of the available
routing gain even though the anomaly score is constant within a public state.

\subsection{Globally Optimized Joint Routing}
\label{app:joint_routing_experiment}

We instantiate Proposition~\ref{thm:joint_routing} on a predeclared four-target
Leduc portfolio containing one censored facing-bet target from each of four
first-round public histories.  Target weights are uniform, the safety budget is
$\rho=0.1$, and $20\%$ of the mixture is reserved for a capacity-normalized
max--min audit.  Reveal rates use one explicit uniform-opponent continuation model, so
Proposition~\ref{thm:joint_routing} applies to this affine rate model rather than
opponent-universal rates.
Every linear subproblem is the safe weighted-reach LP; an exactly normalized
dyadic opponent derived from its row duals, followed by outward-rounded
treeplex evaluation, supplies its objective upper bound.  The full
cross-reveal vector is recomputed for each returned realization plan.

Table~\ref{tab:joint_routing} separates two effects.  Accounting for incidental
cross-reveals lowers the specified uniform coordinate-width envelope by
$11.0\%$.  Optimizing those joint rates lowers it by a further $12.3\%$, to
$73.113$.  The final numerical Frank--Wolfe residual is $0.0291$, giving the
binary64 bracket $[73.084,73.113]$ under the stated $10^{-7}$ feasibility
tolerance.  The exact-arithmetic counterpart is the global bound in
Proposition~\ref{thm:joint_routing}.  Independent best responses verify the floor
for every solo probe, the universal audit, the separable mixture, and the joint
solution; the blueprint interior share gives each returned plan positive floor
margin.
These quantities are deterministic optimization outputs rather than
sampled estimates, so sampling confidence intervals are not applicable.
\begin{table}[t]
\centering\small

\begin{tabular*}{\columnwidth}{@{\extracolsep{\fill}}lr@{}}
\toprule
Routing design & Envelope $\Phi$ \\
\midrule
Diagonal prediction                    & $93.650$ \\
Separable shares with cross-reveals    & $83.325$ \\
Capacity-normalized audit              & $74.404$ \\
Joint width optimum                    & $73.113$ \\
\quad numerical optimum bracket        & $[73.084,73.113]$ \\
\bottomrule
\end{tabular*}
\caption{Safe routing on a four-target Leduc portfolio with a fixed 20\%
capacity-normalized audit share.  Lower uniform coordinate-width envelope is better.  The joint design uses all
cross-reveals under the stated uniform-continuation model; its Frank--Wolfe
residual uses outward-rounded dual-feasible weighted-reach upper bounds for
the linear subproblems.  The displayed numerical bracket uses binary64 outer
arithmetic and the stated feasibility tolerance.}
\label{tab:joint_routing}
\end{table}

\subsection{Public Anomaly Signals}
\label{app:gate_signals}
The early-line statistic $\De$ is the poker instantiation of $\Dpub$; the panel of
low-cost alternatives (Table~\ref{tab:signals}) is less predictive everywhere else.
Revealed-play deviations are uninformative or anti-correlated (Pearson $-0.16$ to
$-0.18$)---on deep-censored families the value-bearing deviation lies on the censored
action---and pooled or post-censoring fold excess is weak ($-0.28$ to $-0.34$), while the
early pre-censoring fold excess reaches $+0.96$ against the population recoverable value.
\begin{table*}[t]
\centering\small

\begin{tabular}{llrr}
\toprule
Signal & Information used & Pearson & Spearman \\
\midrule
Revealed realization dev.       & showdown only      & $-0.172$ & $+0.161$ \\
Behavior TV (all actions)       & showdown only      & $-0.168$ & $+0.161$ \\
Behavior dev., fold only        & showdown only      & $-0.163$ & $+0.161$ \\
Public fold excess (all lines)  & public line        & $-0.276$ & $+0.042$ \\
Public fold excess, river       & public line (late) & $-0.338$ & $-0.214$ \\
\midrule
\textbf{Behavior dev., turn}        & showdown (early)   & $\mathbf{+0.955}$ & $\mathbf{+0.762}$ \\
\textbf{Public fold excess, turn} ($\De$) & \textbf{public line (early)} & $\mathbf{+0.955}$ & $\mathbf{+0.762}$ \\
\bottomrule
\end{tabular}
\caption{Low-cost (LP-free) candidate signals versus the population recoverable value
on the coarse endgame ($\rho=0.5$; total available lift $2.04$).  Only the early
(pre-censoring) turn line predicts where active reveal pays off; revealed-only and
late-line signals are uninformative or anti-correlated.  The public early-line fold
excess $\De$ and its showdown-space twin coincide on this suite because the deviation is
localized on the turn fold.  Pearson / Spearman across the diagnostic opponent suite.}
\label{tab:signals}
\end{table*}

\section{Certification at the Unbucketed River}
\label{app:fulldeck}

The body's endgames use card bucketing so that every robust response is one
exact LP.  This appendix removes the abstraction on the opponent side and
re-runs the pipeline on a \emph{fixed-board, unbucketed river}:
(A$\spadesuit$\,K$\spadesuit$\,Q$\diamondsuit$\,J$\clubsuit$\,9$\heartsuit$),
$1081$ hand combinations per player ($1.07$M deals with exact card removal),
pot-bet/all-in betting capped at three raises, ${\sim}1.7\times10^4$
sequences per player.  At this size the robust-response LP is no longer practical to
materialize---its constraint matrix has $10^{7\text{--}8}$ nonzeros, and
already at bucketed sizes the monolithic solver times out on the sparsest
confidence shape (Table~\ref{tab:g1_solver})---so every result below is
computed by the certified oracle decomposition of
\S\ref{app:fulldeck_solving}.  Two exact,
LP-free primitives anchor soundness at any scale: a best-response dynamic
program over the range vectors gives the floor value $\Wval(x)$ of any
candidate plan, and the same program gives $\vref$ exactly.  The blueprint
is range CFR$^+$
($2000$ iterations, $1.1$\,s single-core, exploitability $1.4\times10^{-4}$,
measured by the exact DP); its security value $\vref=-0.065$ enters the safe
set through the blueprint-shift lemma below, so blueprint
suboptimality re-indexes the budget axis rather than perturbing any
statement.  We verified the identity of Lemma~\ref{lem:bp_shift} here
directly: at five score-spread targets of the first facing-a-bet state,
three budgets $\rok\in\{0.1,0.25,0.4\}$, and two weakened blueprints
($50$ and $300$ CFR iterations, $\varepsilon_{\mathrm{bp}}$ known exactly
by DP), the independently computed sides
$\kappa_{\rok}(I;\vref)$ and $\kappa_{\rok+\varepsilon_{\mathrm{bp}}}(I;v^\star)$
agree to $6.5\times10^{-6}$ in the worst case over all $30$ frontier
points.

\subsection{Certified Solving Past the LP Wall}
\label{app:fulldeck_solving}

The robust response $\max_{x\in\Sset}\min_{y\in C}x^\top\Amat y$ is solved
by an oracle decomposition: a cutting-plane master over a small set of
opponent responses, each cut generated by one inner minimization over $C$,
with the payoff matrix touched only through matrix-vector oracles
($\Amat^\top x$ and $\Amat y$ against the range engine, never materialized).
Iterative solving changes \emph{how} a plan is found, not \emph{what} is
guaranteed: feasibility is restored exactly, and certificates remain sound
at every iterate.

\begin{proposition}[Verify-and-repair composition]
\label{prop:compose}
Let $\hat x\in\Xpoly$ be any candidate plan, $L(\hat x)\le\min_{y\in C}\hat
x^\top\Amat y$ a valid robust lower bound, and $v_t=\vref-\rok$ the floor
target.  If $\Wval(\hat x)<v_t$, set
$\alpha=(v_t-\Wval(\hat x))/(\Wval(x_{\mathrm{bp}})-\Wval(\hat x))\in(0,1]$
and $\tilde x=(1-\alpha)\hat x+\alpha x_{\mathrm{bp}}$; otherwise
$\tilde x=\hat x$, $\alpha=0$.  Then:
\begin{enumerate}
\item \textbf{(floor, exact)} $\Wval(\tilde x)\ge v_t$, by concavity of
      $\Wval$;
\item \textbf{(certificate, computable loss)}
      $\min_{y\in C}\tilde x^\top\Amat y\;\ge\;
       L(\hat x)-\alpha\big(L(\hat x)-\ell_{\mathrm{bp}}\big)$,
      where $\ell_{\mathrm{bp}}=\min_{y\in C}x_{\mathrm{bp}}^\top\Amat y$,
      since $x\mapsto\min_{y\in C}x^\top\Amat y$ is concave;
\item \textbf{(composition)} if additionally the solver reports an upper
      bound $U\ge\max_{x\in\Sset}\min_{y\in C}x^\top\Amat y$, the deployed
      plan $\tilde x$ carries the sound bracket
      $[\,L(\hat x)-\alpha(L(\hat x)-\ell_{\mathrm{bp}}),\,U\,]$, and in
      Theorem~\ref{thm:value_recovery} the total certificate error is at most
      $\eta+\delta_s+\delta_r$ with $\delta_s=(U-L(\hat x))_+$ and
      repair loss $\delta_r=\alpha(L(\hat x)-\ell_{\mathrm{bp}})_+$.
\end{enumerate}
The floor guarantee is unconditional and exact at every scale; only the
exploitation certificate pays $\delta_s+\delta_r$, and both terms are
reported per solve.
\end{proposition}

\begin{proposition}[Anytime certificates]
\label{prop:anytime}
Fix a safe candidate $\hat x\in\Sset$ and a confidence set
$C=\{y:Gy\le h\}\cap\Ypoly$.
\begin{enumerate}
\item \textbf{(lower, weak duality)} For every $\lambda\ge0$,
\[
\min_{y\in\Ypoly}\big[\hat x^\top\Amat y+\lambda^\top(Gy-h)\big]
\;\le\;\min_{y\in C}\hat x^\top\Amat y ,
\]
and the left side is one treeplex best-response DP with
$\lambda$-shifted per-sequence losses.
\item \textbf{(upper, relaxation)} For every finite set of opponents
$y_1,\dots,y_K\in C$,
$\max_{x\in\Sset'}\min_{k}x^\top\Amat y_k
 \;\ge\;\max_{x\in\Sset}\min_{y\in C}x^\top\Amat y$
whenever $\Sset'\supseteq\Sset$ is any outer relaxation of the safe set
(e.g.\ finitely many floor cuts); in particular the master bound is valid
with \emph{any} subset of retained cuts, so bounded-memory cut windows
preserve soundness.
\end{enumerate}
Hence every safe iterate carries a sound bracket.  An arbitrary
$\hat x\in\Xpoly$ first uses Proposition~\ref{prop:compose}; its repaired
lower bound paired with the relaxation upper bound is sound.  Stopping at
any wall-clock or iteration budget therefore yields a conservative
certificate after this repair.
\end{proposition}

Table~\ref{tab:g1_solver} calibrates the decomposition against the
monolithic robust LP on the bucketed endgames, where the LP is still
available: dense confidence shapes reach $10^{-4}$ parity, and on the
sparse shape where the monolithic LP times out the decomposition still
certifies.  Wall-capped solves return certified brackets, exactly as
Proposition~\ref{prop:anytime} prescribes.
\begin{table}[t]
\centering\small
\setlength{\tabcolsep}{3pt}

\begin{tabular}{@{}lrllr@{}}
\toprule
Instance & rows & monolithic LP & decomposition & gap \\
\midrule
b2/dense & $9900$ & $0.9877$ ($122$\,s) & $0.9877$ ($666$\,s) & $6\times 10^{-5}$ \\
b2/sparse & $1260$ & timeout ($600$\,s) & $0.0647$ ($421$\,s) & $3\times 10^{-2}$ \\
b4/dense & $18360$ & $0.8241$ ($373$\,s) & $0.8240$ ($1452$\,s) & $2\times 10^{-4}$ \\
b4/sparse & $2368$ & $0.0606$ ($190$\,s) & $0.0388$ ($857$\,s) & $3\times 10^{-2}$ \\
\bottomrule
\end{tabular}
\caption{Oracle decomposition vs.\ the monolithic robust LP on the bucketed
endgames ($200$-iteration cap, $600$\,s LP budget). \emph{Gap} is the
decomposition's own certificate (upper minus lower bound). Dense shapes reach
$10^{-4}$ parity; on the sparse b2 shape the monolithic LP times out while
the decomposition certifies; on sparse b4 the capped decomposition
remains conservative (its bracket contains the LP value).}
\label{tab:g1_solver}
\end{table}

\subsection{The Public Twin}
\label{app:twin}

\begin{figure}[t]
\centering
\includegraphics[width=\columnwidth]{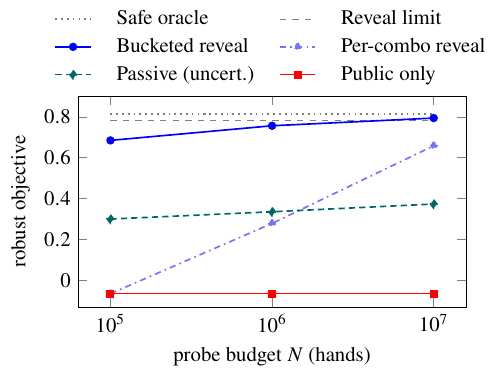}
\caption{Public twin on a fixed-board, unbucketed river.
The public lower bound stays at $\vref$ while reveal certificates approach the
population-limit band.  Bars show two-sided 95\% Student-$t$ CIs over ten seeds
($\rho=0.5$); the passive curve is an uncertified robust objective.}
\label{fig:g1_twin}
\end{figure}

The sharpest unbucketed-river instance of the channel dichotomy is a deployment
opponent constructed so that the public channel is blind to it \emph{by
identity}, not merely at finite samples.  Split the opponent's combos at
the median showdown strength.  At every first opponent decision whose call
closes the betting (call $\to$ showdown, fold $\to$ terminal payoff), weak
combos move a fraction of their fold mass to call (calling with losers)
while strong combos move call mass to fold (folding winners), the two
transfers balanced under the agent's deployment reach weights so the
reach-weighted fold and call frequencies at that node are exactly
preserved.  Raise mass is untouched, so every aggregate downstream of a
raise is unchanged; the transferred mass reaches only terminals, so it has
no downstream aggregate; and the balance holds the node's own aggregate.
Hence every population public row coincides with the base's.  Under blueprint
collection, the empirical public confidence sets for the twin and base therefore
have the same sampling law, although matched finite-sample draws need not be
identical.  The
balance is anchored to the deployment plan's reach weights, which is
exactly the fiber's definition (\S\ref{sec:prelim}: the observation fiber
is taken \emph{at} the data-collecting plan)---the twin is a member of the
deployed plan's fiber with exploitable value; an agent that re-weights its
public channel by deviating is already probing.  Confining the transfer to closing calls at first
decisions is what the fiber demands: composition shifts elsewhere propagate
through reach into downstream public aggregates and re-emerge.  Both
transfers are mistakes for the combos making them, and they are worth
$V=0.815$ to a safe responder against $\vref=-0.065$ (the unconstrained
best response is worth $1.64$; holding the floor costs half).

\paragraph{Reveal-forcing estimation.}
The reveal pins that identify the twin are sound only under two
disciplines, both instances of the paper's identification machinery.
(i)~\emph{Realization space, not behavior space}: showdown-conditioned
action frequencies are MNAR-biased exactly as in
Theorem~\ref{thm:passive_mnar} (folds never reach a showdown), so pins must
bound the non-fold sequence values $w_I\,y[\sigma(I)a]$---whose probe-batch
counts are unbiased---and leave every fold mass to flow conservation
(Lemma~\ref{lem:flow}).
(ii)~\emph{Reveal-forcing probes}: if a probe re-opens the betting after
the target's response, a later opponent fold censors the earlier record and
the counts under-estimate the reach mass---the multi-decision form of the
same bias, and precisely the property excluded by the reveal certificate of
Definition~\ref{def:class}.  Probes that close the action (call-down,
bet-then-call, shove-then-call) are reveal-forcing on their lines; lines no
probe forces remain outside the certified class, and their pins
are simply absent.  Family-wise coverage over all pinned rows uses the
spatial union bound of Appendix~\ref{app:pubconf}.  Because per-combo pins
on the unbucketed river spread $N$ over $1081$ combinations, we also intersect
\emph{bucketed} reveal pins---the same rows aggregated over $K{=}4$
strength quantiles---which trade fiber resolution for a $1081/K$ rate gain.

\paragraph{What each design can certify.}
Table~\ref{tab:g1_residual} prices the twin's fiber at the population
limit (Proposition~\ref{thm:fiber_minimax}).  Public rows certify $\vref$
exactly---nothing, at any sample size.  Bucketed reveal pins under the
three-probe portfolio leave a residual width of $\Delta^-_M=0.032$, i.e.\
they recover $96\%$ of the safe-exploitable gap $V-\vref$
(Proposition~\ref{thm:residual_recovery});
per-combo pins can only be tighter.  The probe \emph{portfolio} matters,
not just the probe count: with only the first two probes, no line forces
the twin's all-in transfers, and the same bucketed design strands
$\Delta^-_M=0.44$---more than half of $V-\vref$---restored to $0.032$ by adding
the shove-then-call probe.  Pins certify only the lines some probe forces;
identifying a multi-line opponent needs a portfolio covering its lines,
which is the \emph{where}-routing problem of \S\ref{sec:scope} in
miniature.
\begin{table}[t]
\centering\small

\setlength{\tabcolsep}{4pt}
\begin{tabular}{@{}lrrr@{}}
\toprule
Evidence design & rows & $R_M$ & $\Delta^-_M$ \\
\midrule
public rows only & $32$ & $-0.0646$ & $0.8794$ \\
$+$ bucketed reveal pins ($K{=}4$) & $88$ & $0.7827$ & $0.0322$ \\
$+$ per-combo reveal pins & --- & $\ge 0.7827$ & $\le 0.0322$ \\
\bottomrule
\end{tabular}
\caption{Residual width of each evidence design against the public twin at
the population limit ($V=0.8149$, $\vref=-0.065$): robust-over-fiber value
$R_M$ (Prop.~\ref{thm:fiber_minimax}) and one-sided width
$\Delta^-_M=V-R_M$. Public rows identify nothing ($R_M=\vref$); bucketed
reveal pins recover all but $4\%$ of the safe-exploitable gap $V-\vref$
(Prop.~\ref{thm:residual_recovery}).}
\label{tab:g1_residual}
\end{table}

\subsection{Discovery without a Public Anomaly}
\label{app:public_null_audit}

We evaluate Corollary~\ref{cor:public_null_audit} on the same unbucketed-river
public twin, using transfer magnitude $0.6$.  The acquisition policy reserves
$20\%$ of the charged hands for a
blueprint public pilot and divides the remainder among a fixed call-down,
bet-then-call, and shove-then-call audit.  Each pure probe is selected by a
private root tag and paired with enough charged blueprint filler to meet
$\rho=0.5$; the opponent is a stationary fixed behavior that cannot condition
on the tag.  An exact treeplex best response verifies every repaired probe
component.  The integer-rounded aggregate remains floor-safe by convexity
because rounding only reduces each pure-probe share; this is an expected
root-mixture guarantee rather than a realized-payoff guarantee.  Within each
cell, public and active confidence families each receive error budget $0.025$,
and both test exclusion of the same $2000$-iteration CFR$^+$ reference.

The reported evaluation uses $100$ seeds disjoint from the development runs.
Table~\ref{tab:public_null_audit}
shows that the public channel detects neither the twin nor the computed CFR
control at any budget.  The active audit detects the twin on every seed by
$N=3000$, while producing no control false positive.  All $800$ cells report
component floor audits above the requirement.  Thus a positive public anomaly
is sufficient for targeted SAD routing but is not necessary for discovery:
this fixed universal three-probe audit detects the constructed deviation,
which is public-null under the blueprint collection policy.
\begin{table}[t]
\centering\small

\begin{tabular*}{\columnwidth}{@{\extracolsep{\fill}}rccc@{}}
\toprule
Charged $N$
& Active twin
& Active ctl.
& Public, each \\
\midrule
$1{,}000$  & $0\ [0,3.7]$       & $0\ [0,3.7]$ & $0\ [0,3.7]$ \\
$2{,}500$  & $0\ [0,3.7]$       & $0\ [0,3.7]$ & $0\ [0,3.7]$ \\
$3{,}000$  & $100\ [96.3,100]$  & $0\ [0,3.7]$ & $0\ [0,3.7]$ \\
$10{,}000$ & $100\ [96.3,100]$  & $0\ [0,3.7]$ & $0\ [0,3.7]$ \\
\bottomrule
\end{tabular*}
\caption{Discovery of the blueprint-public twin by the fixed three-probe
audit on the unbucketed river.  Entries are percentages among $100$ evaluation
seeds with marginal two-sided 95\% Wilson intervals for each budget and
opponent.  Within each cell, public and active interval families split a
$0.05$ error budget.}
\label{tab:public_null_audit}
\end{table}

\subsection{Audit--Refit--Deploy Loop}
\label{app:audit_refit}

We next compose discovery with independent estimation and robust deployment at
$N=10^6$.  Public-only collection uses all $N$ hands under the blueprint.
Fixed reveal reserves $0.2N$ public hands and assigns $0.8N$ to the
prespecified three-probe portfolio.  SAD instead assigns $0.2N$ to public
collection, $0.4N$ to the fixed universal audit, and, only after exclusion of
the reference, $0.4N$ to an independent refit under the same portfolio.
The audit data never enter the response confidence set.  When the audit does
not fire, the remaining budget is charged to blueprint filler and SAD returns
the blueprint.

Table~\ref{tab:sad_audit_refit} reports $30$ matched seeds.  The audit fires on
all fixed-twin seeds (95\% Wilson interval $[88.6,100]\%$) and no control seed
($[0,11.4]\%$).  SAD's twin certificate is $0.655\;(.008)$, and
its population value is $0.760\;(.009)$; on every control seed it returns the
blueprint.  The fixed-reveal comparator uses twice the confirmatory reveal
budget because its route is prespecified; its absolute certificate is
$0.667\;(.006)$, a gain of $0.732\;(.006)$ over $\vref$.  Thus the composed loop establishes
public-null discovery followed by independent refitting and deployment, while
the fixed comparator measures collection after the portfolio is supplied.

All $60$ records pass strict budget, stream-disjointness, coverage, and seed
audits.  Every pure acquisition component clears the required floor by at
least $2{\times}10^{-6}$ in value units, and every response clears it by at
least $1.55{\times}10^{-4}$.  The wall-capped robust solves return feasible
lower certificates; their remaining gaps do not enter the reported certified
values.
\begin{table*}[t]
\centering
\small

\setlength{\tabcolsep}{4pt}
\begin{tabular}{@{}llrrr@{}}
\toprule
Opponent & Method & Certified & Population & Audit fires \\
\midrule
Public twin
& Public only & $-0.065\;(<.000001)$ & $-0.023\;(0)$ & --- \\
& Fixed reveal & $0.667\;(.006)$ & $0.750\;(.006)$ & --- \\
& SAD audit--refit & $0.655\;(.008)$ & $0.760\;(.009)$ & $30/30$ \\
\midrule
Null control
& Public only & $-0.065\;(<.000001)$ & $-0.065\;(0)$ & --- \\
& Fixed reveal & $-0.065\;(<.000001)$ & $-0.065\;(0)$ & --- \\
& SAD audit--refit & $-0.065\;(0)$ & $-0.065\;(0)$ & $0/30$ \\
\bottomrule
\end{tabular}
\caption{Fixed-board public-twin audit--refit--deploy loop at $N=10^6$.
Entries are absolute subgame values; parentheses are two-sided 95\% Student-$t$
CI half-widths over $30$ matched seeds.  SAD uses disjoint
$0.2N/0.4N/0.4N$ public/audit/refit streams.  Fixed reveal assigns $0.2N$
to public data and $0.8N$ to the same prespecified three-probe portfolio
without discovery.  All acquisition and response floor checks pass.}
\label{tab:sad_audit_refit}
\end{table*}

\subsection{The Deployment Grid}
\label{app:fulldeck_grid}

The deployment grid re-runs the body's E-D protocol on the unbucketed river:
five opponent families $\times$ three probe budgets $\times$ ten seeds
($150$ cells), each cell building its blueprint, simulating deploy and
probe batches, and solving the $\Cpub$, per-combo-$\Cid$, and
bucketed-$\Cid$ robust arms with certified brackets and exact floor
checks.  Cells are independent 2-core jobs (${\sim}150$ core-hours total);
laptop reruns of the twin column reproduce the cluster values within seed
95\% Student-$t$ CIs.  Table~\ref{tab:ed_g1_pf} reports every arm.
Figure~\ref{fig:g1_twin} exhibits three properties: the twin's $\Cpub$ arm is flat at $\vref$
with zero variance across all budgets and seeds (the fiber identity,
observed); its bucketed-$\Cid$ arm (which intersects the per-combo pins)
climbs $0.68\to0.79$, entering the band Table~\ref{tab:g1_residual}
predicts---above the bucketed-design fiber value $R_M=0.783$, below the
floor-safe response value $V=0.815$ that brackets the per-combo fiber; and its per-combo
arm needs roughly $100\times$ the budget of the bucketed arm for the same
certificate because the same reveal budget is divided over roughly
$100\times$ more coordinates.
Floors pass the independent audit in all $150$ cells; the control family certifies
$\vref$ and nothing more.  On the over-fold family the passive robust objective
exceeds the probed objectives ($0.203$ vs.\ $0.135$ at $10^7$), but the passive set
has no coverage guarantee and this number is not a statistical certificate even
though it lies below the ex post population value.  That family's portfolio is
the single call-down probe, which never bets and so never elicits an
over-fold: natural blueprint betting is the better reveal there---the
portfolio lesson of App.~\ref{app:twin}, in reverse.
\begin{table*}[t]
\centering\small
\setlength{\tabcolsep}{5pt}

\begin{tabular}{@{}llrrrrr@{}}
\toprule
Opponent & $N$ & $\Cpub$ & passive obj. & $\Cid$ per-combo & $\Cid$ bucketed & Pop. \\
\midrule
Equilibrium (control) & $10^{5}$ & $-0.065\,\pm 0.000$ & $-0.065\,\pm 0.000$ & $-0.065\,\pm 0.000$ & $-0.065\,\pm 0.000$ & $-0.065$ \\
 & $10^{6}$ & $-0.065\,\pm 0.000$ & $-0.065\,\pm 0.000$ & $-0.065\,\pm 0.000$ & $-0.065\,\pm 0.000$ & $-0.065$ \\
 & $10^{7}$ & $-0.065\,\pm 0.000$ & $-0.065\,\pm 0.000$ & $-0.065\,\pm 0.000$ & $-0.065\,\pm 0.000$ & $-0.065$ \\
\addlinespace[1.5pt]
River over-fold & $10^{5}$ & $0.105\,\pm 0.004$ & $0.110\,\pm 0.004$ & $0.105\,\pm 0.004$ & $0.106\,\pm 0.004$ & $0.182$ \\
 & $10^{6}$ & $0.124\,\pm 0.002$ & $0.138\,\pm 0.001$ & $0.126\,\pm 0.001$ & $0.128\,\pm 0.001$ & $0.188$ \\
 & $10^{7}$ & $0.132\,\pm 0.001$ & $0.203\,\pm 0.002$ & $0.132\,\pm 0.002$ & $0.135\,\pm 0.001$ & $0.188$ \\
\addlinespace[1.5pt]
Revealed call & $10^{5}$ & $-0.028\,\pm 0.002$ & $-0.027\,\pm 0.002$ & $-0.028\,\pm 0.002$ & $-0.027\,\pm 0.002$ & $0.015$ \\
 & $10^{6}$ & $-0.005\,\pm 0.001$ & $0.003\,\pm 0.001$ & $-0.004\,\pm 0.001$ & $-0.003\,\pm 0.001$ & $0.018$ \\
 & $10^{7}$ & $0.003\,\pm 0.001$ & $0.012\,\pm 0.001$ & $0.004\,\pm 0.001$ & $0.006\,\pm 0.001$ & $0.019$ \\
\addlinespace[1.5pt]
Sampled (hash-seeded) & $10^{5}$ & $0.060\,\pm 0.001$ & $0.063\,\pm 0.001$ & $0.060\,\pm 0.001$ & $0.061\,\pm 0.001$ & $0.103$ \\
 & $10^{6}$ & $0.070\,\pm 0.001$ & $0.081\,\pm 0.001$ & $0.073\,\pm 0.001$ & $0.074\,\pm 0.001$ & $0.116$ \\
 & $10^{7}$ & $0.077\,\pm 0.001$ & $0.099\,\pm 0.001$ & $0.083\,\pm 0.001$ & $0.086\,\pm 0.000$ & $0.114$ \\
\addlinespace[1.5pt]
\textbf{Public twin} & $10^{5}$ & $-0.065\,\pm 0.000$ & $0.300\,\pm 0.004$ & $-0.065\,\pm 0.000$ & $0.684\,\pm 0.007$ & $0.793$ \\
 & $10^{6}$ & $-0.065\,\pm 0.000$ & $0.336\,\pm 0.001$ & $0.279\,\pm 0.001$ & $0.756\,\pm 0.004$ & $0.798$ \\
 & $10^{7}$ & $-0.065\,\pm 0.000$ & $0.373\,\pm 0.001$ & $0.658\,\pm 0.000$ & $0.794\,\pm 0.001$ & $0.809$ \\
\bottomrule
\end{tabular}
\caption{Full E-D grid at the unbucketed river: robust objective of each
evidence arm (mean $\pm$ two-sided 95\% Student-$t$ CI half-width, $10$
seeds) and population value of the routed arm.  Public and active objectives
are statistical certificates on coverage; the passive column is an
uncertified diagnostic because its showdown-conditioned set need not cover.
All $150$ cells pass the independent floor audit.
On the public twin, $\Cpub$ and per-combo $\Cid$ are numerically identical
until the pins gain statistical weight; bucketed pins certify from
$N{=}10^5$.}
\label{tab:ed_g1_pf}
\end{table*}

\subsection{A Population of Sampled Opponents}
\label{app:zoo}

This subsection measures public and reveal certification across a broader opponent
population.  Two classes are swept: \emph{sampled} opponents (hash-seeded spot-level
fold/call tendencies with per-hand texture, at three strengths, $16$
members each) and $8$ \emph{undertrained} CFR$^+$ opponents (stopped at
$1$--$50$ iterations; past ${\sim}20$ iterations the exploitable gap is
already below $0.05$---undertrained self-play hardens quickly at this
game)---$56$ cells in all, of which $54$ clear the value threshold and
solved-fiber filter and enter Table~\ref{tab:g1_zoo}.  For each opponent one cell computes the exact unsafe ceiling, the
safe oracle $V$, and the population fiber values $R_{\mathrm{pub}}$ and
$R_{\mathrm{grp}}$ (certified solves, as in App.~\ref{app:twin}).
Table~\ref{tab:g1_zoo} summarizes: a typical sampled opponent's safe-exploitable gap
is about three-quarters certifiable from public aggregates alone, and
${\sim}90\%$ with bucketed reveal pins, consistent with the b2-scale
genericity sweep (App.~\ref{app:lowerbound}).  The twin ($0\%\to96\%$)
provides the low-public endpoint.  Floors verify exactly in every cell.
\begin{table}[t]
\centering\small
\setlength{\tabcolsep}{2.5pt}

\begin{tabular}{@{}lrrrr@{}}
\toprule
Class & $n$ & $V{-}\vref$ med.\ [range] & public & reveal \\
\midrule
sampled, strength $0.3$ & $16$ & $0.08$ [0.06, 0.11] & $0.72$ & $0.91$ \\
sampled, strength $0.5$ & $16$ & $0.15$ [0.10, 0.23] & $0.82$ & $0.95$ \\
sampled, strength $0.8$ & $15$ & $0.23$ [0.17, 0.42] & $0.85$ & $0.92$ \\
undertrained CFR$^+$ & $7$ & $0.48$ [0.06, 1.96] & $0.69$ & $0.87$ \\
\midrule
\textbf{all} & $54$ & $0.17$ [0.06, 1.96] & $0.73$ & $0.91$ \\
\bottomrule
\end{tabular}
\caption{A population of non-constructed opponents at the full unbucketed
river: hash-seeded spot-level fold/call perturbations at three strengths and
undertrained CFR$^+$ opponents.  Per class: opponents with safe-exploitable
value $V-\vref>0.02$, its median [min, max], and the median \emph{share} of
that value certified by the population public fiber ($R_{\mathrm{pub}}$) and
by $3$-probe bucketed reveal pins ($R_{\mathrm{grp}}$).  A typical sampled
opponent is about three-quarters public-certifiable and ${\sim}90\%$
reveal-certifiable (means $0.73$/$0.91$, matching the medians); the public
twin ($0\%$ public, $96\%$ reveal) is the constructed corner of this
distribution, not its typical member.  Three degenerate grouped fibers
enter as monotone brackets and are excluded from the reveal statistics.}
\label{tab:g1_zoo}
\end{table}

\subsection{Diagnostic under Drift}
\label{app:drift}

Stationarity (A1) is required for the paper's pooled statistical certificates; this
subsection is an empirical drift evaluation rather than an extension of those guarantees.  The
floor is unaffected because Theorem~\ref{thm:safety} quantifies over every opponent
plan on every hand.

In Table~\ref{tab:g1_drift}, the twin abandons its leak at the stream's midpoint.
The pooled robust objective continues to rise while its value against the current
post-switch opponent falls: their gap grows from $0.38$ to $0.51$ as $N$ increases.
A trailing-window fit has gap approximately zero at each tested budget, whereas
pooled $\Cpub$ cannot detect the switch because both phases are public-null.  These
values illustrate the cost of violating (A1); they are not confidence guarantees for
the current opponent.  The independent floor audit passes in all $90$ cells.
\begin{table}[t]
\centering\small
\setlength{\tabcolsep}{3.5pt}

\begin{tabular}{@{}llrrrr@{}}
\toprule
$N$ & evidence & objective & Current & gap & viol \\
\midrule
$10^{5}$ & pooled & $0.083$ & $-0.296$ & $+0.379$ & $0.080$ \\
 & windowed & $-0.065$ & $-0.065$ & $-0.000$ & $-0.001$ \\
 & $\Cpub$ pooled & $-0.065$ & $-0.065$ & $-0.000$ & $-0.000$ \\
\addlinespace[1.5pt]
$10^{6}$ & pooled & $0.134$ & $-0.328$ & $+0.462$ & $0.246$ \\
 & windowed & $-0.065$ & $-0.065$ & $-0.000$ & $-0.000$ \\
 & $\Cpub$ pooled & $-0.065$ & $-0.065$ & $-0.000$ & $-0.000$ \\
\addlinespace[1.5pt]
$10^{7}$ & pooled & $0.176$ & $-0.338$ & $+0.514$ & $0.454$ \\
 & windowed & $-0.065$ & $-0.064$ & $-0.000$ & $-0.000$ \\
 & $\Cpub$ pooled & $-0.065$ & $-0.065$ & $-0.000$ & $-0.000$ \\
\bottomrule
\end{tabular}
\caption{Drift evaluation: the public twin reverts to equilibrium at the
stream's midpoint ($10$ seeds; 95\% Student-$t$ CI half-widths for displayed
robust-objective and current-value means are at most $0.018$).  \emph{Current}\ is
realized value against the \emph{current} (post-switch) opponent;
\emph{gap} $=$ robust objective $-$ current value; \emph{viol} is the
current opponent's constraint violation in the arm's set.  The pooled
gap \emph{grows} with budget while the floor holds in
every cell; the trailing-window gap remains ${\approx}0$;
pooled $\Cpub$ never detects the switch (the twin's defining property).
When the leak instead \emph{grows} mid-stream (over-fold $0.15\to0.5$,
$N{=}10^7$), pooling errs the other way---objective $0.180$ vs current value
$0.393$---and the window is again tight ($0.401$ vs $0.438$).}
\label{tab:g1_drift}
\end{table}

\section{Game Instances and Experimental Setup}
\label{app:setup}

\subsection{Turn--River Abstraction}
The turn--river hold'em endgame is evaluated at two card-abstraction granularities, which
the body refers to as \emph{coarse} and \emph{fine}.  The \emph{coarse} instance
(two public river buckets; $5{,}221$ sequences / $2{,}430$
opponent information sets; $\vref=-0.043$) is small enough for repeated population and
finite-sample sweeps.  The \emph{fine} instance (a canonical turn
board, four public river buckets; $9{,}721$ sequences / $4{,}590$ opponent information
sets; $\vref=-0.053$) serves as a higher-granularity confirmation.
All payoffs are in chips with the ante normalized to $1$ (effective stack $10$), so
every reported value is a multiple of the ante; this convention also covers the
unbucketed river of Appendix~\ref{app:fulldeck}.
In both instances the blueprint $x_{\mathrm{bp}}$ is a sequence-form Nash
equilibrium computed numerically to tolerance by a single linear program (HiGHS) over
the realization polytopes; its security value is the $\vref$ quoted above.  All
experiments use these blueprints, and the population ladder pins confidence sets to
solver values, so the
coarse/fine results carry no sampling error.  The exact (unbucketed) river and the
eight-bucket abstraction inflate the sequence-form LP well beyond what the repeated
population and finite-sample sweeps allow, which is why we cap the granularity at the
four-bucket fine instance.

\subsection{Opponent Families}
The opponents are exact-Nash player~2 strategies perturbed by a single structural leak,
spanning the identifiability spectrum from purely revealed to deeply censored.  Writing
the betting line so that a turn round precedes the river (marked by a `/' separator):
\begin{itemize}
  \item \textbf{Equilibrium (control).} The unperturbed Nash strategy; a safe response
        must neither lose to it nor infer a spurious gain.
  \item \textbf{River over-fold (uniform).} On river fold-lines, shift weight $w$ onto
        folding for all hands (graded $w\in\{0.05,\dots,0.8\}$); a revealed-adjacent leak
        observable in the public line.
  \item \textbf{River over-fold (strong).} The same river leak but only on strong hands
        (top rank $\ge 9$), weight $0.6$.
  \item \textbf{Revealed call (strong).} On river lines where both fold and call are
        legal, strong hands shift weight $0.6$ onto calling: a purely \emph{revealed}
        leak with no censoring, the high-value/high-reveal-cost corner of the frontier.
  \item \textbf{Turn over-fold (deep).} On turn fold-lines (before the river), shift
        weight $w$ onto folding (graded $w\in\{0.3,0.5,0.7\}$): the deeply
        \emph{censored} leak, since the fold ends the hand before the type is revealed.
\end{itemize}
Different experiments use different graded members, so oracle values are comparable
within a table but not across tables: the population ladders
(Tables~\ref{tab:ladder},~\ref{tab:ladder_fine}) use the uniform river over-fold at
$w=0.5$ and the deep turn over-fold at $w=0.5$, while the deployment and baseline suites
(Tables~\ref{tab:deploy},~\ref{tab:deploy_pf},~\ref{tab:baselines}) use the stronger
members $w=0.8$ (river) and $w=0.7$ (turn) so that finite-sample effects act on leaks
worth exploiting; the strong river over-fold and revealed-call members ($0.6$, strong
hands) are shared.  The remaining graded variants populate the diagnostic suite used for
the signal panel (Appendix~\ref{app:gate_signals}).

\section{Reproducibility}
\label{app:repro}

All bucketed-endgame experiments run on a single workstation (Apple~M3~Max, 14 cores
[10P\,+\,4E], 36~GB~RAM; macOS~15.7).  Orchestration is in Python~3.14 (compatible with
$\ge 3.11$); the heavy computation---exact Nash blueprints, the robust safe-response LPs
over $\Cpub$ and $\Cid$, and the LP-free floor verifier---runs in a Rust core
(edition~2021, \texttt{rustc}~$\ge 1.83$) built on HiGHS~2.2.0 and \texttt{ndarray}~0.17
and exposed to Python through PyO3~0.29, with \texttt{numpy}~$\ge 2.1$,
\texttt{scipy}~$\ge 1.14$, and \texttt{pebble}~$\ge 5$ for analysis and per-cell
parallelism.  The unbucketed-river experiments of App.~\ref{app:fulldeck} use the same Rust core as a
standalone cell binary, run as independent 2-core jobs on an AMD EPYC (Rome) cluster
($446$ cells $=150$ deployment $+150$ passive $+56$ zoo $+90$ drift; ${\sim}400$
core-hours; one JSON checkpoint per cell); the remaining unbucketed-river benches are
single-core workstation runs, and workstation reruns of the twin column reproduce the
cluster values within their two-sided 95\% Student-$t$ CIs. Finite-sample runs use
consecutive integer seeds from a base of $2026$, matched across the compared methods.
The primary $N=10^6$ grid uses seeds $2026$--$2035$ and contains $164$ cells:
$4$ opponents $\times$ [$4$ empirical arms $\times$ $10$ seeds $+$ one population
oracle].  The complementary routing grid contains $488$ cells across three reveal budgets;
the two lower-budget matched-acquisition runs add $240$ empirical cells, for $404$
unique cells across the primary and budget-sweep grids;
the reported-protocol capacity-only ablation adds $30$ cells;
the public/reveal allocation sensitivity adds $80$ reported-protocol cells;
the fixed non-fold acquisition comparator adds $40$ matched-seed cells;
the tighter-floor controller comparison adds $60$ matched-seed cells;
the replicated method table uses seeds $2026$--$2050$, and the unbucketed-river grid
uses seeds $1$--$10$.  The public-null audit fixes budgets
$N\in\{1000,2500,3000,10000\}$ and uses evaluation seeds $1001$--$1100$,
matched between the blueprint-public twin and computed CFR control, for $800$
cells.  Its detection and false-positive rates use two-sided 95\% Wilson
intervals separately at each budget and opponent.  The four-target Leduc
joint-routing benchmark is deterministic; its reported interval is the
binary64 Frank--Wolfe optimization bracket under the stated $10^{-7}$
feasibility tolerance, with outward-rounded, dual-feasible weighted-reach
upper bounds for each linear subproblem; it is not a sampling interval.  The population ladder, the necessity
construction, and the population-fiber residuals are deterministic exact computations and
carry no seed.  Finite-seed means use two-sided 95\% Student-$t$ CIs; exact-zero
controls have zero-width intervals.  Paired method comparisons use exact Wilcoxon
signed-rank tests with Holm correction across each stated comparison family.

The reveal-capacity audit fixes the $60$ largest non-fold mass deviations of
the deep-turn over-fold opponent and solves all target--budget cells at
$\rho\in\{0.1,0.5\}$.  Twelve targets require a constrained revealing suffix,
giving $24$ additional fixed-action solves.  Across the $120$ target--budget
cells, the largest primal--certificate width is $6.33{\times}10^{-7}$ and the
largest absolute negative floor margin is $2.10{\times}10^{-8}$, both within
the stated $10^{-6}$ audit tolerance.  Its calibration uses
$N\in\{10^5,3{\times}10^5,10^6\}$ and seeds $9101$--$9105$, for $360$
cells and $168$ million simulated hands.  The audit--refit--deploy loop uses
seeds $7001$--$7030$ for the public twin and null control, producing $60$
matched records at $N=10^6$.  It assigns fractions $0.2/0.4/0.4$ to
public/audit/refit data, gives the audit and response sets separate
$\delta=0.05$ budgets, caps each robust solve at $60$ seconds, and requires
each acquisition component to clear the floor by $10^{-6}$ in value units.

\codeavailabilityblock
The empirical-audit supplement contains $610$ sampled-controller records---$570$ with
floor-safe collection and $40$ for the deliberately unconstrained fixed
non-fold comparator---all $800$ public-null audit records, all $446$
unbucketed-river checkpoints, $60$ audit--refit--deploy records, the $120$-cell
capacity audit, the $360$-cell hand-level calibration, and four deterministic
population-oracle records, with configuration fingerprints,
audit flags, and a standard-library script that regenerates
the reported tables, figure source, uncertainty summaries, and paired tests.
No external dataset is used.  Each deployment cell records both sampling batches and
seeds, the frozen route, certified and realized value, confidence-set feasibility, and
the independent floor-audit result.  Public screening uses simultaneous Hoeffding
bounds; response constraints use empirical-Bernstein intervals under the shared
$\delta=0.1$ budget.

\paragraph{Hyperparameters.}
The safety budget $\rho\in\{0,0.1,0.5\}$ is swept in the capacity studies.  The
primary controller uses total budget $N=10^6$; its budget study uses
$N\in\{10^5,3\!\times\!10^5,10^6\}$, while the complementary reveal-only diagnostic uses
$N\in\{10^3,10^4,10^5\}$.  The public fraction $0.2$, confidence level $\delta=0.1$,
screening level $\delta_r=0.1$, target threshold $10^{-6}$, and probe tie-break
weight $B=10^6$ are fixed across opponents and budgets.  The latter two are numerical
safeguards rather than tuned quantities; every reported comparison uses the same
confidence allocation.  The allocation sensitivity additionally evaluates public
fractions $0.5$ and $0.8$ at $N=10^6$.
Settings follow their statistical or computational roles rather than reported outcomes:
safety and hand budgets are evaluation axes, while the primary $0.2$ public fraction
preserves an independent screen and assigns most hands to reveal acquisition.  The
public-null audit uses a $0.2/0.8$ public/audit split, three equally weighted fixed probes,
and error budget $0.025$ per cell for each public and active confidence family.  The
public twin uses transfer $0.6$ on the fixed board and small-bet river rules, with a
$2000$-iteration CFR$^+$ reference.  The audit--refit loop uses the same fixed
call-down, bet-then-call, and shove-then-call portfolio, with no audit sample
reused for response fitting.  Joint routing uses a $0.2$ audit share, uniform target
weights and continuation rates, and one-thread HiGHS simplex with presolve.

\end{document}